\documentclass[onefignum,onetabnum]{siamart251216}

\usepackage[numbers,sort]{natbib}

\usepackage{graphicx}
\usepackage[shortlabels]{enumitem}

\usepackage{amsmath,amsfonts,amssymb,stmaryrd}

\usepackage{algorithm}
\usepackage{algpseudocodex}      
\newcommand{\algoitem}{$\cdot$\,} 
\algrenewcommand\textproc{\texttt}   
\algnewcommand{\algorithmicand}{\textbf{and} }
\algnewcommand{\algorithmicor}{\textbf{or} }
\algnewcommand{\OR}{\algorithmicor}
\algnewcommand{\AND}{\algorithmicand}
\algnewcommand{\InlineIf}[2]{
   \algorithmicif\ #1\ \algorithmicthen\ #2}
\algnewcommand{\InlineElse}[1]{
  \algorithmicelse\ #1}
\algnewcommand{\InlineFor}[2]{\algorithmicfor\ #1\ \algorithmicdo\ #2} 
\algrenewcommand\Call[2]{\nameref{#1}\ifthenelse{\equal{#2}{}}{}{\ensuremath{(#2)}}}%
\newcommand{\algoName}[1]{Algorithm \nameref{#1}}
\makeatletter
\newcommand{\algoCaptionLabel}[2]{
     \caption[\textproc{#1}]{\textproc{#1}\ifthenelse{\equal{#2}{}}{}{$(#2)$} }%
     \NR@gettitle{\textproc{#1}}
      \label{algo:#1}
     }%
\makeatother

\usepackage{zref-clever}
\hypersetup{colorlinks=true, allcolors=blue}
\zcsetup{cap, abbrev, nameinlink=false}

\theoremstyle{plain}
\theoremheaderfont{\normalfont\itshape}
\theorembodyfont{\normalfont}
\theoremseparator{.}
\theoremsymbol{\proofbox}
\newtheorem{remark}[theorem]{Remark}
\theoremsymbol{}

\zcRefTypeSetup{lemma}{
  Name-sg = Lemma,
  Name-pl = Lemmas,
  name-sg = Lemma,
  name-pl = Lemmas
}
\zcRefTypeSetup{theorem}{
  Name-sg = Theorem,
  Name-pl = Theorems,
  name-sg = Theorem,
  name-pl = Theorems
}
\zcRefTypeSetup{proposition}{
  Name-sg = Proposition,
  Name-pl = Propositions,
  name-sg = Proposition,
  name-pl = Propositions
}
\zcRefTypeSetup{definition}{
  Name-sg = Definition,
  Name-pl = Definitions,
  name-sg = Definition,
  name-pl = Definitions
}
\zcRefTypeSetup{corollary}{
  Name-sg = Corollary,
  Name-pl = Corollaries,
  name-sg = Corollary,
  name-pl = Corollaries
}
\zcRefTypeSetup{remark}{
  Name-sg = Remark,
  Name-pl = Remarks,
  name-sg = Remark,
  name-pl = Remarks
}
\zcRefTypeSetup{example}{
  Name-sg = Example,
  Name-pl = Examples,
  name-sg = Example,
  name-pl = Examples
}
\zcRefTypeSetup{notation}{
  Name-sg = Notation,
  Name-pl = Notations,
  name-sg = Notation,
  name-pl = Notations
}

\zcRefTypeSetup{problem}{
 Name-sg = Problem,
 Name-pl = Problems,
 name-sg = problem,
 name-pl = problems
}

\zcRefTypeSetup{ALG@line}{
 Name-sg = Line,
 Name-pl = Lines,
 name-sg = line,
 name-pl = lines
}

\makeatletter
\AddToHook{env/algorithmic/begin}{\def\@currentcounter{ALG@line}}
\makeatother
\AddToHook{env/proposition/begin}{\zcsetup{countertype={theorem=proposition}}}
\AddToHook{env/definition/begin}{\zcsetup{countertype={theorem=definition}}}
\AddToHook{env/remark/begin}{\zcsetup{countertype={theorem=remark}}}
\AddToHook{env/lemma/begin}{\zcsetup{countertype={theorem=lemma}}}
\AddToHook{env/corollary/begin}{\zcsetup{countertype={theorem=corollary}}}

\usepackage[skins]{tcolorbox} 
\tcbuselibrary{theorems}
\usetikzlibrary{shapes}

\newtcbtheorem
{problem}{Problem}{%
  label is zlabel,%
  label type=problem,%
  boxsep=0cm,%
  left=0.2cm, top=0.2cm, right=0.2cm, bottom=0.10cm,%
  enhanced,%
  attach boxed title to top left={xshift=0.2cm, yshift=-2mm},%
  colbacktitle=white, coltitle=black, fonttitle=\itshape, %
  colback=white, %
  boxed title style={ boxrule=0.04em, colframe=black, boxsep=0cm },%
}{pbm}

\newcommand{\K}{\mathbb{K}} 
\newcommand{\N}{\mathbb{N}} 
\newcommand{\Z}{\mathbb{Z}} 

\newcommand{\M}{\mathcal{M}} 
\newcommand{\rdeg}{\mathsf{rdeg}} 
\newcommand{\cdeg}{\mathsf{cdeg}} 
\newcommand{\rlm}{\mathsf{rlm}} 
\newcommand{\clm}{\mathsf{clm}} 
\newcommand{\diag}[1]{\operatorname{diag}(#1)} 

\makeatletter
\newcommand*{\rem}{%
  \nonscript\mskip-\medmuskip\mkern5mu%
  \mathbin{\operator@font rem}\penalty900\mkern5mu%
  \nonscript\mskip-\medmuskip
}
\makeatother

\makeatletter
\newcommand*{\quo}{%
  \nonscript\mskip-\medmuskip\mkern5mu%
  \mathbin{\operator@font quo}\penalty900\mkern5mu%
  \nonscript\mskip-\medmuskip
}
\makeatother

\newcommand{\D}{\mathbb{D}} 
\newcommand{\V}{\mathbb{V}} 

\newcommand{\C}{\mathbb{C}} 
 
\renewcommand{\L}{\mathbb{L}} 
\newcommand{\U}{\mathbb{U}} 
\newcommand{\ident}{\mathbb{I}} 

\newcommand{\Zmat}{\mathbb{S}} 

\newcommand{\mA}{\mathsf{A}} 
\newcommand{\mB}{\mathsf{B}} 
\newcommand{\mY}{\mathsf{Y}} 
\newcommand{\mZ}{\mathsf{Z}} 
\newcommand{\mP}{\mathsf{P}} 
\newcommand{\mQ}{\mathsf{Q}} 
\newcommand{\mF}{\mathsf{F}} 
\newcommand{\mG}{\mathsf{G}} 
\newcommand{\mE}{\mathsf{E}} 
\newcommand{\mW}{\mathsf{W}} 
\newcommand{\mH}{\mathsf{H}} 
\newcommand{\mM}{\mathsf{M}} 
\newcommand{\mS}{\mathsf{S}} 
\newcommand{\mR}{\mathsf{R}} 
\newcommand{\mN}{\mathsf{N}} 
\newcommand{\mL}{\mathsf{L}} 
\newcommand{\mK}{\mathsf{K}} 
\newcommand{\mX}{\mathsf{X}} 

\newcommand{\va}{\mathsf{a}} 
\newcommand{\vc}{\mathsf{c}} 
\newcommand{\vd}{\mathsf{d}} 
\newcommand{\vx}{\mathsf{x}} 
\newcommand{\vy}{\mathsf{y}} 
\newcommand{\vz}{\mathsf{z}} 
\newcommand{\vu}{\mathsf{u}} 
\newcommand{\vv}{\mathsf{v}} 
\newcommand{\vw}{\mathsf{w}} 
\newcommand{\vg}{\mathsf{g}} 
\newcommand{\vh}{\mathsf{h}} 
\newcommand{\vt}{\mathsf{t}} 
\newcommand{\ve}{\mathsf{e}} 
\newcommand{\vq}{\mathsf{q}} 
\newcommand{\vp}{\mathsf{p}} 
\newcommand{\vs}{\mathsf{s}} 
\newcommand{\vr}{\mathsf{r}} 

\usepackage{xcolor}

\newcommand{\field}{\K}        
\newcommand{\xring}{\K[x]}     
\newcommand{\matspace}[2]{\field^{#1 \times #2}} 
\newcommand{\xmatspace}[2]{\xring^{#1 \times #2}} 
\newcommand\scalemath[2]{\scalebox{#1}{\mbox{\ensuremath{\displaystyle #2}}}}

\newcommand{\bigO}[1]{O(#1)} 
\newcommand{\bigOPar}[1]{O\!\left(#1\right)} 
\newcommand{\softO}[1]{\tilde{O}(#1)} 
\newcommand{\expmm}{\omega} 
\newcommand{\timepm}[1]{\mathcal{M}(#1)} 
\newcommand{\timepmPar}[1]{\mathcal{M}\!\left(#1\right)} 
\newcommand{\polrev}[2]{\operatorname{rev}_{#1}(#2)} 
\newcommand{\vecrev}[1]{\operatorname{vrev}(#1)} 
\newcommand{\crev}[2]{\operatorname{crev}_{#1}(#2)} 
\newcommand{\rrev}[2]{\operatorname{rrev}_{#1}(#2)} 

\newcommand{\dispRk}{\alpha} 
\newcommand{\dispOp}{\phi} 
\newcommand{\dispSyl}[1]{\nabla[#1]} 
\newcommand{\dispSte}[1]{\Delta[#1]} 

\newcommand{\shift}{\vs} 
\newcommand{\shiftt}{\vt} 
\newcommand{\vsol}[1][v]{\vp_{#1}} 
\newcommand{\ssol}[1][\vv]{p_{#1}} 
\newcommand{\polmod}{M} 
\newcommand{\mpol}{f} 
\newcommand{\lpol}{\mathfrak{L}} 
\newcommand{\invmpol}{\bar{f}} 

\newcommand{\trsp}{\mathsf{T}} 

\newcommand{\sumti}{\rho} 

\newcommand{\rl}[2]{\left[\begin{smallmatrix} #1 \\ #2 \end{smallmatrix}\right]} 

\newcommand{\relmod}{\mathcal{R}} 
\newcommand{\relsol}{\mathcal{S}} 
\newcommand{\PMB}{\textproc{PM-Basis}}   
\newcommand{\SAM}{\textproc{ShiftAroundMaxAppBasis}}   
\newcommand{\appmod}[2]{\mathcal{A}_{#1}(#2)} 

\headers{Displacement structure: deterministic systems and nullspaces}{S. Khichane and V. Neiger}

\title{Matrices with displacement structure: a deterministic approach for
linear systems and nullspace bases\thanks{Submitted, 2026.}}

\author{Sara Khichane\thanks{Sorbonne Université, CNRS, LIP6, F-75005 Paris,
  France\\ (\email{sara.khichane@lip6.fr}, \email{vincent.neiger@lip6.fr}).}
\and Vincent Neiger\footnotemark[2]}

\begin{document}

\maketitle

\begin{abstract}
  For matrices with a displacement structure, the fastest known algorithms for
  solving linear systems and computing nullspace bases are randomized, with
  respective costs of \(\tilde{O}(\alpha^{\omega-1} n)\) and
  \(\tilde{O}(\alpha^2 n)\) operations in the base field. Here \(n \times n\)
  is the matrix size, \(\alpha\) is its displacement rank, \(\omega > 2\) is a
  feasible exponent for matrix multiplication, and \(\tilde{O}(\cdot)\) hides
  logarithmic factors.
  For a long time, such structured solvers were also the fastest known method
  to perform Hermite-Pad\'e or simultaneous Pad\'e approximations, and their
  interpolation variants, until recently when these problems received
  efficient and deterministic solutions via polynomial matrix algorithms. This
  work proposes to take the converse direction, by showing
  that these recent polynomial matrix tools can be used to improve structured
  matrix computations.
  We present a new approach which, for the Toeplitz-like, Vandermonde-like, and
  Cauchy-like structures, solves linear systems and computes nullspace bases in
  \(\tilde{O}(\alpha^{\omega-1}(m+n))\) operations for an arbitrary \(m \times
  n\) matrix, deterministically and with no assumption on the base field. This
  bound matches the best known randomized one for linear systems and improves
  upon it for nullspaces. The approach reformulates the problem as a modular
  equation on univariate polynomials, and describes all its solutions in three
  steps: computing a basis of solutions to a vector M-Pad\'e approximation
  problem, solving a polynomial matrix equation to get rid of unwanted
  unknowns, and finding the sought unknowns via simultaneous M-Pad\'e
  approximation.
\end{abstract}

\begin{keywords}
  structured matrices,
  displacement rank,
  linear systems,
  nullspace,
  polynomial matrices,
  m-Pad\'e approximation,
  deterministic algorithms,
  complexity
\end{keywords}

\begin{MSCcodes}
  68W30, 15B05, 68Q25
\end{MSCcodes}


\section{Introduction}
\label{sec:intro}

In computational linear algebra, besides the quest for general efficient
algorithms for dense matrices, lies the fundamental question of designing
faster algorithms tailored to matrices with special properties that are often
encountered in concrete situations. This includes broad families such as sparse
matrices with a large number of zero entries, low-rank matrices which factor
into matrices with few columns or rows, rank-structured and quasiseparable
matrices with low-rank off-diagonal blocks, and matrices with low displacement
rank. In this article, we focus on the latter family.

\paragraph{The displacement structure}

This notion was introduced in the seminal article of Kailath, Kung, and Morf \cite{kailathkungmorf1979},
to provide a framework for characterizations and computations with matrices
that are ``close'' to a Toeplitz matrix. Later generalizations also encompass
Hankel, Vandermonde, and Cauchy matrices \cite[Table\,1.1]{pan2001}; see also
\zcref[S]{sec:matvec_polops:prelim,sec:matvec_polops:special} for definitions and
notations. Beyond these specific examples, this type of structure covers all
\(m \times n\) matrices whose image through a certain operator has a rank
\(\dispRk\) which is ``small'' compared to the maximal possible rank
\(\min(m,n)\). For example, for the Toeplitz-like operator, the smaller
\(\dispRk\) is, the closer the matrix is to an actual Toeplitz matrix, which
has \(\dispRk \le 2\).
More precisely, denoting the base field by \(\field\), the
so-called displacement operator is an invertible linear operator \(\dispOp :
\matspace{m}{n} \to \matspace{m}{n}\), and for a given matrix \(\mA \in
\matspace{m}{n}\), its displacement rank is the rank \(\dispRk\) of
\(\dispOp(\mA)\). Any pair of matrices $\mG \in \matspace{m}{\dispRk}$ and $\mH
\in \matspace{n}{\dispRk}$ such that \(\dispOp(\mA) = \mG \mH^\trsp\) is called
a \(\dispOp\)-generator for \(\mA\): since \(\dispOp\) is invertible, one can
use \((\mG,\mH)\) as a data structure to represent \(\mA\).

This storage uses only \(\dispRk(m+n)\) coefficients from the field, instead of
\(mn\), and the task is then to design efficient matrix algorithms that operate
on this compact data structure rather than on the usual dense representation,
and thereby take advantage of the low displacement rank. A central tool is
matrix-vector products, for which efficient algorithms are known by relying on
fast univariate polynomial computations, and which are then exploited in
algorithms for more general problems on structured matrices such as computing
matrix-matrix products, ranks, nullspaces, determinants, linear system
solutions, etc. Indeed, it is well-known that for classical displacement structures,
matrix-vector products are performed efficiently via operations in
\(\xring\) \cite[Tbl.\,1.2]{pan2001}, such as truncated polynomial
multiplication for Toeplitz matrices or multipoint evaluation for Vandermonde
matrices. This actually extends to general forms of displacement structures
where fast matrix-vector products are realized through operations on polynomial
vectors in \(\xring^{\dispRk}\) or polynomial matrices in
\(\xmatspace{\dispRk}{\beta}\), as for example in \cite{GohbergOlshevsky1994}
and \cite[Sec.\,5]{structured2017}.

The starting point of such univariate polynomial interpretations of structured
matrix-vector products is Gohberg-Semencul-type inversion formulas
\cite{GohbergSemencul1972,kailathkungmorf1979,GohbergOlshevsky1994}, including
for example the so-called ``\(\Sigma L U\)'' formula for Toeplitz-like
matrices. These formulas express the matrix \(\mA\) from its generators
\((\mG,\mH)\) through a combination of elementary structured matrices such as
Toeplitz or Cauchy matrices, which themselves lead to interpretations as
univariate polynomial operations. For a general study of such inversion
formulas, one may refer to \cite[Sec.\,4.3 and\,4.4]{pan2001},
\cite{PanWang2003}, and \cite[Sec.\,3]{structured2017}; we will discuss this
again in \zcref[S]{sec:matvec_polops}. This approach leads to matrix-vector
products in a complexity that is quasilinear in the size \(\dispRk(m+n)\) of
the generators.

\paragraph{Problem and context}

In this article, we focus on solving linear systems and computing nullspaces
(and hence the rank), in the case of matrices that are represented through
displacement generators. Hereafter, \(\field\) is an arbitrary field, which is
assumed to be effective: one has procedures to add, subtract, multiply, and
invert field elements, and to test whether a given element is zero.

This article studies complexity aspects of \zcref[S]{pbm:linsys,pbm:kernel}. A
major question is whether one can solve them by algorithms that are both
efficient and deterministic: we will answer it positively. Another objective is
to better understand what recent fast polynomial matrix algorithms can achieve
on such structured matrix problems, given the well-known links between
structured linear algebra and variants of Pad\'e approximation.

\begin{problem}{Structured linear system}{linsys}
  \emph{Parameter:}
  an invertible displacement operator \(\dispOp : \matspace{m}{n} \to \matspace{m}{n}\).

  \emph{Input:}
  matrices $\mG \in \matspace{m}{\dispRk}$ and $\mH \in \matspace{n}{\dispRk}$, a vector $\vv \in \K^{m \times 1}$.

  \emph{Output:} a vector $\vu \in \matspace{n}{1}$ such that $\mA \vu = \vv$
  if one exists, otherwise $\vu = \emptyset$, where $\mA \in \matspace{m}{n}$
  is the matrix defined by $\dispOp(\mA) = \mG \mH^{\trsp}$.
\end{problem}

\begin{problem}{Structured nullspace}{kernel}
  \emph{Parameter:}
  an invertible displacement operator \(\dispOp : \matspace{m}{n} \to \matspace{m}{n}\).

  \emph{Input:}
  matrices $\mG \in \matspace{m}{\dispRk}$ and $\mH \in \matspace{n}{\dispRk}$.

  \emph{Output:} the rank of the matrix $\mA \in \matspace{m}{n}$ defined by
  $\dispOp(\mA) = \mG \mH^{\trsp}$, and displacement generators (with respect
  to a suitable displacement operator) for a matrix whose columns form a basis
  of the nullspace \(\{\vz \in \matspace{n}{1} \mid \mA \vz = 0\}\) of \(\mA\).
\end{problem}

In complexity bounds, each of the above-listed basic field operations is
counted at unit cost, and we measure the algebraic cost of an algorithm as the
number of these basic field operations that are performed during the run of the
algorithm on a given input. To give asymptotic upper bounds, we use the
\(\bigO{\cdot}\) notation. Occasionally, for readability in explanatory texts,
we give simplified bounds with the notation \(\softO{\cdot}\) which hides
logarithmic factors. We discuss algorithms that are based on polynomial
multiplication and matrix multiplication \cite{ModernComputerAlgebra,BCS97}. We
let \(\expmm\) be a feasible exponent for matrix multiplication with \(2 <
\expmm \le 3\) \cite{AlmanDuanWilliamsXuXuZhou2025}, and \(d \mapsto
\timepm{d}\) be a time function for polynomial multiplication in degree less
than \(d\), which can be taken quasilinear \cite{CantorKaltofen1991}. More
details are presented in \zcref[S]{sec:notation}.

\paragraph{Divide and conquer structured solvers}

Focusing on square Toeplitz-like matrices, thus with \(m=n\), efficient
algorithms for \zcref[S]{pbm:linsys,pbm:kernel} from the late 1970s used
$\bigO{\dispRk n^2}$ operations in \(\field\) \cite{kailathkungmorf1979}. In
1980, Morf \cite{morf80} and Bitmead and Anderson \cite{bitmeadanderson1980}
improved this to $\bigO{\dispRk^2 \timepm{n} \log(n)}$, through a Strassen-like
divide and conquer algorithm based on Schur complements \cite{Strassen1969}
\cite[Sec.\,5.2 and\,5.3]{pan2001}.
The latter results require the input matrix to have generic rank profile,
meaning that its leading principal minors are nonzero. This requirement can be
lifted by pre- and post-multiplication by random matrices, and in 1994,
Kaltofen made this strategy efficient thanks to random structured multipliers
\cite{kaltofen1994}, obtaining a Las Vegas randomized algorithm in
$\bigO{\dispRk^2 \timepm{n} \log(n)}$.
In the 1990s, the Vandermonde-like and Cauchy-like structures were also
handled in the same cost bound \cite{pan1990, GohbergOlshevsky1994,
  cardinal1999, panzheng2000, panzhengchenprovidence1999,
OlshevskyShokrollahi2000}, either by a modification of this approach for the
Toeplitz-like case, or thanks to a direct reduction to it. A general
presentation of this approach can be found in \cite[Sec.\,5.6
and\,5.7]{pan2001}.

When \(\dispRk \in \Theta(n)\), the above cost bounds become cubic in \(n\),
similarly to dense linear algebra when not exploiting fast matrix
multiplication. In 2008, Bostan, Jeannerod, and Schost brought fast dense
matrix multiplication into the above divide and conquer scheme
\cite{BostanJeannerodSchost2008}, reaching the cost $\bigO{\dispRk^{\expmm-1}
\timepm{n} \log^2(n)}$. This result does not cover \zcref[S]{pbm:kernel}; it
focuses on the Toeplitz-like, Vandermonde-like, and Cauchy-like displacement
operators, which are also the ones that are handled in this article. Further
improvements of the exponents in this bound
\(\softO{\dispRk^{\expmm-1} n}\) seem difficult: this is close to the size \(2
\dispRk n\) of the input generators, and for dense matrices with \(\dispRk \in
\Theta(n)\) this is the usual complexity \(\softO{n^\expmm}\) of dense linear
algebra. The result of \cite{BostanJeannerodSchost2008} was later generalized
by the same authors and Mouilleron \cite{structured2017}, with an algorithm for
\zcref[S]{pbm:linsys} that is faster by a logarithmic factor, supports the
rectangular case, and covers a wider array of displacement operators.

All in all, the fastest known algorithms for \zcref[S]{pbm:linsys,pbm:kernel} are
Las Vegas randomized, with a complexity of $\bigO{\dispRk^{\expmm-1}
\timepm{m+n} \log(m+n)}$ operations in \(\field\) for \zcref[S]{pbm:linsys}
\cite{structured2017} for a large family of displacement operators, and of
$\bigO{\dispRk^{2} \timepm{n} \log(n)}$ for \zcref[S]{pbm:kernel} assuming \(m=n\)
and for operators including Toeplitz-like, Vandermonde-like, and Cauchy-like
displacements \cite[Sec.\,5.6 and\,5.7]{pan2001}. It would not surprise us if
the results in \cite{BostanJeannerodSchost2008,structured2017} could be augmented
to also solve \zcref[S]{pbm:kernel}; yet, studying this is beyond the scope of
this paper.

As mentioned above, the randomization in these algorithms uses multiplication
by structured matrices built from random field elements; more precisely, one
draws these elements from a subset of \(\field\) of size polynomial in \(m+n\)
\cite{kaltofen1994,structured2017}. Over a small finite field, one therefore
has to work in a field extension of \(\field\) of degree logarithmic in
\(m+n\), which increases the cost by polylogarithmic factors. Besides, by
design, the resulting cost bounds hold only in expectation or with high
probability. A main appeal of a deterministic algorithm with the same cost is
that it removes both limitations. We also stress that we work over an arbitrary
field, which does not allow for the deterministic methods specific to the real
or complex numbers that rely on positive definiteness
\cite[Sec.\,5.5]{pan2001}.

\paragraph{Structured matrices and Pad\'e approximation}

Strong links between structured linear systems and variants of Pad\'e
approximation have long been known. A Pad\'e approximation instance can be
written as a Toeplitz linear system \cite{BakerGravesMorris1996}, and the
Hermite-Pad\'e or simultaneous Pad\'e extensions involving \(\dispRk\) power
series \cite{Hermite1874, Hermite1893, Pade1894} relate to block Toeplitz
systems of displacement rank \(\bigO{\dispRk}\). The interpolation versions of
these problems, sometimes called M-Pad\'e approximation \cite{Mahler1968,
Beckermann1992, vanBarelBultheel1992}, lead similarly to structured matrices
such as Vandermonde-like or Cauchy-like ones. Both directions of this
correspondence have been exploited in algorithms.

In one direction, Pad\'e approximation, done via a half-gcd-like algorithm, was
used to solve Toeplitz systems in \(\bigO{\timepm{n}\log(n)}\) operations
\cite{brent1980fast}. Matrix Pad\'e approximants were used to find inversion
formulas and fast solvers for block Toeplitz and mosaic Hankel matrices
\cite{LabahnChoiCabay1990, LabahnBeckermannCabay1995}
\cite[Sec.\,A]{EberlyGiesbrechtGiorgiStorjohannVillard2007}. Conversely, for a
long time, the fastest known algorithms for various Hermite-Pad\'e or M-Pad\'e
approximation instances were obtained by viewing them as structured linear
systems. For typical cases where the input consists of \(\dispRk\) power series
at precision \(n\), or of \(\dispRk\) polynomials evaluated at \(n\) points,
the above-mentioned divide and conquer solvers achieve \(\softO{\dispRk^2 n}\)
\cite{pan2001} and \(\softO{\dispRk^{\expmm-1} n}\)
\cite[Cor.\,1]{BostanJeannerodSchost2008}. This reduction was still the fastest
approach in 2015 for more general interpolation problems with multiplicities
\cite[Cor.\,2]{BostanJeannerodSchost2008} \cite{Chowdhury2015}.

Yet, in the last two decades, these cost bounds have been progressively matched
by dedicated deterministic algorithms that compute minimal bases of
Hermite-Pad\'e approximants \cite{Storjohann2006, ZhouLabahn2012,
JeannerodNeigerVillard2020} and M-Pad\'e approximants \cite{Neiger2016fast,
JeannerodNeigerSchostVillard2017, RosenkildeStorjohann2021}. These algorithms
build upon a divide and conquer approach that generalizes the half-gcd to
matrices \cite{BeckermannLabahn1994, GiorgiJeannerodVillard2003}, and they
exploit fast polynomial matrix operations. Building on this recent progress, in
this article we return to the first direction, similarly to
\cite{brent1980fast, LabahnChoiCabay1990, LabahnBeckermannCabay1995} but for an
arbitrary displacement rank and for general Toeplitz-like, Vandermonde-like,
and Cauchy-like matrices rather than specific block/mosaic Toeplitz ones. We
show that these recent polynomial matrix-based tools yield deterministic
structured solvers whose cost matches that of the fastest known randomized
ones.

\paragraph{Main result}

In this paper, we present a deterministic approach for
\zcref[S]{pbm:linsys,pbm:kernel}, and deduce explicit algorithms that realize this
approach for Toeplitz-like, Vandermonde-like, and Cauchy-like matrices. We
summarize these results in the following theorem. For more details,
and for the corresponding algorithms, we refer to \zcref[S]{thm:algo_toeplitz,
thm:algo_vandermonde, thm:algo_cauchy} and
\zcref[S]{algo:StructuredSolve-Toeplitz, algo:StructuredSolve-Vandermonde,
algo:StructuredSolve-Cauchy}.

\begin{theorem}
  \label{thm:main}%
  Let \(\dispRk \le \min(m,n)\), and let \(\dispOp\) be the displacement
  operator for the Toeplitz-like, Vandermonde-like, or Cauchy-like structure.
  There is a deterministic algorithm that solves
  \zcref[S]{pbm:linsys,pbm:kernel} using \(\bigO{\dispRk^{\expmm-1} (\timepm{m}
  \log(m) + \timepm{n} \log(n))}\) operations in \(\field\). The returned
  nullspace basis is described by generators of length at most \(\dispRk+1\),
  for a Toeplitz-like operator (first two cases) or a Vandermonde-like one
  (third case).
\end{theorem}

Hankel-like matrices are also covered, since reversing the order of the columns
transforms a Hankel-like matrix into a
Toeplitz-like one with the same displacement rank. Furthermore, the assumption
\(\dispRk \le \min(m,n)\) is harmless: if it does not hold, one can use some
fast rank-revealing matrix decomposition \cite{IbarraMoranHui1982,
Storjohann2000phd} to reduce to this situation deterministically, using
\(\bigO{\dispRk^{\expmm-1}(m+n)}\) operations in \(\field\).

Compared with the Las Vegas algorithm of \cite{structured2017} for
\zcref[S]{pbm:linsys}, our algorithm has the same cost bound, is deterministic,
and also supports rectangular matrices. In addition, it solves
\zcref[S]{pbm:kernel}, and thus improves upon the best known randomized bound
\(\softO{\dispRk^2 n}\) \cite[Sec.\,5.6 and\,5.7]{pan2001}. The supported
operators are described in \zcref[S]{sec:matvec_polops:prelim}; in particular,
we assume that the lists of points involved in Vandermonde- and Cauchy-like
operators are repetition-free. Note that randomized algorithms in
\cite{structured2017} cover a wider family of operators (see
\cite[p.\,735]{structured2017}).

\paragraph{Overview of the approach}

Reformulating a structured linear system as a polynomial equation is not new:
the equations we use are essentially those behind the well-known polynomial
interpretations of structured matrix-vector products (see
\zcref[S]{sec:matvec_polops}). The novelty is the algorithmic approach that we
design in order to solve these equations. These interleave two approximation
problems, and we are not aware of previous work that would handle equations of
this type via deterministic polynomial matrix algorithms. We now give an
overview of this approach.

In \zcref[S]{sec:matvec_polops}, for each of the studied structures, we give
equivalent polynomial formulations of the \(m \times n\) linear system \(\mA
\vu = \vv\) of displacement rank \(\dispRk\). This takes the form of a
univariate polynomial equation which can be roughly written as
\begin{equation}
  \label{eqn:overview:poly_formulation}
  v =
  \begin{bmatrix} g_1 & \cdots & g_\dispRk \end{bmatrix}
  \begin{bmatrix}
    c_1 \\ \vdots \\ c_\dispRk
  \end{bmatrix}
  \rem \polmod_1,
  \quad\text{where}~
  \begin{bmatrix}
    c_1 \\ \vdots \\ c_\dispRk
  \end{bmatrix}
  =
  \begin{bmatrix}
    h_1 \\ \vdots \\ h_\dispRk
  \end{bmatrix}
  u
  \rem \polmod_2,
\end{equation}
up to a reversal of the unknowns in the Toeplitz-like and Vandermonde-like cases.
Here, \(\polmod_1\) and \(\polmod_2\) are known polynomials of respective
degrees \(m\) and \(n\), which only depend on the considered displacement
operator; \(v\) and the \(g_i\)'s are polynomials of degree less than \(m\),
and the \(h_i\)'s are polynomials of degree less than \(n\), that are all known
from \(\vv\) and the generators \((\mG,\mH)\) of \(\mA\); and \(u\) is an
unknown polynomial directly related to the sought solution vector \(\vu\).

The two identities that intervene in \zcref[S]{eqn:overview:poly_formulation},
viewing the \(c_i\)'s as the unknowns in the first identity and viewing \(u\)
as the unknown in the second one, are classically encountered for moduli of the
form \(\polmod_1=x^m\) and \(\polmod_2=x^n\), in the context of Hermite-Pad\'e
approximation and simultaneous Pad\'e approximation
\cite{Hermite1874,Hermite1893,Pade1894}. In the general form, with arbitrary moduli, we
use the terminology of \emph{vector M-Pad\'e approximation} for the first
identity, and \emph{simultaneous M-Pad\'e approximation} for the second one
\cite{Mahler1968,Beckermann1992,vanBarelBultheel1992}. More details and references
on this topic are provided in \zcref[S]{sec:mpade}.

The first main step in our approach is to ignore the actual shape of the
\(c_i\)'s and treat them as unknowns: through nonhomogeneous vector M-Pad\'e
approximation (\zcref[S]{sec:mpade:vector}), we compute an \(\dispRk \times
\dispRk\) polynomial matrix \(\mP\) and an \(\dispRk \times 1\) polynomial
vector \(\vsol\) that generate all possible vectors \(\vc = [c_i]_i\) that
satisfy this equation:
\[
  v = [g_1 \; \cdots \; g_\dispRk] \vc \rem \polmod_1
  \quad\Leftrightarrow\quad
  \vc \in \{\mP \lambda + \vsol \mid \lambda \in \xmatspace{\dispRk}{1}\}.
\]
Then, reintroducing the constrained form of \([c_i]_i\) as in
\zcref[S]{eqn:overview:poly_formulation}, we get
\begin{equation}
  \label{eqn:overview:before_inverse}
  \begin{bmatrix}
    c_1 \\ \vdots \\ c_\dispRk
  \end{bmatrix}
  =
  \mP \lambda + \vsol
  =
  \begin{bmatrix}
    h_1 \\ \vdots \\ h_\dispRk
  \end{bmatrix}
  u
  \rem \polmod_2,
\end{equation}
where both \(\lambda \in \xmatspace{\dispRk}{1}\) and \(u \in \xring_{<n}\) are
unknown, but we are only interested in finding the latter polynomial. This
equation implies \(\deg(\mP \lambda + \vsol) < n\): from this, and thanks to
special properties of the matrix \(\mP\) which is computed in the so-called
Popov form (see \zcref[S]{sec:mpade:polymat}), we deduce degree constraints on the
unknown \(\lambda\) (see
\zcref[S]{lem:predictable_degree,lem:degree_bound_lambda}). In particular, we must
have \(\deg(\lambda) < n\), but even more precisely, we get the constraints
\(\deg(\lambda_i) < n - \delta_i\), where \(\delta_i\) is the maximum degree
occurring in the \(i\)th column of \(\mP\).

We are not aware of deterministic methods that would find even just
one solution \((\lambda,u)\) to \zcref[S]{eqn:overview:before_inverse} within a
complexity bound that meets our target \(\softO{\dispRk^{\expmm-1}
(m+n)}\). Indeed, the customary approach would be to interpret the equation as
\begin{equation*}
  \vsol
  =
  \begin{bmatrix}
    {-\mP} & \vh \\
  \end{bmatrix}
  \begin{bmatrix}
    \lambda \\
    u
  \end{bmatrix}
  \rem \polmod_2,
  \;\;\;\text{where}\;
  \vh =
  \begin{bmatrix}
    h_1 & \cdots & h_\dispRk
  \end{bmatrix}^\trsp.
\end{equation*}
This is similar to the vector M-Pad\'e approximation problem solved in the
first step, except that now the equation is given by the \(\dispRk \times
(\dispRk+1)\) matrix \([{-\mP} \;\; \vh]\) instead of just one vector. This can
also be interpreted as a matrix version of simultaneous M-Pad\'e approximation.
Using either viewpoint, because of this \(\dispRk \times (\dispRk+1)\)
matrix input, the fastest available algorithms
\cite{Neiger2016fast,RosenkildeStorjohann2021} for solving this equation have a
cost bound in \(\softO{\dispRk^\expmm (m+n)}\), which is a factor at least
\(\dispRk\) above our target.

To overcome this obstacle, we exploit the fact that we are only interested in
\(u\) and do not need to compute the vector \(\lambda\). With this in mind, we
use a second step that reduces the above equation to a simultaneous M-Pad\'e
approximation problem. Smoothing away some technicalities for the sake of
presentation, this goes as follows: the obtained matrix \(\mP\) is invertible
modulo \(\polmod_2\), and thus \zcref[S]{eqn:overview:before_inverse} can be
rewritten as
\begin{equation}
  \label{eqn:overview:almost_smpade}%
  \lambda
  =
  (\mF u - \vw) \rem \polmod_2,
  \;\;\;\text{where}\;\;
  \mF = \mP^{-1} \vh \rem \polmod_2
  \;\;\text{and}\;\;
  \vw = \mP^{-1}\vsol \rem \polmod_2.
\end{equation}
These vectors \(\mF\) and \(\vw\) can be computed efficiently, essentially
through the minimal kernel basis algorithm of \cite{ZLS2012}. Now \(\lambda\)
is isolated, and since we do not need its values, we can simplify
\zcref[S]{eqn:overview:almost_smpade} based on the degree constraints it satisfies:
we seek \(u \in \xring_{<n}\) such that
\[
  \text{the}\;\; i\text{th entry of}\;\;
  (\mF u - \vw) \rem \polmod_2
  \;\;\text{has degree less than} \;\; n - \delta_i .
\]
This is exactly an instance of simultaneous M-Pad\'e approximation. Hence our
third step: solve this instance, thanks to an extension of the algorithm of
\cite{RosenkildeStorjohann2016,RosenkildeStorjohann2021} to the nonhomogeneous
case.

The above description explains how to solve the linear system \(\mA \vu =
\vv\) deterministically, and when the input vector \(\vv\) is zero, this
computes vectors in the nullspace of \(\mA\). We prove that, because the
subprocedures give us generating sets for all solutions of both approximation
problems that intervene above, this approach actually computes a concise
description of the whole nullspace of \(\mA\). From the latter, one can
directly deduce the rank of \(\mA\) and displacement generators for a basis of
its nullspace.

\paragraph{Perspectives}

We expect our approach to be applicable to other displacement operators among
those discussed in \cite{pan2001,structured2017}. For some of them, how to
instantiate it seems straightforward, such as cyclic displacements using
multipliers \(\Zmat_{n,\rho}\) with \(\rho \neq 0\) (see
\zcref[S]{sec:matvec_polops:prelim}) instead of \(\Zmat_{n,0}\). Other operators, such
as the general ones of \cite[p.\,735]{structured2017} or Cauchy-like ones
with repeated points, would likely require polynomial approximation problems
that are more general than the ones used here, for example with several moduli
as in \cite{Neiger2016fast,RosenkildeStorjohann2021}.

Another direction concerns the cost bound itself. Under the slightly stronger
assumption that \(\timepm{kd} \in \bigO{k^{\expmm-1-\varepsilon}\timepm{d}}\)
for some constant \(\varepsilon > 0\) (which holds for quasilinear
\(\timepm{\cdot}\)), we expect that one could obtain a refined bound involving
logarithmic factors in \((m+n)/\dispRk\) only, such as
\(\bigO{\dispRk^{\expmm} (\timepm{m/\dispRk} \log(m/\dispRk) +
\timepm{n/\dispRk} \log(n/\dispRk))}\). This should follow from careful
analyses of the dominant steps of our algorithms, under the above assumption.
Such a bound would provide a better continuum between the complexity of fast
Pad\'e approximation in \(\bigO{\timepm{n} \log(n)}\) when \(\dispRk \in
\bigO{1}\) and dense linear algebra in \(\bigO{n^\expmm}\) (without logarithmic
factors) when \(\dispRk \in \Theta(n)\). We leave this refinement to future
work.

There has been recent progress on the inversion of square structured matrices.
In a preprint \cite{jeannerodpreprint2025}, Jeannerod gives a deterministic
reduction, with cost \(\softO{\dispRk^{\expmm-1}n}\), from the inversion of \(n
\times n\) Toeplitz-like matrices of displacement rank \(\dispRk\) to the
inversion of \(4n \times 4n\) block Toeplitz matrices with \(4\dispRk \times
4\dispRk\) blocks. The latter can be inverted deterministically in
\(\softO{\dispRk^{\expmm-1}n}\) operations
\cite{EberlyGiesbrechtGiorgiStorjohannVillard2007}. Given an invertible matrix
\(\mA\), being able to compute its inverse \(\mA^{-1}\) allows one to solve
\(k\) linear systems with \(\mA\) and right-hand sides \(\vv_1,\ldots, \vv_k\)
simply by doing a matrix product \( \mA^{-1} \cdot [\vv_1 \;\; \cdots \;\;
\vv_k]\). In the same situation, with the algorithm we develop in this article,
one would need to call \(k\) times the solver with each \(\vv_i\): studying
whether our approach can be used to compute generators for \(\mA^{-1}\) is left
as a perspective. A further perspective would be to solve \(k\) linear systems
with \( \mA\) and right-hand sides \(\vv_1,\ldots, \vv_k\), without redoing the
whole procedure for each right-hand side (which, in particular, computes \(k\)
times the same nullspace of \(\mA\)), in the case where \(\mA\) is an
arbitrary matrix: not necessarily invertible, and not necessarily square. 

\paragraph{Outline}

In \zcref[S]{sec:notation}, we introduce basic notation used throughout the paper.
We recall the definitions of displacement operators and structured matrices in
\zcref[S]{sec:matvec_polops}, which also contains reformulations of structured
linear systems as modular equations on univariate polynomials. These polynomial
formulations serve as the starting point of our approach, and we solve them by
a combination of two successive approximation problems. In \zcref[S]{sec:mpade}, we
study these two problems, called vector M-Padé approximation and simultaneous
M-Padé approximation, and we solve them efficiently in the nonhomogeneous case
by elaborating upon recent work on the homogeneous case. Finally, in
\zcref[S]{sec:main}, we present our deterministic approach for solving
\zcref[S]{pbm:linsys,pbm:kernel}, by explicitly showing how to instantiate it in
the three cases of Toeplitz-like, Vandermonde-like, and Cauchy-like structures.

\section{Notation}
\label{sec:notation}

The ring of univariate polynomials in \(x\) over \(\field\) is denoted by
\(\xring\); we use subscripts to indicate degree bounds, such as
\(\xring_{<d}\) for polynomials of degree less than \(d\). We denote
by \(\xring_{d}\) the set of polynomials of degree exactly \(d\). We let
\(\timepm{\cdot}\) be a time function for polynomial multiplication: two
polynomials in \(\xring_{<d}\) can be multiplied in \(\timepm{d}\) operations
in \(\field\). Although we do not require that \(\timepm{d}\) be quasilinear,
we recall that one can always take \(\timepm{d} \in \bigO{d
\log(d)\log(\log(d))}\) \cite{CantorKaltofen1991}, which is in \(\softO{d}\).

The space of matrices with \(m\) rows and \(n\) columns over \(\field\), resp.\
\(\xring\), is denoted by \(\matspace{m}{n}\), resp.\ \(\xmatspace{m}{n}\). For
a given matrix $\mM$, we denote by $\mM_{i*}$ its $i$-th row, by $\mM_{*j}$
its $j$-th column, and by \(\mM_{ij}\) its entry at \((i,j)\). We denote by
$\diag{a_1, \ldots, a_n}$ the \(n\times n\) diagonal matrix with diagonal
entries \(a_1,\ldots,a_n\); depending on the context, these entries will be in
\(\field\) or in \(\xring\).
The notation \(\expmm\) stands for a feasible exponent for the complexity of
matrix multiplication, with \(\expmm > 2\): two matrices in \(\matspace{m}{m}\)
can be multiplied in \(\bigO{m^\expmm}\) operations in \(\field\). Typical
values include \(\expmm = 3\) (naive multiplication), \(\expmm = \log_2(7)\)
(Strassen's algorithm \cite{Strassen1969}), or \(\expmm < 2.3714\) (best known
bound, \cite{AlmanDuanWilliamsXuXuZhou2025}).

With the above notation, one can multiply two polynomial matrices in
\(\xmatspace{\dispRk}{\dispRk}\) of degree less than \(d\) using
\(\bigO{\dispRk^\expmm \timepm{d}}\) operations in \(\field\). We use cost
bounds from \cite{JeannerodNeigerVillard2020, JeannerodNeigerSchostVillard2017,
RosenkildeStorjohann2016, RosenkildeStorjohann2021}, which require some mild
assumptions on matrix and polynomial multiplication:
\begin{itemize}
  \item the function \(d \mapsto \M(d)/d\) is nondecreasing;
  \item \(\timepm{kd} \in \bigO{k^{\expmm - 1} \M(d)}\) for integers \(k, d \ge 1\).
\end{itemize}
The assumption in the first item is often called the superlinearity property
\cite[Sec.\,8.3]{ModernComputerAlgebra}. The one in the second item roughly
means that if one uses subcubic matrix multiplication, then subquadratic
polynomial multiplication should be used accordingly. This is a mild assumption
since it is always satisfied for a quasilinear \(\timepm{\cdot}\), which is
feasible over any field \(\field\) as noted above.
For simplicity when deriving cost bounds, we allow real arguments in
\(\timepm{\cdot}\), as in \(\timepm{d/\dispRk}\), by setting \(\timepm{d} =
\timepm{\lceil d \rceil}\) for a real \(d > 0\). This extended
\(\timepm{\cdot}\) is still nondecreasing, and the above assumptions remain
valid up to a constant factor: \(\timepm{d}/d \le 2 \timepm{d'}/d'\) for \(1
\le d \le d'\), and \(\timepm{kd} \in \bigO{k^{\expmm-1} \timepm{d}}\) for real
\(k \ge 1\) and \(d \ge 0\). Again for simplicity, in cost bounds we often
write \(\log(d)\) for \(\max(1,\log_2(d))\) (this matters for \(d \le 1\) such
that \(\log_2(d)\) might be zero or negative).

We consider two kinds of bases for \(\xring_{<d}\), as a \(\field\)-vector
space of dimension \(d\). The first one is its canonical basis, also
called monomial basis, which is \((1,x,x^2,\ldots,x^{d-1})\). The second one,
its Lagrange basis, is defined from a list \(\vx = (x_1,\ldots,x_{d}) \in
\field^d\) which is repetition-free, that is, \(x_i \neq x_j\) for all
\(i \neq j\). In that case, the Lagrange basis of \(\vx\) is
\((\lpol_{\vx,1},\ldots,\lpol_{\vx,d})\), where \(\lpol_{\vx,i} = \prod_{j
\neq i} (x - x_j) / (x_i - x_j) \in \xring_{<d}\); by definition,
\(\lpol_{\vx,i}(x_i) = 1\) and \(\lpol_{\vx,i}(x_j) = 0\) for \(j \neq i\). In
this context, we often use the so-called master polynomial
\(\mpol_\vx = \prod_{1 \le i \le d} (x - x_i) \in \xring_d\). Furthermore,
given some values \(\vv \in \K^d\), we call (Lagrange)
interpolant of \((\vx,\vv)\) the unique polynomial \(p\) in \(\xring_{<d}\)
such that \(p(x_i) = v_i\) for \(1 \le i \le d\).

For a list of points \(\vx = (x_1,\ldots,x_k) \in \field^k\) (not necessarily
repetition-free), and for a polynomial \(f \in \xring\) or more generally a
rational fraction \(f \in \field(x)\) without pole in \(\{x_1,\ldots,x_k\}\),
we write \(f(\vx) = [f(x_1) \;\; \cdots \;\; f(x_k)]^\trsp \in
\matspace{k}{1}\) for the vector of evaluations of \(f\) at \(\vx\). For
example, with this notation, the above interpolant \(p\) of \((\vx,\vv)\)
satisfies \(p(\vx) = \vv\).

For polynomials \(p, q, \polmod \in\xring\) with \(\polmod \neq 0\), we write
\(q = p \bmod \polmod\) to mean that \(q-p\) is divisible by \(\polmod\). When
\(q\) is more specifically the unique remainder of degree less than
\(\deg(\polmod)\) in the division of \(p\) by \(\polmod\), we write \(q = p
\rem \polmod\); the corresponding quotient is denoted by \(p \quo
\polmod\). These operators have lower precedence than addition and
multiplication: for example, an expression such as ``\(p + g \cdot (h \cdot u
\quo x^{n-1} \rem x^{\ell}) \rem x^m\)'' (similar to that in
the proof of \zcref[S]{thm:linsys_poly:toeplitz}) stands for ``\((p + (g \cdot (((h
\cdot u) \quo x^{n-1}) \rem x^{\ell}))) \rem x^m\)''. We extend the
``\({\quo}\)'' and ``\({\rem}\)'' operators to vectors and matrices of polynomials:
for \(\mM \in \xmatspace{m}{n}\) and \(\polmod\) as above, we write \(\mM \quo
\polmod\) (resp.\ \(\mM \rem \polmod\)) for the \(m \times n\) polynomial
matrix whose entries are \(\mM_{ij} \quo \polmod\) (resp.\ \(\mM_{ij} \rem
\polmod\)).

We make use of reversals. For a vector \(\vv = (v_1, \ldots, v_{n}) \in
\K^n\), its reversal is \(\vecrev{\vv} = (v_{n}, \ldots, v_1)\).
For a polynomial \(p \in \xring_{\le d}\), its \(d\)-reversal is
\(\polrev{d}{p} = x^d p(1/x) \in \xring_{\le d}\), whose coefficients are those
of \(p\) in reversed order and possibly with shifted index.

\section{Matrix-vector products as polynomial operations}
\label{sec:matvec_polops}

In this section, we translate matrix-vector products \(\mA \vu\) over
\(\field\) into equivalent operations over \(\xring\), assuming \(\mA\) has
some given displacement structure. We start by giving some preliminary
definitions and notations about structured matrices in
\zcref[S]{sec:matvec_polops:prelim}. Then, in \zcref[S]{sec:matvec_polops:special}, we
summarize well-known such polynomial interpretations when \(\mA\) is a
Toeplitz, Vandermonde, or Cauchy matrix. Finally, in
\zcref[S]{sec:matvec_polops:general}, we use this to derive similar interpretations
more generally for a matrix \(\mA\) that is structured with respect to one of
the three displacement operators considered in this paper.
The core idea behind these interpretations of the matrix-vector product
\(\mA \vu\) is to use inversion formulas for the displacement equation \(\dispOp(\mA)
= \mG \mH^{\trsp}\), that is, formulas that express the matrix \(\mA\) as a
function of its displacement generators \((\mG,\mH)\). Although, to
the best of our knowledge, some of the polynomial interpretations in
\zcref[S]{sec:matvec_polops:general} are not present as such in the literature,
it should not be seen as surprising that they are feasible, and
using such inversion formulas is a customary approach to deal with
matrix-vector products (see for example
\cite{GohbergOlshevsky1994,pan2001,PanWang2003},
\cite[Sec.\,2]{BostanJeannerodSchost2008}).


\subsection{Displacement structures: notation and definitions}
\label{sec:matvec_polops:prelim}

We follow the classical terminology for the displacement structure
\cite{kailathkungmorf1979} \cite[Sec.\,1.3]{pan2001}. For matrices in
\(\matspace{m}{n}\), their structure depends on the choice of a linear operator
\(\dispOp : \matspace{m}{n} \to \matspace{m}{n}\). Those most commonly
encountered in the literature are defined from two matrices $\mP \in \K^{m
\times m}$ and $\mQ \in \K^{n \times n}$, and are the Sylvester operator
\(\dispOp(\mA) = \dispSyl{\mP, \mQ}(\mA) = \mP \mA - \mA \mQ\) and the Stein operator
\(\dispOp(\mA) = \dispSte{\mP, \mQ}(\mA) = \mA - \mP \mA \mQ\). For a given operator
\(\dispOp\), the \(\dispOp\)-displacement rank of a matrix \(\mA\) is
the rank of \(\dispOp(\mA)\). If this \(\dispOp\)-displacement rank is \(\le
\dispRk\), a \(\dispOp\)-generator for \(\mA\) is a compact
representation of the displaced matrix with only \(\dispRk(n+m)\) field
elements, formed by any pair of matrices $(\mG, \mH) \in \K^{m \times \dispRk}
\times \K^{n \times \dispRk}$ such that \(\dispOp(\mA) = \mG \mH^{\trsp}\).

The matrices \(\mP\) and \(\mQ\) are usually very special. For example, they may be
a diagonal matrix $\D(\vx) = \diag{x_1,\ldots,x_n} \in \K^{n \times n}$ for
some list $\vx = (x_1, \ldots, x_{n}) \in \K^{n}$, or a so-called cyclic
down-shift matrix $\Zmat_{n,\rho} \in \K^{n \times n}$, or its transpose
\(\Zmat_{n,\rho}^{\trsp}\). The latter is defined for a given $\rho \in
\K$ as
\[
  \Zmat_{n,\rho} =
  \begin{bmatrix}
      &        &        & \rho \\
    1 &        &        &      \\
      & \ddots &        &      \\
      &        & 1 &
  \end{bmatrix}
  \in \K^{n \times n}.
\]
Here, and hereafter, entries that are left blank in a matrix are zero entries.

The rest of this paper focuses on the following structures:
\begin{itemize}
  \item \emph{Toeplitz-like:} operator \(\dispOp = \dispSte{\Zmat_{m,0},
    \Zmat_{n,0}^{\trsp}}\).
  \item \emph{Vandermonde-like:} operator \(\dispOp = \dispSte{\D(\vx),
    \Zmat_{n,0}^{\trsp}}\) for a given repetition-free list $\vx \in \K^m$.
  \item \emph{Cauchy-like:} operator \(\dispOp = \dispSyl{\D(\vx), \D(\vy)}\) for given
    repetition-free and disjoint lists $\vx \in \K^m$ and $\vy \in \K^n$.
\end{itemize}
We restrict our attention to the case where the displacement operator is
invertible, meaning that $\dispOp(\mA)$ completely determines $\mA$, and
conversely \cite[Def.\,4.3.1]{pan2001}. For the Sylvester operator, this is the
case when $\mP$ and $\mQ$ do not share any eigenvalue, while for the Stein
operator, this is the case when no product of an eigenvalue of $\mP$ with an
eigenvalue of $\mQ$ equals $1$ \cite[Thm.\,4.3.2]{pan2001}. One can easily
verify that the three displacement operators above are indeed invertible,
thanks to the assumptions on the point lists $\vx$ and $\vy$ for the
Vandermonde-like and Cauchy-like structures.

\subsection{Fundamental structured matrices and polynomial operations}
\label{sec:matvec_polops:special}

Here are families of special structured matrices that play a fundamental role
in each of the structures we consider. For each one, we define the matrix
structure (where entries that are left blank are zero entries), and state what
polynomial operation it corresponds to. We omit proofs of these
straightforward correspondences, except for the last one which might be
lesser-known.

The first two families are lower and upper triangular Toeplitz matrices. They
correspond to polynomial truncated products and middle products, respectively,
and they are Toeplitz-like with displacement rank at most \(1\).

\begin{lemma}[\(\L_\ell(\vg) \leftrightarrow\) truncated product]
  \label{lem:matrixL}%
  For integers \(\ell \le m\), for $\vg \in \K^{m}$, define
  \[
    \L_\ell(\vg) =
    \scalemath{0.75}{\begin{bmatrix}
        g_0 &   &        &   \\
        g_1 & g_0 &        &   \\
        \vdots & \vdots & \ddots &        \\
        g_{\ell-1} & g_{\ell-2} & \cdots & g_0 \\
        \vdots & \vdots & \vdots & \vdots \\
        g_{m-1} & g_{m-2} & \cdots & g_{m-\ell}
    \end{bmatrix}}
    \in \matspace{m}{\ell}.
  \]
  This is the matrix of the linear map \(\xring_{<\ell} \to
  \xring_{<m}, p \mapsto g p \rem x^m\) in the canonical bases, where \(g \in
  \xring_{<m}\) is the polynomial with coefficient vector \(\vg\).
\end{lemma}

\begin{lemma}[\(\U_\ell(\vh) \leftrightarrow\) middle product]
  \label{lem:matrixU}%
  For integers \(\ell \le n\), for $\vh \in \K^{n}$, define
  \[
    \U_\ell(\vh) =
    \scalemath{0.75}{\begin{bmatrix}
        h_0 & h_1 & \cdots & h_{\ell-1} & \cdots & h_{n-1} \\
            & h_0 & \cdots & h_{\ell-2} & \cdots & h_{n-2} \\
            &     & \ddots & \vdots     & \cdots & \vdots \\
            & & & h_0                   & \cdots & h_{n-\ell} \\
    \end{bmatrix}}
    \in \K^{\ell \times n}.
  \]
  This is the matrix of the linear map \(\xring_{<n} \to \xring_{<\ell}, p
  \mapsto (h p \quo x^{n-1}) \rem x^\ell\) in the canonical bases, where \(h
  \in \xring_{<n}\) is the polynomial with coefficient vector \([h_{n-1},
  \ldots, h_0]\). Note that in the square case \(\ell = n\) one can omit ``\({}
  \rem x^\ell\)'' in the above formula.
\end{lemma}

The third family is that of Vandermonde matrices, whose corresponding
polynomial operation is the evaluation at the points that define the
matrix. They are Vandermonde-like with displacement rank \(1\).

\begin{lemma}[\(\V_n(\vx) \leftrightarrow\) polynomial evaluation]
  \label{lem:matrixV}%
  For integers \(m,n\), for $\vx \in \field^m$, define
  \begin{equation*}
    \V_n(\vx) =
    \scalemath{0.75}{\begin{bmatrix}
        1 & x_1 & x_1^2 & \cdots & x_1^{n-1} \\
        1 & x_2 & x_2^2 & \cdots & x_2^{n-1} \\
        \vdots & \vdots & \vdots & \ddots & \vdots \\
        1 & x_m & x_m^2 & \cdots & x_m^{n-1} \\
    \end{bmatrix}}
    \in \matspace{m}{n}.
  \end{equation*}
  This is the matrix of the linear map \(\xring_{<n} \to \field^m, p \mapsto
  p(\vx)\) in the canonical bases. If \(\vx\) is
  repetition-free, \(\V_n(\vx)\) is also the matrix of the linear
  map \(\xring_{<n} \to \xring_{<m}, p \mapsto p \rem \mpol_\vx\), using the
  canonical basis for \(\xring_{<n}\) and the Lagrange basis of \(\vx\) for
  \(\xring_{<m}\).
\end{lemma}

Finally, Cauchy matrices are defined from disjoint lists of points \(\vx\) and
\(\vy\) and correspond to the evaluation at \(\vx\) of some rational fraction
whose poles are \(\vy\). They are Cauchy-like with displacement rank \(1\).

\begin{lemma}[\(\C(\vx, \vy) \leftrightarrow\) rational fraction evaluation]
  \label{lem:matrixC}%
  For integers \(m,n\), for $\vx \in \field^m$ and \(\vy \in \field^n\) with
  $x_i \neq y_j$ for all $i,j$, define
  \begin{equation*}
    \C(\vx, \vy) =
    \scalemath{0.75}{\begin{bmatrix}
        \frac{1}{x_1 - y_1} & \frac{1}{x_1 - y_2} & \cdots & \frac{1}{x_1 - y_n} \\
        \frac{1}{x_2 - y_1} & \frac{1}{x_2 - y_2} & \cdots & \frac{1}{x_2 - y_n} \\
        \vdots & \vdots & \ddots & \vdots \\
        \frac{1}{x_m - y_1} & \frac{1}{x_m - y_2} & \cdots & \frac{1}{x_m - y_n} \\
    \end{bmatrix}}
    \in \matspace{m}{n}.
  \end{equation*}
  This is the matrix of the linear map \(\field^n \to \field^m,
  (v_1,\ldots,v_n) \mapsto v(\vx)\) in the canonical bases,
  where \(v \in \field(x)\) is the rational fraction \(\sum_{1 \le j \le n} v_j
  / (x - y_j)\).
  If \(\vx\) and \(\vy\) are repetition-free, then \(\C(\vx, \vy)\) is also the
  matrix of the linear map
  \[
    \xring_{<n} \to \xring_{<m}, p \mapsto \invmpol_\vy (\mpol'_\vy p \rem \mpol_\vy) \rem \mpol_\vx,
    \quad \text{where } \invmpol_\vy = \mpol_\vy^{-1} \rem \mpol_\vx \in \xring_{<m},
  \]
  using the Lagrange basis of \(\vy\) for \(\xring_{<n}\) and the Lagrange
  basis of \(\vx\) for \(\xring_{<m}\).
\end{lemma}
\begin{proof}
  Let \(\vv = [v_j]_{1 \le j \le n} \in \matspace{n}{1}\) and define \(v =
  \sum_{1 \le j \le n} v_j / (x - y_j) \in \field(x)\). By definition of the
  Cauchy matrix, it is obvious that the \(i\)th entry of the vector \(\C(\vx,
  \vy) \vv\) is \(v(x_i)\), for \(1 \le i \le m\). So we focus on proving the
  second claim. Note that \(\mpol_\vy\) is invertible modulo \(\mpol_\vx\)
  since \(\vx\) and \(\vy\) are disjoint; by construction, \(\invmpol_\vy(x_i)
  =
  \mpol_\vy(x_i)^{-1}\). The assumption that \(\vx\) and \(\vy\) are
  repetition-free allows us to use the Lagrange basis \(\{\lpol_{\vy,j}, 1 \le
  j \le n\}\) of \(\xring_{<n}\), and similarly for \(\xring_{<m}\).

  Let \(p = \sum_{1 \le j \le n} v_j \lpol_{\vy,j} \in \xring_{<n}\).
  The polynomial \(q = \invmpol_\vy (\mpol'_\vy p \rem \mpol_\vy) \rem
  \mpol_\vx\) is in \(\xring_{<m}\), and our goal is to prove that \(v(x_i) =
  q(x_i)\) holds for all \(1
  \le i \le m\), or, equivalently, \(v(x_i) = \mpol_\vy(x_i)^{-1} (\mpol'_\vy p
  \rem \mpol_\vy)(x_i)\).
  By definition of \(\lpol_{\vy,j}\), and using \(\mpol'_\vy(y_j) = \prod_{i
  \neq j} (y_j - y_i)\), we obtain
  \[
    \mpol_\vy v =
    \mpol_\vy \sum_{1 \le j \le n} \frac{v_j}{x - y_j} =
    \sum_{1 \le j \le n} \mpol'_\vy(y_j) v_j \mathfrak{L}_{\vy,j}
    = \mpol'_\vy p \rem \mpol_\vy.
  \]
  The last equality holds by the Lagrange interpolation formula, since
  \(\mpol'_\vy p \rem \mpol_\vy\) has degree less than \(n =
  \deg(\mpol_\vy)\) and its value at \(y_j\) is \(\mpol'_\vy(y_j)v_j\), for all \(1
  \le j \le n\). Taking the value at \(x_i\) yields \(\mpol_\vy(x_i) v(x_i) =
  (\mpol'_\vy p \rem \mpol_\vy)(x_i)\), which concludes the proof.
\end{proof}

\begin{remark}
  The ingredients in this proof are close to those behind the classical
  factorization
  \[
    \C(\vx, \vy) =
    \D(\mpol_\vy(\vx))^{-1} \V_n(\vx) \V_n(\vy)^{-1} \D(\mpol_\vy'(\vy))
    ,
  \]
  found for example in \cite[Eq.\,(3.6.5)]{pan2001} in the case \(n = m\).
\end{remark}

\subsection{Structured matrices and polynomial operations}
\label{sec:matvec_polops:general}

A major question when computing with a structured matrix \(\mA\) is to
understand how it can be expressed as a function of its generator
\((\mG,\mH)\), itself defined from the displaced matrix \(\dispOp(\mA)\). This
is the question of inverting the displacement operator and recovering \(\mA\)
from \((\mG,\mH)\), which has received a fairly general study; see
\cite[Sec.\,4.3 and\,4.4]{pan2001} and \cite{PanWang2003} and the references
therein. For Toeplitz-like matrices, this is known as the $\Sigma$-LU
decomposition, often stated in the square case \(m = n\). The rectangular case,
stated below, is directly deduced from \cite[Thm.\,4.7]{PanWang2003}. For
Vandermonde-like and Cauchy-like matrices, one may refer to
\cite[Ex.\,4.4.6.(b) and\,1.4.1]{pan2001}.

\begin{lemma}
  \label{lem:inv_disp}%
  Let $\mA \in \matspace{m}{n}$, \(\dispRk \le \min(m,n)\), and $(\mG, \mH) \in
  \matspace{m}{\dispRk} \times \matspace{n}{\dispRk}$, whose columns are
  denoted by \(\vg_j = \mG_{*j}\) and \(\vh_j = \mH_{*j}\). Let also  \(\ell
  = \min(m, n)\).
  \begin{enumerate}
    \item%
      \label{lem:inv_disp:toeplitz}%
      Toeplitz-like:
      $\dispSte{\Zmat_{m,0}, \Zmat_{n,0}^{\trsp}}(\mA) = \mG \mH^{\trsp}
      \Rightarrow
      \mA = \sum_{j=1}^{\dispRk} \L_\ell(\vg_{j}) \U_\ell(\vh_{j})$.

    \item%
      \label{lem:inv_disp:vandermonde}%
      Vandermonde-like:
      \(\dispSte{\D(\vx), \Zmat_{n,0}^{\trsp}}(\mA) = \mG \mH^{\trsp}
      \Rightarrow
      \mA = \sum_{j=1}^{\dispRk} \D(\vg_j) \V_n(\vx) \U_n(\vh_j)\).

    \item%
      \label{lem:inv_disp:cauchy}%
      Cauchy-like:
      \(\dispSyl{\D(\vx), \D(\vy)}(\mA) = \mG \mH^{\trsp}
      \Rightarrow
      \mA = \sum_{j=1}^{\dispRk} \D(\vg_j) \C(\vx, \vy) \D(\vh_j)\).
  \end{enumerate}
\end{lemma}

As a consequence of these formulas, we obtain the following results that give a
polynomial interpretation of structured linear systems.

\begin{theorem}[Polynomial form of Toeplitz-like system]
  \label{thm:linsys_poly:toeplitz}%
  Let \(\dispOp = \dispSte{\Zmat_{m,0}, \Zmat_{n,0}^{\trsp}}\). Let $\mA \in
  \matspace{m}{n}$, \(\dispRk \le \min(m,n)\), and $(\mG, \mH) \in
  \matspace{m}{\dispRk} \times \matspace{n}{\dispRk}$ be such that
  $\dispOp(\mA) = \mG \mH^{\trsp}$. For \(1 \le j \le \dispRk\), define the
  polynomials \(g_j \in \xring_{<m}\) and \(h_j \in \xring_{<n}\) whose
  coefficient vectors are \(\mG_{*j}\) and \(\vecrev{\mH_{*j}}\), respectively.
  For any \(\vu \in \matspace{n}{1}\) and \(\vv \in \matspace{m}{1}\),
  \[
    \vv = \mA \vu
    \quad\Leftrightarrow\quad
    v =
    \begin{bmatrix} g_1 & \cdots & g_\dispRk \end{bmatrix}
    \left(
      \begin{bmatrix}
        h_1 \\ \vdots \\ h_\dispRk
      \end{bmatrix}
      u
      \quo x^{n-1}
    \right)
    \rem x^m,
  \]
  where \(u \in \xring_{<n}\) and \(v \in \xring_{<m}\) are the polynomials
  with coefficient vectors \(\vu\) and \(\vv\).
\end{theorem}

\begin{proof}
  One has \(\mA \vu = \sum_{1 \le j \le \dispRk} \L_\ell(\mG_{*j})
  \U_\ell(\mH_{*j})\vu\), where \(\ell = \min(m,n)\), by
  \zcref[S]{lem:inv_disp:toeplitz} of \zcref[S]{lem:inv_disp}. For \(1 \le j \le
  \dispRk\), let \(f_j = g_j (h_j u \quo x^{n-1} \rem x^\ell) \rem x^m\).
  According to \zcref[S]{lem:matrixL,lem:matrixU}, its coefficient vector is
  \(\L_\ell(\mG_{*j}) \U_\ell(\mH_{*j})\vu\). Thus, it suffices to show \(f_j =
  g_j (h_j u \quo x^{n-1}) \rem x^m\), and then the conclusion follows by
  summing over \(j\). This equality is clear when \(\ell = m\), and when
  \(\ell = n\) it follows from \(\deg(h_j u \quo x^{n-1}) < n\), itself
  deduced from \(\deg(h_j u) \le 2(n-1)\).
\end{proof}

\begin{theorem}[Polynomial form of Vandermonde-like system]
  \label{thm:linsys_poly:vandermonde}%
  Let \(\vx \in \field^m\) be repetition-free and let \(\dispOp =
  \dispSte{\D(\vx), \Zmat_{n,0}^{\trsp}}\). Let $\mA \in \matspace{m}{n}$,
  \(\dispRk \le \min(m,n)\), and $(\mG, \mH) \in \matspace{m}{\dispRk} \times
  \matspace{n}{\dispRk}$ be such that \(\dispOp(\mA) = \mG \mH^{\trsp}\).
  For \(1 \le j \le \dispRk\), define \(g_j
  \in \xring_{<m}\) as the interpolant of \((\vx,\mG_{*j})\) and
  \(h_j \in \xring_{<n}\) as the polynomial with coefficient vector
  \(\vecrev{\mH_{*j}}\). For any \(\vu \in \matspace{n}{1}\) and \(\vv \in
  \matspace{m}{1}\),
  \[
    \vv = \mA \vu
    \quad\Leftrightarrow\quad
    v =
    \begin{bmatrix} g_1 & \cdots & g_\dispRk \end{bmatrix}
    \left(
      \begin{bmatrix}
        h_1 \\ \vdots \\ h_\dispRk
      \end{bmatrix}
      u
      \quo x^{n-1}
    \right)
    \rem \mpol_\vx,
  \]
  where \(u \in \xring_{<n}\) has coefficient vector \(\vu\) and \(v \in
  \xring_{<m}\) is the interpolant of \((\vx,\vv)\).
\end{theorem}

\begin{proof}
  By \zcref[S]{lem:inv_disp:vandermonde} of \zcref[S]{lem:inv_disp}, one has
  \(\mA \vu = \sum_{j=1}^{\dispRk} \D(\mG_{*j}) \V_n(\vx) \U_n(\mH_{*j}) \vu\).
  Observe that \(\D(\mG_{*j})\) is the matrix of the multiplication map
  \(\xring_{<m} \to \xring_{<m}, p \mapsto g_j p \rem \mpol_\vx\),
  using the Lagrange basis of \(\vx\) for \(\xring_{<m}\).
  Hence, the interpolant of \((\vx, \D(\mG_{*j}) \V_n(\vx) \U_n(\mH_{*j}) \vu)\) is
  \[
    g_j (h_j u \quo x^{n-1} \rem \mpol_\vx) \rem \mpol_\vx
    =
    g_j (h_j u \quo x^{n-1}) \rem \mpol_\vx
  \]
  according to \zcref[S]{lem:matrixU,lem:matrixV}, and summing these over \(j\)
  yields the interpolant of \((\vx, \mA \vu)\).
\end{proof}

\begin{theorem}[Polynomial form of Cauchy-like system]
  \label{thm:linsys_poly:cauchy}%
  Let \(\dispOp = \dispSyl{\D(\vx), \D(\vy)}\) for some disjoint and
  repetition-free \(\vx \in \field^m\) and \(\vy \in \field^n\).
  Let \(\invmpol_\vy = \mpol_\vy^{-1} \rem \mpol_\vx \in \xring_{<m}\).
  Let $\mA \in \matspace{m}{n}$, \(\dispRk \le \min(m,n)\), and $(\mG, \mH) \in
  \matspace{m}{\dispRk} \times \matspace{n}{\dispRk}$ be such that
  \(\dispOp(\mA) = \mG \mH^{\trsp}\). For \(1 \le j \le \dispRk\), define \(\bar{g}_j
  \in \xring_{<m}\) and \(\bar{h}_j \in \xring_{<n}\) as the interpolants of
  \((\vx,\mG_{*j})\) and \((\vy,\mH_{*j})\), respectively,
  and let \(g_j = \invmpol_\vy \bar{g}_j \rem \mpol_\vx \in \xring_{<m}\)
  and \(h_j = \mpol'_\vy \bar{h}_j \rem \mpol_\vy \in \xring_{<n}\).
  For any \(\vu \in \matspace{n}{1}\) and \(\vv \in \matspace{m}{1}\),
  \[
    \vv = \mA \vu
    \quad\Leftrightarrow\quad
    v =
    \begin{bmatrix} g_1 & \cdots & g_\dispRk \end{bmatrix}
    \left(
      \begin{bmatrix}
        h_1 \\ \vdots \\ h_\dispRk
      \end{bmatrix}
      u
      \rem \mpol_\vy
    \right)
    \rem \mpol_\vx,
  \]
  where \(u \in \xring_{<n}\) and \(v \in \xring_{<m}\) are the interpolants of
  \((\vy,\vu)\) and \((\vx,\vv)\), respectively.
\end{theorem}

\begin{proof}
  By \zcref[S]{lem:inv_disp:cauchy} of \zcref[S]{lem:inv_disp}, one has \(\mA
  \vu = \sum_{j=1}^{\dispRk} \D(\mG_{*j}) \C(\vx, \vy) \D(\mH_{*j}) \vu\).
  Similarly to the above remark about \(\D(\mG_{*j})\), the matrix
  \(\D(\mH_{*j})\) is that of the map \(\xring_{<n} \to
  \xring_{<n}, p \mapsto \bar{h}_j p \rem \mpol_\vy\), using the Lagrange basis of
  \(\vy\) for \(\xring_{<n}\). Hence, according to
  \zcref[S]{lem:matrixU,lem:matrixC}, the interpolant of \((\vx, \D(\mG_{*j})
  \C(\vx, \vy) \D(\mH_{*j}) \vu)\) is
  \[
    \bar{g}_j (\invmpol_\vy (\mpol'_\vy (\bar{h}_j u \rem \mpol_\vy) \rem \mpol_\vy) \rem \mpol_\vx) \rem \mpol_\vx
    =
    g_j (h_j u \rem \mpol_\vy) \rem \mpol_\vx.
  \]
  As before, summing these over \(j\) yields the interpolant of \((\vx, \mA \vu)\).
\end{proof}

\section{Computing nonhomogeneous M-Pad\'e approximants}
\label{sec:mpade}

In this section, we describe algorithms for two of the three costliest steps
used in our deterministic solver of \zcref[S]{sec:main}. These tools are
computing two variants of approximant bases: Hermite-Pad\'e approximants
\cite{Hermite1893,Pade1894,Mahler1968} and M-Pad\'e approximants
\cite{Lubbe1983,Beckermann1992,vanBarelBultheel1992}. Both variants are about
solving polynomial equations modulo some polynomial \(\polmod \in \xring\) of
degree \(d\). Classical Hermite-Pad\'e approximation is with \(\polmod =
x^d\), whereas for M-Pad\'e approximation the modulus splits as \(\polmod = \mpol_{\vx} =
\prod_{1 \le i \le d} (x - x_i)\) with known and distinct \(x_i\)'s. While these
two situations are sufficient for our needs in \zcref[S]{sec:main}, the
algorithms developed in this section work efficiently for an arbitrary
\(\polmod \in \xring\) of degree \(d\);
we still call M-Pad\'e approximation this general situation. Thus, with our
terminology, Hermite-Pad\'e approximation is the specific case of M-Pad\'e
approximation with \(\polmod = x^d\).

There are two types of M-Pad\'e approximation equations, known to be related by
some duality, called vector approximants (or type I, or Latin) and simultaneous
approximants (or type II, or German) \cite{Mahler1968}. In
\zcref[S]{sec:mpade:vector,sec:mpade:nonhsimultaneous}, we design efficient
algorithms for computing both types when the equations are nonhomogeneous; see
\zcref[S]{pbm:vmpade,pbm:smpade} for a detailed description of the solved problems.
To achieve this, we use known results from \cite{Neiger2016fast,
JeannerodNeigerVillard2020, RosenkildeStorjohann2016, RosenkildeStorjohann2021}
about the homogeneous case; for simultaneous approximants, this homogeneous
case is presented in \zcref[S]{sec:mpade:hsimultaneous}. Before this, preliminary
definitions and properties on polynomial matrices are given in
\zcref[S]{sec:mpade:polymat}.

\subsection{Polynomial matrices in reduced and Popov forms}
\label{sec:mpade:polymat}

We start with column degrees and different forms of column reduced matrices
\cite{wolovich1974,heymann1975,kailath1980linear}, and their extension with
degree shifts \cite{vanBarelBultheel1992,BeckermannLabahnVillard1999}.
Computing matrices in such forms is a core ingredient in most algorithms for
M-Pad\'e approximants (see \zcref[S]{sec:mpade:vector,sec:mpade:hsimultaneous}).

\begin{definition}[Shifted degrees; leading matrix; pivots; reversal]
  \label{def:degrees_lmat_reduced}%
  Let $\shift \in \Z^m$ and let $\mP \in \xmatspace{m}{n}$. The $\shift$-column
  degree of $\mP$, denoted by $\cdeg_{\shift}(\mP)$, is the tuple $\shiftt =
  (t_1, \ldots, t_n)$ where
  \[
    t_j = \cdeg_{\shift}(\mP_{*j}) = \max_{1 \le i \le m} (\deg(\mP_{ij}) + s_i) \in \Z \cup \{-\infty\}
  \]
  for $1 \le j \le n$.
  The \(\shift\)-column leading matrix of $\mP$, denoted
  by $\clm_{\shift}(\mP)$, is the matrix in $\matspace{m}{n}$ whose entry at
  position $(i,j)$ is the coefficient of degree $t_j - s_i$ of the polynomial
  $\mP_{ij}$ (this is \(0\) if $t_j - s_i < 0$). 
  If all columns of \(\mP\) are nonzero, the \(\shift\)-pivot index of \(\mP\)
  is the tuple \((\pi_1, \ldots, \pi_n)\) where \(\pi_j\) is the largest \(i\)
  such that \(\deg(\mP_{ij}) + s_i = t_j\), that is, the
  tuple of indices of the bottom-most nonzero entries of the columns of
  \(\clm_{\shift}(\mP)\).
  For \(\delta =
  (\delta_1,\ldots,\delta_n) \in \N^n\) such that \(\deg(\mP_{ij}) \le
  \delta_j\) holds for all \(i,j\), the \(\delta\)-column reversal of \(\mP\),
  denoted by \(\crev{\delta}{\mP}\), is the matrix in \(\xmatspace{m}{n}\)
  whose entry at position \((i,j)\) is \(\polrev{\delta_j}{\mP_{ij}}\). In
  particular, when \(\delta = (d,\ldots,d)\), we simply write
  \(\polrev{d}{\mP}\).
\end{definition}

For ease of notation, if $\shift = (0, \ldots, 0)$, we write $\cdeg(\mP)$ for
$\cdeg_{0}(\mP)$ and $\clm(\mP)$ for $\clm_{0}(\mP)$. When we write a
comparison of tuples of (shifted) degrees, this is to be understood as
entry-wise comparison. For example, the above constraint on \(\delta\), meaning
that all entries in the \(j\)th column of \(\mP\) have degree at most
\(\delta_j\), can be written as \(\cdeg(\mP) \le \delta\). For column reversals,
observe that \(\crev{\delta}{\mP} = \mP(x^{-1}) \diag{x^{\delta_1},\ldots,x^{\delta_n}}\).

These definitions are straightforwardly adapted to row-wise considerations.
For a shift $\shift$ now in $\Z^n$, this yields the
\(\shift\)-row degree $\rdeg_{\shift}(\mP)$ and the \(\shift\)-row leading
matrix $\rlm_{\shift}(\mP) \in \matspace{m}{n}$; and for a row degree bound
\(\delta \in \N^m\) such that
\(\rdeg(\mP) \le \delta\), this yields the \(\delta\)-row reversal
\(\rrev{\delta}{\mP}\). We will often use the fact
that \(\rdeg(\mP) < \delta\) is equivalent to \(\cdeg_{-\delta}(\mP) < 0\).

\begin{definition}[Reduced, weak Popov, Popov]
  \label{def:reduced_popov}%
  With the notation in \zcref[S]{def:degrees_lmat_reduced}, the matrix \(\mP\) is
  said to be
  \begin{itemize}
    \item \emph{$\shift$-reduced} if the matrix $\clm_{\shift}(\mP)$ has full
      column rank (equivalently in the square case: $\clm_{\shift}(\mP)$ is invertible);
    \item \emph{\(\shift\)-weak Popov} if none of the columns of \(\mP\) is zero
      and its \(\shift\)-pivot index is strictly increasing
      (equivalently in the square case: $\clm_{\shift}(\mP)$ is invertible and upper triangular);
    \item \emph{\(\shift\)-Popov} if \(\mP\) is square, $\clm_{\shift}(\mP)$ is
      invertible and unit upper triangular, and \(\rlm(\mP) = \clm(\mP^\trsp) =
      \ident_m\) holds.
  \end{itemize}
\end{definition}

Note that the property of being in \(\shift\)-reduced (resp.\
\(\shift\)-weak Popov) form is preserved when discarding columns, when scaling
a row by a nonzero constant, and when discarding zero rows while restricting
the shift accordingly.

Shifted Popov forms can be defined in the rectangular case, but we will not
need them; when \(\mP\) is square, in any of the three cases \(\mP\) is
nonsingular, i.e.\ \(\det(\mP) \neq 0\). When using the row-wise variants of
these definitions, we indicate it explicitly (\(\shift\)-row reduced,
\(\shift\)-row weak Popov, and \(\shift\)-row Popov). By definition,
\(\shift\)-weak Popov implies \(\shift\)-reduced, and a square matrix in
\(\shift\)-(column) Popov form is also row reduced. One important consequence
of reducedness is the predictable degree property
\cite[Thm.\,6.3-13]{kailath1980linear}, presented here in its shifted extension
\cite[Lem.\,3.6]{BeckermannLabahnVillard1999}.

\begin{lemma}[Predictable degree property]
  \label{lem:predictable_degree}%
  Let $\mP \in \xmatspace{m}{n}$, let $\shift \in \Z^m$, and let \(\shiftt =
  \cdeg_{\shift}(\mP)\). Then, for any \(\vq \in \xmatspace{n}{1}\), one has
  \(\cdeg_{\shift}(\mP \vq) \le \cdeg_{\shiftt}(\vq)\), with equality if \(\mP\) is
  $\shift$-reduced.
\end{lemma}

We define the \emph{\(\shift\)-negative part} of a matrix as the submatrix
formed by its columns of negative \(\shift\)-column degree. The following
lemma is a direct consequence of the predictable degree property.

\begin{lemma}[Negative part]
  \label{lem:negative_part}%
  Let \(\shift \in \Z^m\) and let \(\mQ \in \xmatspace{m}{n}\) be in
  \(-\shift\)-reduced (resp.\ \(-\shift\)-weak Popov) form. The
  \(-\shift\)-negative part \(\mP \in \xmatspace{m}{\ell}\) of \(\mQ\) is
  \(-\shift\)-reduced (resp.\ \(-\shift\)-weak Popov), and denoting \(\shiftt =
  -\cdeg_{-\shift}(\mP)\) (whose entries are positive), one has
  \[
    \{ \mQ\vv \mid \vv \in \xmatspace{n}{1}, \cdeg_{-\shift}(\mQ\vv) < 0 \}
    =
    \{ \mP\vu \mid \vu \in \xmatspace{\ell}{1}, \rdeg(\vu) < \shiftt \},
  \]
  and this set is a \(\field\)-vector space of dimension \(t_1 + \cdots + t_\ell\).
\end{lemma}
\begin{proof}
  The vector space dimension is obvious from the right-hand side and the
  fact that \(\mP\) has full column rank. For \(\vu \in \xmatspace{\ell}{1}\),
  \zcref[S]{lem:predictable_degree} gives \(\cdeg_{-\shift}(\mP\vu) =
  \cdeg_{-\shiftt}(\vu)\), and observe that \(\rdeg(\vu) < \shiftt\) if
  and only if \(\cdeg_{-\shiftt}(\vu) < 0\). It follows that the set in the
  right-hand side is included in the one of the left-hand side, by extending
  \(\vu\) with zero entries into \(\vv\) such that \(\mP \vu = \mQ \vv\).
  The converse follows from \zcref[S]{lem:predictable_degree} again, which implies
  that in any vector \(\vv \in \xmatspace{n}{1}\) such that
  \(\cdeg_{-\shift}(\mQ\vv) < 0\), the only entries \(v_j\) that can be
  nonzero are those such that \(\cdeg_{-\shift}(\mQ_{*j}) < 0\).
\end{proof}

We will also make use of a folklore refinement of the predictable degree
property for weak Popov forms and their pivot indices (it is proved for
example, in its row-wise variant, in \cite[Lem.\,2.6]{NeigerPernet2021}).

\begin{lemma}[Predictable pivot property]
  \label{lem:weak_popov_product}%
  Let \(\mP \in \xmatspace{m}{n}\), let $\shift \in \Z^m$, let \(\shiftt =
  \cdeg_{\shift}(\mP)\), and assume \(\mP\) is in \(\shift\)-weak Popov form
  with \(\shift\)-pivot index \((\pi_1, \ldots, \pi_n)\). For any \(\mQ \in
  \xmatspace{n}{k}\) in \(\shiftt\)-weak Popov form with \(\shiftt\)-pivot index
  \((\theta_1, \ldots, \theta_{k})\), one has \(\cdeg_{\shift}(\mP\mQ) =
  \cdeg_{\shiftt}(\mQ)\), \(\clm_{\shift}(\mP\mQ) = \clm_{\shift}(\mP)
  \clm_{\shiftt}(\mQ)\), and \(\mP\mQ\) is in \(\shift\)-weak Popov form with
  \(\shift\)-pivot index \((\pi_{\theta_1}, \ldots, \pi_{\theta_{k}})\).
\end{lemma}

\subsection{Nonhomogeneous vector M-Pad\'e approximants}
\label{sec:mpade:vector}

We consider a nonhomogeneous variant of the classical
notion \cite{Mahler1968,Lubbe1983,Beckermann1992,vanBarelBultheel1992} of
vector M-Pad\'e approximants; this is similar to Hermite-Pad\'e approximants of
type 1 \cite{Hermite1893,Pade1894}, but working with an arbitrary modulus \(M
\in \xring_d\) instead of specifically \(M = x^d\).

\begin{problem}{vector M-Pad\'e approximation, nonhomogeneous}{vmpade}
  \emph{Input:}
  positive integers \(d\) and \(\dispRk\),
  a modulus polynomial \(\polmod \in \xring\) of degree \(d\),
  a polynomial row vector \(\mF = [f_1 \;\; \cdots \;\; f_{\dispRk}] \in \xmatspace{1}{\dispRk}_{<d}\),
  a shift \(\shift \in \Z^\dispRk\),
  a polynomial \(v \in \xring_{<d}\).

  \emph{Output:}
  a tuple \((\mP,\mu,\vsol)\), consisting of \par

  \algoitem an \(\shift\)-weak Popov basis \(\mP \in \xmatspace{\dispRk}{\dispRk}\) of
    \(\{\vp \in \xmatspace{\dispRk}{1} \mid \mF \vp = 0 \bmod \polmod\}\); \par

  \algoitem the monic generator \(\mu \in \xring_{\le d}\) of the ideal \(\{q \in
  \xring \mid \exists \vp \in \xmatspace{\dispRk}{1}, \mF \vp = v q \bmod \polmod\}\); \par

  \algoitem a polynomial vector \(\vsol \in \xmatspace{\dispRk}{1}_{<d}\) such that \(\mF \vsol = v \mu \bmod \polmod\).
\end{problem}

Here, bases of \(\xring\)-submodules of \(\xmatspace{\dispRk}{1}\) such as
\(\{\vp \in \xmatspace{\dispRk}{1} \mid \mF \vp = 0 \bmod \polmod\}\) are
represented as polynomial matrices whose columns form a basis of this module.
These bases are necessarily square and nonsingular, as this module is free of
rank \(\dispRk\).

Fast algorithms for vector M-Pad\'e approximation have been designed in the
homogeneous case \(v = 0\), and assuming \(\polmod = \mpol_{\vx}\) for some
\(\vx \in \field^d\) \cite{BeckermannLabahn1994, GiorgiJeannerodVillard2003,
JeannerodNeigerSchostVillard2017}. Here, repetitions are allowed, hence this
includes \(\polmod = x^d\) when \(\vx = (0,\ldots,0)\). Up to logarithmic
factors, the algorithms in these references use \(\softO{\dispRk^{\expmm-1}
d}\) operations in \(\field\). However, the best known cost bound, which
involves the smallest amount of logarithmic factors, has been obtained with
tools specifically tailored to the case \(\polmod = x^d\)
\cite{Storjohann2006,ZhouLabahn2012,JeannerodNeigerVillard2020}. For the
homogeneous case, and when the shift entries are not too unbalanced (the
precise condition appears in \zcref[S]{prop:vmpade}), the best known bound is
\(\bigO{\dispRk^{\expmm} \M(\lceil d/\dispRk \rceil) \log(d)}\)
\cite[Thm.\,1.3]{JeannerodNeigerVillard2020}. For the general case with an
arbitrary \(\polmod\), the algorithms in the above-listed references are not
directly applicable, because they make use of the fact that \(\polmod\) splits
over \(\field\) into linear factors and exploit the knowledge of \(\vx\). This
general case was handled in \cite{Neiger2016fast} by elaborating upon the above
results, inducing a slightly larger logarithmic overhead.

Here, we show that one can support arbitrary moduli \(\polmod\), nonhomogeneous
equations \(v \neq 0\), and keep the logarithmic overhead low, achieving the
best known cost \(\bigO{\dispRk^{\expmm} \M(\lceil d/\dispRk \rceil) \log(d)}\)
with the only restriction that the shift satisfies the above-mentioned
balancedness. This is obtained through a reduction from the general case to the
Hermite-Pad\'e case \(\polmod=x^d\) (see \cite[Lem.\,2.5]{Neiger2016fast}), and
through the folklore idea that one can handle nonhomogeneous equations through
a suitable augmented homogeneous equation and shift (see \zcref[S]{thm:vmpade}).

The next lemma shows that the output of \zcref[S]{pbm:vmpade} is enough to
represent all solutions to the nonhomogeneous equation.

\begin{lemma}
  \label{lem:vmpade_generates_all}%
  Let \((d,\dispRk,\polmod,\mF,\shift,v)\) be some input to \zcref[S]{pbm:vmpade}.
  Any corresponding output \((\mP,\mu,\vsol)\) generates all solutions
  \((\vp,q)\) to \(\mF \vp = v q \bmod \polmod\) in the sense that
  \[
    \left\{
      \begin{bmatrix}
        \vp \\
        q
      \end{bmatrix}
      \in \xmatspace{(\dispRk+1)}{1}
      \Bigm\vert
      \mF \vp = v q \bmod \polmod
    \right\}
    =
    \left\{
      \begin{bmatrix}
        \mP \lambda + \vsol \nu \\
        \mu \nu
      \end{bmatrix}
      \Bigm\vert
      \lambda \in \xmatspace{\dispRk}{1},
      \nu \in \xring
    \right\}.
  \]
  In particular, \(\mF \vp = v \bmod \polmod\) has a solution \(\vp\) if and
  only if \(\mu = 1\), and in that case the set of all such solutions is
  \(\{\mP \lambda + \vsol \mid \lambda \in \xmatspace{\dispRk}{1}\}\).
\end{lemma}
\begin{proof}
  Let \(\vp \in \xmatspace{\dispRk}{1}\) and \(q \in \xring\) such that \(\mF
  \vp = v q \bmod \polmod\). By definition of \(\mu\), this implies \(q = \mu
  \nu\) for some \(\nu \in \xring\). We obtain \(\mF \vp = v \mu \nu = \mF \vsol \nu
  \bmod \polmod\), hence \(\mF (\vp - \vsol \nu) = 0 \bmod \polmod\). It follows
  from the basis property of \(\mP\) that \(\vp - \vsol \nu = \mP \lambda\) for
  some \(\lambda \in \xmatspace{\dispRk}{1}\). This shows the inclusion
  \(\subseteq\). For the other inclusion, given some  \(\lambda \in
  \xmatspace{\dispRk}{1}\) and \(\nu \in \xring\), it is easily verified that
  the properties of \((\mP,\mu,\vsol)\) imply \(\mF (\mP \lambda + \vsol \nu) =
  \mF \vsol \nu = v \mu \nu \bmod \polmod\).
\end{proof}

Now, we show that nonhomogeneous equations can be handled through an augmented
equation that is homogeneous, along with a suitably chosen shift.

\begin{theorem}
  \label{thm:vmpade}%
  Let \((d,\dispRk,\polmod,\mF,\shift,v)\) be some input to \zcref[S]{pbm:vmpade}.
  Define the \(\xring\)-module \(\mathcal{Q} = \{\vq \in
  \xmatspace{(\dispRk+1)}{1} \mid [\mF \;\; {-v}] \vq = 0 \bmod \polmod\}\).
  Consider the shift \(\bar{\shift} = (\shift,\max(\shift)+d) \in
  \Z^{\dispRk+1}\), so that any \(\bar{\shift}\)-weak Popov basis \(\mQ \in
  \xmatspace{(\dispRk+1)}{(\dispRk+1)}\) of \(\mathcal{Q}\) satisfies
  \(\mQ_{\dispRk+1,j} = 0\) for \(1 \le j \le \dispRk\).

  Let \(\mP \in \xmatspace{\dispRk}{\dispRk}\), let \(\mu \in \xring_{\le d}\) be
  monic, and let \(\vp \in \xmatspace{\dispRk}{1}_{\le \deg(\mu) + d}\). The matrix
  \[
    \mQ =
    \begin{bmatrix}
      \mP & \vp \\
      0 & \mu
    \end{bmatrix}
    \in \xmatspace{(\dispRk+1)}{(\dispRk+1)}
  \]
  is an \(\bar{\shift}\)-weak Popov basis of \(\mathcal{Q}\) if and only if
  \((\mP,\mu,\vp \rem \polmod)\) is a solution to \zcref[S]{pbm:vmpade}.
  In that case, if \(\mQ\) is in \(\bar{\shift}\)-Popov form, then \(\mP\) is
  in \(\shift\)-Popov form and \(\vp\) is the unique solution to \(\mF \vp =
  v \mu \bmod \polmod\) such that \(\rdeg(\vp) < \rdeg(\mP) \le d\)
  holds componentwise.
\end{theorem}

\begin{proof}
  First, we focus on the claim about \(\mQ_{\dispRk+1,j}\) being zero when
  \(\mQ\) is any \(\bar{\shift}\)-weak Popov basis of \(\mathcal{Q}\). Write
  such a basis with blocks as
  \[
    \mQ =
    \begin{bmatrix}
      \mP & \vp \\
      \vz & \mu
    \end{bmatrix}
    \text{ where }
    \mP \in \xmatspace{\dispRk}{\dispRk}, \;
    \vp \in \xmatspace{\dispRk}{1}, \;
    \vz \in \xmatspace{1}{\dispRk},
    \text{ and } \mu \in \xring.
  \]
  Let \(\delta = (\delta_1,\ldots,\delta_{\dispRk+1}) \in \Z_{\ge
  0}^{\dispRk+1}\) be the list of diagonal degrees \(\delta_j =
  \deg(\mQ_{j,j})\). It is known that
  \(\delta_1+\cdots+\delta_{\dispRk+1} \le \deg(\polmod) = d\) (see for example
  \cite[Cor.\,2.4]{NeigerVu2017}), hence in
  particular, \(\delta_j \le d\) for \(1 \le j \le \dispRk\). On the other
  hand, by definition of \(\bar{\shift}\)-weak Popov forms, we get
  \(\deg(\vz_j) + d + \max(\shift) < \delta_j + s_j\). This implies
  \(\deg(\vz_j) < \delta_j - d \le 0\), and therefore \(\vz_j = 0\), proving the
  first claim.

  Now we focus on the main claim. Define \(\mQ\) as in the statement, from
  some \((\mP,\mu,\vp)\).

  Assume \(\mQ\) is an \(\bar{\shift}\)-weak Popov basis of \(\mathcal{Q}\).
  First, it is obvious that \(\vsol := \vp \rem \polmod\) has degree less than
  \(d\) and satisfies \(\mF \vsol = v \mu \bmod \polmod\). Second, if
  \(\hat{\mu}\) is the generator of the ideal in the output of
  \zcref[S]{pbm:vmpade} and \(\hat{\vp}\) is a corresponding vector such that \(\mF
  \hat{\vp} = v \hat{\mu} \bmod \polmod\), then the vector \([\begin{smallmatrix}
  \hat{\vp} \\ \hat{\mu} \end{smallmatrix}]\) is in \(\mathcal{Q}\) and
  thus generated by the columns of \(\mQ\), which implies that \(\hat{\mu}\) is
  a multiple of \(\mu\), and therefore \(\mu = \hat{\mu}\).
  Third, \(\mP\) is in \(\shift\)-weak Popov form, since it is the
  \(\dispRk\times\dispRk\) leading principal submatrix of \(\mQ\) in
  \(\bar{\shift}\)-weak Popov form and the first \(\dispRk\) components of
  \(\bar{\shift}\) are given by \(\shift\). The block-triangular form of
  \(\mQ\) implies that \(\mF \mP = 0 \bmod \polmod\); it remains to
  prove that any \(\vq \in \xmatspace{\dispRk}{1}\) such that \(\mF \vq = 0
  \bmod \polmod\) is generated by the columns of \(\mP\). Indeed, for any such
  vector \(\vq\), one has that \([\begin{smallmatrix} \vq \\ 0
    \end{smallmatrix}]\) is in \(\mathcal{Q}\), so that \([\begin{smallmatrix} \vq \\ 0
  \end{smallmatrix}] = \mQ \lambda\) for some \(\lambda = [\lambda_i]_i \in
  \xmatspace{(\dispRk+1)}{1}\); equating the bottom entries of the vectors in the
  latter identity yields \(\lambda_{\dispRk+1} = 0\), hence
  \(\vq = \mP \lambda_{1:\dispRk}\) with \(\lambda_{1:\dispRk} = [\lambda_i]_{1 \le i
  \le \dispRk}\).

  We have proved that \((\mP,\mu,\vsol)\) is a solution to \zcref[S]{pbm:vmpade}.
  Now, conversely, we assume the latter, and want to prove that \(\mQ\) is an
  \(\bar{\shift}\)-weak Popov basis of \(\mathcal{Q}\). By construction of
  \(\mQ\), the \(\bar{\shift}\)-pivots of its first \(\dispRk\) columns are on
  its diagonal. This is also the case for its last column, since for \(1 \le i
  \le \dispRk\) we have \(\deg(\vp_i) \le \deg(\mu) + d\), which implies
  \(\deg(\vp_i) + \bar{\shift}_i \le \deg(\mu) + \max(\shift) + d = \deg(\mu) +
  \bar{\shift}_{\dispRk+1}\). Thus \(\mQ\) is in \(\bar{\shift}\)-weak Popov
  form. It is easily verified that all columns of \(\mQ\) are in
  \(\mathcal{Q}\), so it remains to prove that any \(\vq \in \mathcal{Q}\) is
  generated by these columns. Write \(\vq = [\begin{smallmatrix} \hat{\vp} \\
  \hat{\mu} \end{smallmatrix}]\); since \(\mF \hat{\vp} = v \hat{\mu} \bmod
  \polmod\), the definition of \(\mu\) implies that \(\hat{\mu} =
  \mu \nu\) for some \(\nu \in \xring\).  It follows that \(\mF \hat{\vp} = v \mu \nu =
  \mF \vp \nu \bmod \polmod\), hence \(\mF (\hat{\vp} - \vp \nu) = 0 \bmod
  \polmod\). Therefore, from the basis property of \(\mP\), there exists
  \(\lambda \in \xmatspace{\dispRk}{1}\) such that \(\hat{\vp} - \vp \nu = \mP
  \lambda\). We conclude that \(\vq = [\begin{smallmatrix} \hat{\vp} \\
  \hat{\mu} \end{smallmatrix}] = \mQ [\begin{smallmatrix} \lambda \\ \nu \end{smallmatrix}]\).

  The last claims, for the situation where \(\mQ\) is in \(\bar{\shift}\)-Popov
  form, follow from the definitions.
\end{proof}

We are now ready to describe \zcref[S]{algo:VectorMPade}, and its cost bound
in the next proposition.

\begin{proposition}
  \label{prop:vmpade}%
  Given some input \((d,\dispRk,\polmod,\mF,\shift,v)\) to \zcref[S]{pbm:vmpade},
  \algoName{algo:VectorMPade} outputs the unique solution \((\mP,\mu,\vsol)\)
  such that \(\mP\) is in \(\shift\)-Popov form and \(\rdeg(\vsol) < \rdeg(\mP)\).
  The sum of the entries of \(\rdeg(\mP)\) is at most \(d\), and \(\det(\mP)\)
  divides \(\polmod\). Assuming that \(\dispRk\) and \(\sum_{1\le i \le
  \dispRk} (s_i - \min(\shift))\) are in \(\bigO{d}\), it uses
  \(
    \bigO{\dispRk^\expmm \M(d/\dispRk) \log(d)}
  \)
  operations in \(\field\).
\end{proposition}

\begin{proof}
  The property on the sum of the entries of \(\rdeg(\mP)\) is a general fact
  for such bases of relations; see for example
  \cite[Cor.\,2.4]{NeigerVu2017}, observing that \(\mP\) is
  \(\shift\)-Popov and hence row reduced. The property that \(\det(\mP)\)
  divides \(\polmod\) is classical as well; it follows for example from results on
  matrix fractions. Indeed, one has equality between left and right matrix
  fraction descriptions \(\polmod^{-1} \mF = (\polmod/\gamma)^{-1} (\mF/\gamma)
  = \vq \mP^{-1}\) for some vector \(\vq \in \xmatspace{1}{\dispRk}\) and \(\gamma =
  \gcd(\polmod,\mF)\); both fractions are irreducible by construction, so the
  determinants of their denominators are equal \cite[Lem.\,6.5-9,
  p.\,446]{kailath1980linear}, that is, \(\det(\mP) = \polmod/\gamma\) up to a
  constant factor in \(\field\setminus\{0\}\).

  In the nonhomogeneous case \(v\neq 0\), the algorithm directly relies on
  \zcref[S]{thm:vmpade}. Note that the used shift \(\bar{\shift} =
  (\shift,\max(\shift)+d)\) satisfies \(\min(\bar{\shift}) = \min(\shift)\) and
  therefore \(\sum_{i} (\bar{s}_i - \min(\bar{\shift})) \le 2\sum_{i} (s_i -
  \min(\shift)) + d \in \bigO{d}\). Then, both the correctness
  and cost bound follow directly from those of the case \(v = 0\).

  The rest of the proof is about the homogeneous case \(v = 0\). One has \(\mu
  = 1\) and one may take \(\vsol = 0\), so the task is to find the
  \(\shift\)-Popov basis \(\mP \in \xmatspace{\dispRk}{\dispRk}\) of \(\{\vp
  \in \xmatspace{\dispRk}{1} \mid \mF \vp = 0 \bmod \polmod\}\). For this,
  \algoName{algo:VectorMPade} relies on \cite[Lem.\,2.5]{Neiger2016fast}, which
  states that \(\mP\) is the leading principal \(\dispRk \times \dispRk\)
  submatrix of the \(\shiftt\)-Popov basis \(\mQ \in
  \xmatspace{(\dispRk+1)}{(\dispRk+1)}\) of the module
  \[
    \mathcal{A} = \{\vp \in \xmatspace{(\dispRk+1)}{1} \mid [\mF \;\; \polmod] \vp = 0 \bmod x^\tau\},
  \]
  for \(\shiftt = (\shift,\min(\shift))\) and \(\tau =
  \max(\shift)-\min(\shift)+2d\). \zcref[S]{step:VectorMPade:WeakPopovApp} computes
  a \(\shiftt\)-weak Popov basis of \(\mathcal{A}\)
  \cite[Prop.\,7.3]{JeannerodNeigerVillard2020}, from which
  \zcref[S]{step:VectorMPade:FindDegs} deduces the diagonal degrees \(\delta \in
  \Z^{\dispRk+1}_{\ge 0}\) of the sought \(\shiftt\)-Popov basis.
  Then, \zcref[S]{step:VectorMPade:KnownDegApp} computes the \(\shiftt\)-Popov
  basis \(\mQ\) of \(\mathcal{A}\)
  \cite[Prop.\,5.1]{JeannerodNeigerVillard2020}, which completes the proof of
  correctness of \algoName{algo:VectorMPade}.

  For the complexity of \zcref[S]{step:VectorMPade:WeakPopovApp} we apply
  \cite[Prop.\,7.3]{JeannerodNeigerVillard2020}, with the main quantity being
  here
  \[
    \xi = \tau + \sum_{1\le i \le \dispRk+1} (t_i - \min(\shiftt)) = \tau + \sum_{1\le i \le \dispRk} (s_i - \min(\shift)).
  \]
  Thus \(\xi \le (\dispRk + 1) \tau\), and the result in
  \cite[Prop.\,7.3]{JeannerodNeigerVillard2020} together with the upper bound
  in the first item of \cite[Thm.\,1.3]{JeannerodNeigerVillard2020} show that
  \zcref[S]{step:VectorMPade:WeakPopovApp} uses \(\bigO{\dispRk^\expmm \M(\lceil
  \xi / \dispRk \rceil) \log(\tau)}\) operations in \(\field\). This is in
  \(\bigO{\dispRk^\expmm \M(d / \dispRk) \log(d)}\), since our assumptions
  imply that \(\dispRk\), \(\tau\), and \(\xi\) are all in \(\bigO{d}\). Finally,
  according to \cite[Prop.\,5.1]{JeannerodNeigerVillard2020},
  \zcref[S]{step:VectorMPade:KnownDegApp} uses \(\bigO{\dispRk^\expmm
  \M(\tau/\dispRk) \log(1+\tau/\dispRk)}\)
  operations in \(\field\), which is again bounded by \(\bigO{\dispRk^\expmm
  \M(d / \dispRk) \log(d)}\).
\end{proof}

\begin{algorithm}[htb]
  \algoCaptionLabel{VectorMPade}{d, \dispRk, \polmod, \mF, \shift = 0, v = 0}
  \begin{algorithmic}[1]

    \Require{%
      \algoitem positive integers \(d\) and \(\dispRk\), \par
      \algoitem a modulus polynomial \(\polmod \in \xring\) of degree \(d\), \par
      \algoitem a polynomial row vector \(\mF = [f_1 \;\; \cdots \;\; f_{\dispRk}] \in \xmatspace{1}{\dispRk}_{<d}\), \par
      \algoitem a shift \(\shift = (s_1,\ldots,s_\dispRk) \in \Z^{\dispRk}\) (default: \(\shift = (0,\ldots,0)\), uniform case), \par
      \algoitem a polynomial \(v \in \xring_{<d}\) (default: \(v = 0\), homogeneous case). \par
    }

    \Ensure{the solution \((\mP,\mu,\vsol)\) to \zcref[S]{pbm:vmpade} with
    \(\mP\) in \(\shift\)-Popov form and \(\rdeg(\vsol) < \rdeg(\mP)\).}

    \If{\(v = 0\)} \Comment{homogeneous: rely on approximant basis at sufficiently large order}

    \State \(\tau \gets \max(\shift) - \min(\shift) + 2d\)
    \State \(\mQ \gets\)
    \textproc{ShiftAroundMinAppBasis}\((\tau, [\mF \;\; \polmod], (s_1,\ldots,s_\dispRk,\min(\shift)))\)
      \Comment{\cite[Algo.\,7]{JeannerodNeigerVillard2020}, up to transposition, to match the row-wise convention of
      that paper}
      \label{step:VectorMPade:WeakPopovApp}

    \State \(\delta \in \Z^{\dispRk+1}_{\ge 0} \gets (\deg(\mQ_{1,1}), \ldots, \deg(\mQ_{\dispRk+1,\dispRk+1}))\)
      \Comment{\(\mQ\) is in \(\xmatspace{(\dispRk+1)}{(\dispRk+1)}\)}
      \label{step:VectorMPade:FindDegs}
    \State \(\mQ \gets\) \textproc{KnownDegAppBasis}\((\tau, [\mF \;\; \polmod], (s_1,\ldots,s_\dispRk,\min(\shift)), \delta)\)
      \Comment{\cite[Algo.\,5]{JeannerodNeigerVillard2020}}
      \label{step:VectorMPade:KnownDegApp}
    \State \(\mP \in \xmatspace{\dispRk}{\dispRk} \gets\) leading principal \(\dispRk\times\dispRk\) submatrix of \(\mQ\)
    \State \Return \((\mP, 1, 0)\)

    \Else \Comment{\(v \neq 0\), nonhomogeneous case, reduce to \(v = 0\) via augmented equation}

    \State \(\bar{\shift} \in \Z^{\dispRk+1} \gets (s_1,\ldots,s_\dispRk,\max(\shift)+d)\)
    \State \((\mQ, 1, 0) \gets\)
        \Call{algo:VectorMPade}{d, \dispRk+1, \polmod, [\mF \;\; {-v}], \bar{\shift}, 0}
    \State write \(\mQ  \in \xmatspace{(\dispRk+1)}{(\dispRk+1)}\) as \(\mQ = [\begin{smallmatrix} \mP & \vsol \\ 0 & \mu \end{smallmatrix}]\)
          where \(\mP \in \xmatspace{\dispRk}{\dispRk}\), \(\vsol \in \xmatspace{\dispRk}{1}\), \(\mu \in \xring\)
    \State \Return \((\mP, \mu, \vsol)\)

    \EndIf
    \end{algorithmic}
\end{algorithm}

\subsection{Homogeneous simultaneous M-Pad\'e approximants}
\label{sec:mpade:hsimultaneous}

Given a modulus \(\polmod\in\xring\) of degree \(d\), a polynomial vector \(\mF
= [f_1 \;\; \cdots \;\; f_\dispRk]^\trsp \in
\xmatspace{\dispRk}{1}\), and integers \(\shift = (s_1, \ldots,
s_\dispRk)\) all in \(\{0,\ldots,d\}\), the classical version of simultaneous
M-Pad\'e approximation \cite{Hermite1874,Mahler1968} asks to find a nonzero
denominator \(p \in \xring_{<d}\) and a vector of associated numerators \(\vr =
[r_1 \;\; \cdots \;\; r_\dispRk]^\trsp \in \xmatspace{\dispRk}{1}\) such that
\(f_i p = r_i \bmod \polmod\) and \(\deg(r_i) < s_i\), for all \(1 \le i \le
\dispRk\). These conditions can be written concisely as \(\mF p = \vr \bmod
\polmod\) and \(\rdeg(\vr) < \shift\). In some applications, including ours in
\zcref[S]{sec:main}, the only object of interest is the simultaneous
denominator \(p\): then, noting that the degree conditions imply \(r_i = f_i p
\rem \polmod\), the problem becomes that of finding \(p \in \xring_{<d}\) such
that \(\rdeg(\mF p \rem M) < \shift\).

\begin{remark}
  \label{rmk:existence_homogeneous_smpade}%
  Observe that such a nonzero solution \(p\) may not exist if the degree bounds
  \(\shift\) are too restrictive. On the other hand, for the most permissive
  bounds \(\shift = (d,\ldots,d)\), any \(p \in \xring_{<d}\) is a solution.
  More generally, one may interpret the problem as a \(\field\)-linear system
  with \(s_1 + \cdots + s_{\dispRk} + d\) unknowns (which are the coefficients
  of \(p\) and \(\vr\)) and  \(\dispRk d\) equations (which express \(\mF p =
  \vr \bmod \polmod\)). In particular, there is necessarily a nonzero solution
  \((p,\vr)\), hence with nonzero \(p\), as soon as \(s_1 + \cdots + s_\dispRk
  > (\dispRk-1) d\).
\end{remark}

As such, this homogeneous problem has been solved efficiently
\cite{RosenkildeStorjohann2016,RosenkildeStorjohann2021}, using
\(\softO{\dispRk^{\expmm-1} d}\) operations in \(\field\). More precisely,
these references solve the problem of computing a list \(\mP = [p_1 \;\; \cdots
\;\; p_\ell] \in \xmatspace{1}{\ell}\) of \(\ell \le \dispRk+1\) polynomials
that can be used to generate all such solutions \(p\) through \(\xring\)-linear
combinations. Note that, for some instances, it may be impossible to recover
the associated numerators \(\mF p_i \rem \polmod\) within the same cost bound,
as already the total number of coefficients from \(\field\) this involves may
be in \(\Theta(\dispRk^2 d)\), which exceeds this bound since \(\expmm-1 < 2\).
Besides, it is not known how to compute a single small denominator \(p\)
faster than by calling the algorithm of
\cite{RosenkildeStorjohann2016,RosenkildeStorjohann2021} which recovers a list
\(\mP\) that generates all solutions. Regarding these generating sets, we
use the following definition, adapted from
\cite[Sec.\,1]{RosenkildeStorjohann2021} (see
\zcref[S]{rmk:solution_specification}). We state it in a slightly more general
setting than above, following \cite{RosenkildeStorjohann2021} and as needed in
\zcref[S]{sec:mpade:nonhsimultaneous}: \(\mF\) may have several columns,
and the degrees of the corresponding several denominators may be bounded
individually.

\begin{definition}
  \label{def:smpade_basis}%
  Let \(d\), \(\dispRk\), and \(\beta\) be positive integers, let \(\polmod
  \in \xring\) have degree \(d\), let \(\mF \in
  \xmatspace{\dispRk}{\beta}_{<d}\), let \(\shift = (s_1,\ldots,s_\dispRk) \in
  \Z^{\dispRk}\) with \(0 \le s_i \le d\) for \(1 \le i \le \dispRk\), and let
  \(\vd = (d_1,\ldots,d_\beta) \in \Z^\beta\) with \(1 \le d_j \le d\) for \(1
  \le j \le \beta\). We write \(\bar{\shift} = (\shift,\vd) =
  (s_1,\ldots,s_\dispRk,d_1,\ldots,d_\beta) \in \Z^{\dispRk+\beta}\).
  A vector \(\vp \in \xmatspace{\beta}{1}\) is said to be a \emph{solution
  for \((\polmod,\mF,\shift,\vd)\)} if \(\rdeg(\vp) < \vd\) and \(\rdeg(\mF\vp
  \rem \polmod) < \shift\).
  A \emph{solution basis for \((\polmod,\mF,\shift,\vd)\)} is any triple
  \((\ell,\mP,\shiftt)\) consisting of an integer \(\ell \in \Z_{\ge 0}\), a
  matrix \(\mP \in \xmatspace{\beta}{\ell}\), and a shift \(\shiftt \in
  \Z^\ell\), with the following properties:
  \begin{itemize}[nosep]
    \item each column of \(\mP\) is a nonzero solution for
      \((\polmod,\mF,\shift,\vd)\);
    \item for any solution \(\vp \in \xmatspace{\beta}{1}\) for
      \((\polmod,\mF,\shift,\vd)\), the vector \(\rl{\vr}{\vp}\), with \(\vr =
      \mF\vp \rem \polmod\), is a \(\xring\)-linear combination of the
      columns of the matrix
      \[
        \bar{\mP} =
        \begin{bmatrix}
          \mF \mP \rem \polmod \\
          \mP
        \end{bmatrix}
        \in \xmatspace{(\dispRk+\beta)}{\ell};
      \]
    \item \(\bar{\mP}\) is \(-\bar{\shift}\)-column reduced, and
      \(\cdeg_{-\bar{\shift}}(\bar{\mP}) = -\shiftt\).
  \end{itemize}
  If furthermore \(\bar{\mP}\) is in \(-\bar{\shift}\)-weak Popov form and
  \(\pi = (\pi_1,\ldots,\pi_\ell)\) is its \(-\bar{\shift}\)-pivot index, we
  say that \((\ell,\mP,\shiftt,\pi)\) is a \emph{weak Popov solution basis for
  \((\polmod,\mF,\shift,\vd)\)}.
\end{definition}

We will mostly use the case \(\beta = 1\) and \(\vd = (d)\), which is the one
described at the beginning of this section. For simplicity, in that situation,
we omit \(\vd\) from the notation: a solution for \((\polmod,\mF,\shift)\),
with \(\mF \in \xmatspace{\dispRk}{1}_{<d}\), is a polynomial \(p \in
\xring_{<d}\) such that \(\rdeg(\mF p \rem \polmod) < \shift\). We still need
the definition with \(\beta = 2\) for deriving an algorithm for the
nonhomogeneous case in \zcref[S]{sec:mpade:nonhsimultaneous}, and it turns out that
handling a general \(\beta \ge 1\) instead of specifically \(\beta \in
\{1,2\}\) does not bring any complication.

The next lemma explains further the name ``solution basis'', and gives
some basic properties.

\begin{lemma}
  \label{lem:smpade_solution_space}%
  With the notation in \zcref[S]{def:smpade_basis},
  let \((\ell,\mP,\shiftt)\) be a solution basis for
  \((\polmod,\mF,\shift,\vd)\). Then the set of solutions for
  \((\polmod,\mF,\shift,\vd)\) is
  \(
    \{\mP\vq \mid \vq \in \xmatspace{\ell}{1}, \rdeg(\vq) < \shiftt\},
  \)
  which is a \(\field\)-vector space of dimension \(t_1 + \cdots + t_\ell \le
  d_1 + \cdots + d_\beta\). Furthermore, \(\rdeg(\mP) < \vd\), \(1 \le t_i \le
  \max(\vd)\) for \(1 \le i \le \ell\), and \(\ell \le \min(\dispRk+\beta,
  d_1+\cdots+d_\beta)\). In particular, \(\ell = 0\) if and only if the only
  solution is the trivial one \(\vp = 0\).
\end{lemma}
\begin{proof}
  Since the columns of \(\mP\) are solutions, \(\rdeg(\mP) < \vd\) and
  \(\rdeg(\mF\mP \rem \polmod) < \shift\), and thus \(-\shiftt =
  \cdeg_{-\bar{\shift}}(\bar{\mP}) < 0\). Since in addition each column of
  \(\mP\) is nonzero, \(-\shiftt \ge -\max(\vd)\). Besides, \(\bar{\mP}\) is
  \(-\bar{\shift}\)-reduced, hence has full column rank, so that \(\ell \le
  \dispRk+\beta\).

  It is easily verified from the definitions that taking the last \(\beta\)
  rows is a bijection between the set of solutions and \(\{\bar{\mP}\vu \mid
  \vu \in \xmatspace{\ell}{1}, \cdeg_{-\bar{\shift}}(\bar{\mP}\vu) < 0\}\).
  Thus, the description of the set of solutions and its dimension follow from
  \zcref[S]{lem:negative_part} applied to \(\mQ = \bar{\mP}\), which is its own
  \(-\bar{\shift}\)-negative part since \(\cdeg_{-\bar{\shift}}(\bar{\mP}) =
  -\shiftt < 0\). Finally, the set of solutions is contained in \(\{\vp \in
  \xmatspace{\beta}{1} \mid \rdeg(\vp) < \vd\}\), which has dimension \(d_1 +
  \cdots + d_\beta\), and \(\ell \le t_1 + \cdots + t_\ell\) since \(t_i \ge
  1\) for all \(i\).
\end{proof}

Observe that \(\ell > \beta\) is possible, in which case \(\mP\) has rank less
than \(\ell\). In that case, \(\vq \mapsto \mP\vq\) is still injective on
\(\{\vq \mid \rdeg(\vq) < \shiftt\}\) (as shown in the lemma), but not on
\(\xmatspace{\ell}{1}\).

\begin{remark}
  \label{rmk:solution_specification}%
  What we call here a \emph{solution basis} \((\ell,\mP,\shiftt)\) for
  \((\polmod,\mF,\shift,\vd)\) is directly related to what is called a
  \emph{solution specification} in
  \cite[Def.\,1.5]{RosenkildeStorjohann2021}, in the case where all moduli
  are equal. Let us outline the differences, which are superficial.
    We use right-multiplication and column bases, whereas
      \cite{RosenkildeStorjohann2021} uses left-multiplication and row bases;
      these viewpoints coincide up to matrix transposes.
    We swap the order of numerators and denominators: in \(\bar{\mP}\)
      and \(\bar{\shift}\), we place \(\mF\mP \rem \polmod\)
      before \(\mP\) and \(\shift\) before \(\vd\); this does not
      affect reducedness, but makes it easier to reach weak Popov output
      (see \zcref[S]{rmk:num_den_order_weak_popov} and note that no pivots or weak Popov forms are
      considered in \cite{RosenkildeStorjohann2021}).
    We explicitly add \(\ell\) to the solution \((\ell,\mP,\shiftt)\),
      whereas \cite[Def.\,1.5]{RosenkildeStorjohann2021} leaves this parameter
      implicit through the dimension of \(\mP\).
    Our solution contains the positive shift \(\shiftt\) such that
      \(\cdeg_{-\bar{\shift}}(\bar{\mP}) = -\shiftt\), whereas in
      \cite[Def.\,1.5]{RosenkildeStorjohann2021} the solution rather contains
      the negative shift \(-\shiftt\).
\end{remark}

In \zcref[S]{sec:appendix_smpade} and \zcref[S]{algo:HomogeneousSMPade}, we
revisit the divide-and-conquer algorithm of
\cite[Algo.\,3]{RosenkildeStorjohann2021} in the case of a common modulus
\(\polmod\), and we obtain the following result.

\begin{theorem}
  \label{thm:smpade_homogeneous}%
  There is an algorithm \nameref{algo:HomogeneousSMPade} which takes as input
  \((\polmod,\mF,\shift,\vd)\) as in \zcref[S]{def:smpade_basis} and computes a
  weak Popov solution basis
  \((\ell,\mP,\shiftt,\pi)\) for \((\polmod,\mF,\shift,\vd)\) using
  \(
    \bigO{(\dispRk+\beta)^{\expmm-1} \timepm{\beta d} \log(d)}
  \)
  operations in \(\field\).
\end{theorem}

For comparison, in the same context and assuming \(\beta < \dispRk\),
\cite[Thm.\,4.5]{RosenkildeStorjohann2021} achieves
\[
  \bigO{\dispRk^\expmm \timepm{\beta d/\dispRk} \log(\beta d/\dispRk)^2 +
    \dispRk \beta^{\expmm-1} \timepm{d} \log(d)^2 + \dispRk^{\expmm-1} \beta d
  \log(\dispRk)},
\]
which is in \(\bigO{\dispRk^{\expmm-1} \timepm{\beta d} \log(d)^2}\) by
superlinearity when \(\dispRk \le d\); there, the output solution basis is \(-\bar{\shift}\)-reduced
as required by definition, but not necessarily \(-\bar{\shift}\)-weak Popov.
Using the latter bound from \cite{RosenkildeStorjohann2021} would lead to an
additional logarithmic factor in our main result (\zcref[S]{thm:main}).

\subsection{Nonhomogeneous simultaneous M-Pad\'e approximants}
\label{sec:mpade:nonhsimultaneous}

Our main result in \zcref[S]{sec:main} requires an efficient algorithm
for a nonhomogeneous generalization of the above simultaneous M-Pad\'e
approximation problem. In this variant, one has an additional input vector
\(\vv = [v_1 \;\; \cdots \;\; v_\dispRk]^\trsp \in \xmatspace{\dispRk}{1}\) and
seeks \(p\) and \(\vr\) such that  \(\mF p = \vv + \vr \bmod \polmod\), with
\(\rdeg(\vr) < \shift\). As hinted at above, we will not need to compute the
numerators \(\vr\), hence our problem is to find a polynomial \(p \in
\xring_{<d}\) such that \(\rdeg((\mF p - \vv) \rem \polmod) < \shift\). The
homogeneous case discussed above arises as the particular case \(\vv=0\).
Note that when \(\vv\neq 0\), there may be no solution to this variant,
independently from the bounds in \(\shift\).

\begin{problem}{simultaneous M-Pad\'e approximation, nonhomogeneous}{smpade}
  \emph{Input:}
  positive integers \(d\) and \(\dispRk\),
  a modulus polynomial \(\polmod \in \xring\) of degree \(d\),
  polynomial column vectors \(\mF\) and \(\vv\) in \(\xmatspace{\dispRk}{1}_{<d}\),
  a shift \(\shift = (s_1,\ldots,s_\dispRk) \in \Z^{\dispRk}\) with \(0 \le s_i \le d\).

  \emph{Output:}
  a tuple \((\ell,\mP,\shiftt,\ssol)\) consisting of \par

  \algoitem a solution basis \((\ell,\mP,\shiftt)\) for \((\polmod,\mF,\shift)\), as in \zcref[S]{def:smpade_basis}; \par

  \algoitem a polynomial \(\ssol \in \xring_{< d}\) such that \(\rdeg((\mF
  \ssol - \vv) \rem \polmod) < \shift\), if one exists, else \(\ssol = \emptyset\).
\end{problem}

The next lemma shows that the output of \zcref[S]{pbm:smpade} is enough to
represent all solutions to the nonhomogeneous equation, similarly to
\zcref[S]{lem:vmpade_generates_all} concerning \zcref[S]{pbm:vmpade}.

\begin{lemma}
  \label{lem:smpade_generates_all}%
  Let \((d,\dispRk,\polmod,\mF,\shift,\vv)\) be some input to
  \zcref[S]{pbm:smpade}. Any corresponding output \((\ell,\mP,\shiftt,\ssol)\) with
  \(\ssol \neq \emptyset\) generates all solutions \(p \in \xring_{<d}\) to
  \(\rdeg((\mF p - \vv) \rem \polmod) < \shift\) in the sense that
  \[
    \{p \in \xring_{<d} \mid \rdeg((\mF p - \vv) \rem \polmod) < \shift\}
    =
    \{\ssol + \mP \vq \mid \vq \in \xmatspace{\ell}{1}, \rdeg(\vq) < \shiftt\}.
  \]
  When \(\ell = 0\), this means that the left-hand side set is the singleton \(\{\ssol\}\).
\end{lemma}
\begin{proof}
  Let \(p \in \xring\). From \(\deg(\ssol) < d\) and \(\rdeg((\mF
  \ssol - \vv) \rem \polmod) < \shift\) it follows that
  \begin{align*}
  & \deg(p) < d \text{ and } \rdeg((\mF p - \vv) \rem \polmod) < \shift \\
  & \Leftrightarrow \deg(p - \ssol) < d \text{ and } \rdeg(\mF (p - \ssol) \rem \polmod) < \shift \\
  & \Leftrightarrow p - \ssol \text{ is a solution for } (\polmod, \mF, \shift).
  \end{align*}
  Thus \(p\) is in the left-hand side set if and only if \(p - \ssol\) is a
  solution for \((\polmod, \mF, \shift)\), and the conclusion follows from
  \zcref[S]{lem:smpade_solution_space} with \(\beta = 1\).
\end{proof}

\begin{algorithm}
  \algoCaptionLabel{SimultaneousMPade}{d, \dispRk, \polmod, \mF, \shift, \vv = 0}
  \begin{algorithmic}[1]

    \Require{%
      \algoitem positive integers \(d\) and \(\dispRk\), \par
      \algoitem a modulus polynomial \(\polmod \in \xring\) of degree \(d\), \par
      \algoitem a polynomial column vector \(\mF\) in \(\xmatspace{\dispRk}{1}_{<d}\), \par
      \algoitem a shift \(\shift = (s_1,\ldots,s_\dispRk) \in \Z^{\dispRk}\) with \(0 \le s_i \le d\) for all \(i\), \par
      \algoitem a polynomial column vector \(\vv\) in \(\xmatspace{\dispRk}{1}_{<d}\) (default: \(\vv = 0\), homogeneous case). \par
    }

    \Ensure{a tuple \((\ell,\mP,\shiftt,\ssol)\) which solves \zcref[S]{pbm:smpade}.}

    \LComment{compute a weak Popov solution basis for a homogeneous instance}
    \State \((k, \mL, \bar{\shiftt}, \pi) \gets 
          \Call{algo:HomogeneousSMPade}{\polmod, [\mF \;\; {-\vv}], \shift, (d,1)}\),
          with \(0 \le k \le \dispRk+2\),
          \(\mL \in \xmatspace{2}{k}\),
          \(\bar{\shiftt} \in \Z_{>0}^k\),
          \(\pi \in \{1,\ldots,\dispRk+2\}^k\)
          \Comment{\zcref[S]{algo:HomogeneousSMPade} in \zcref[S]{sec:appendix_smpade}}
          \label{step:SimultaneousMPade:HomogeneousSMPade}

    \LComment{deduce the solution to the nonhomogeneous instance (\zcref[S]{pbm:smpade})}
    \If{\(k = 0\) \textbf{or} \(\pi_k \neq \dispRk+2\)}
      \Comment{the second row of \(\mL\) is zero}
      \State \Return \((k,\mL_{1*},\bar{\shiftt},\emptyset)\)
      \label{step:SimultaneousMPade:EarlyExit}
    \EndIf
    \LComment{from here on, \(\mL_{2,k} \in \field\setminus\{0\}\) and \(\mL_{2,j} = 0\) for \(1 \le j < k\)}
    \State \(\mP \in \xmatspace{1}{(k-1)}_{<d} \gets [\mL_{1,1} \;\; \cdots \;\; \mL_{1, k-1}]\);
          \(\shiftt \in \Z^{k-1} \gets (\bar{t}_{1}, \ldots, \bar{t}_{k-1})\)
    \State \(\ssol \in \xring_{<d} \gets \mL_{1,k} / \mL_{2,k}\)
      \label{step:SimultaneousMPade:ssol}
    \State \Return \((k-1, \mP, \shiftt, \ssol)\)
      \label{step:SimultaneousMPade:return}
  \end{algorithmic}
\end{algorithm}

The approach of \zcref[S]{algo:SimultaneousMPade} consists in solving the
homogeneous instance \((\polmod, [\mF \;\; {-\vv}], \shift, (d,1))\), which has
\(\beta = 2\). Its solutions are all vectors \(\rl{p}{c}
\in \xmatspace{2}{1}\) such that \(p \in
\xring_{<d}\), \(c \in \field\), and \(\rdeg((\mF p - c\vv) \rem \polmod) <
\shift\). Those with \(c = 0\) correspond to the solutions \(p\) for
the homogeneous instance \((\polmod,\mF,\shift)\), and those with \(c = 1\)
correspond to the nonhomogeneous solutions, that is, polynomials \(p = \ssol\)
as in the second item of \zcref[S]{pbm:smpade}. Because 
\algoName{algo:HomogeneousSMPade} returns a weak Popov solution basis, at most
one of its columns has a nonzero entry \(c\), and if any it must be the last
column, which makes it obvious how to separate both types of solutions.

\begin{proposition}
  \label{prop:smpade}%
  \algoName{algo:SimultaneousMPade} correctly solves \zcref[S]{pbm:smpade},
  using \(\bigO{\dispRk^{\expmm-1} \timepm{d} \log(d)}\) operations in
  \(\field\).
\end{proposition}

\begin{proof}
  Let \(\mS = [\mF \;\; {-\vv}]\) and \(\bar{\shift} = (s_1,\ldots,s_\dispRk,
  d, 1)\). By \zcref[S]{thm:smpade_homogeneous},
  \zcref[S]{step:SimultaneousMPade:HomogeneousSMPade} computes a weak Popov solution
  basis \((k,\mL,\bar{\shiftt},\pi)\) for the instance \((\polmod, \mS, \shift, (d,1))\)
  using \(\bigO{\dispRk^{\expmm-1} \timepm{d} \log(d)}\) operations in \(\field\).
  The only other step of the algorithm that uses arithmetic operations is the
  final division of \(\mL_{1,k}\) (a polynomial of degree less than \(d\)) by
  \(\mL_{2,k}\) (a nonzero constant), which costs \(\bigO{d}\). This proves the
  announced cost bound.
  According to \zcref[S]{def:smpade_basis}, the matrix
  \[
    \bar{\mL} =
    \begin{bmatrix}
      \mS\mL \rem \polmod \\
      \mL
    \end{bmatrix}
    \in \xmatspace{(\dispRk+2)}{k}
  \]
  is in \(-\bar{\shift}\)-weak Popov form, with
  \(\cdeg_{-\bar{\shift}}(\bar{\mL}) = -\bar{\shiftt}\) and
  \(-\bar{\shift}\)-pivot index \(\pi\).
  As explained above, the solutions for \((\polmod, \mS, \shift, (d,1))\) are
  the vectors \(\rl{p}{c}\) with \(p \in \xring_{<d}\), \(c \in \field\), and
  \(\rdeg((\mF p - c\vv) \rem \polmod) < \shift\). In particular, all entries
  of the last row \(\mL_{2*}\) of \(\bar{\mL}\) are either zero or nonzero
  constants. The choice of \(\bar{\shift}\) is such that if \(\mL_{2,j}\) is a
  nonzero constant, then the \(j\)th column of \(\bar{\mL}\) has
  \(-\bar{\shift}\)-degree at least \(-1\), and since this
  \(-\bar{\shift}\)-degree is negative, it must be exactly \(-1\). Since this
  value is reached at the bottommost row, this shows \(\pi_j = \dispRk+2\).
  Since \(\pi\) is strictly increasing by assumption, the only possible nonzero
  entry in \(\mL_{2*}\) is in the last column, and this happens if and only if
  \(k > 0\) and \(\pi_k = \dispRk + 2\).

  If the second row of \(\mL\) is zero, there cannot be a nonhomogeneous
  solution \(\ssol \in \xring_{<d}\) as in the second item of
  \zcref[S]{pbm:smpade}. Indeed, otherwise, \(\rl{\ssol}{1}\) would be a solution
  for \((\polmod, \mS, \shift, (d,1))\) and therefore should be a
  \(\xring\)-linear combination of the columns of \(\mL\), which is impossible
  if \(\mL_{2*}\) is zero. This shows that the output at
  \zcref[S]{step:SimultaneousMPade:EarlyExit} is correct.
  If the second row of \(\mL\) is not zero, meaning \(k>0\) and \(\mL_{2,k}
  \in \field \setminus \{0\}\), then the vector \(\rl{\ssol}{1}\) with \(\ssol
  = \mL_{1,k}/\mL_{2,k}\) (as in \zcref[S]{step:SimultaneousMPade:ssol})
  is a nonzero constant multiple of a solution and is thus also a solution for
  \((\polmod, \mS, \shift, (d,1))\). Thus \(\deg(\ssol) < d\) and
  \(\rdeg((\mF\ssol - \vv) \rem \polmod) < \shift\).

  Finally, defining \(k' = k\) if the second row of \(\mL\) is zero, and \(k' =
  k-1\) otherwise, let \(\mP = [\mL_{1,1} \;\; \cdots \;\; \mL_{1,k'}]\) and
  \(\shiftt = (\bar{t}_1,\ldots,\bar{t}_{k'})\) (as returned either at
  \zcref[S]{step:SimultaneousMPade:EarlyExit} or at
  \zcref[S]{step:SimultaneousMPade:return}). We show that \((k',\mP,\shiftt)\) is a
  solution basis for \((\polmod,\mF,\shift)\). Since for \(1 \le j \le k'\),
  \(\rl{\mL_{1,j}}{0}\) is a nonzero solution for \((\polmod, \mS, \shift,
  (d,1))\), we get that \(\mL_{1,j}\) is a nonzero solution for
  \((\polmod,\mF,\shift)\). Since \(\mS \rl{\mL_{1,j}}{0} = \mF
  \mL_{1,j}\), the matrix \(\bar{\mP} = [\begin{smallmatrix} \mF\mP \rem
  \polmod \\ \mP \end{smallmatrix}]\) is the submatrix formed by the first
  \(k'\) columns and first \(\dispRk+1\) rows of \(\bar{\mL}\). Only zero
  entries appear in the discarded row, so as noted in \zcref[S]{sec:mpade:polymat},
  this implies that \(\bar{\mP}\) is in \(-(\shift,d)\)-weak Popov form, hence
  \(-(\shift,d)\)-reduced, and \(\cdeg_{-(\shift,d)}(\bar{\mP}) = -\shiftt\).
  Finally, if \(p\) is a solution for \((\polmod,\mF,\shift)\) and \(\vr = \mF
  p \rem \polmod\), then \(\rl{p}{0}\) is a solution for \((\polmod, \mS, \shift,
  (d,1))\), hence \([\vr^\trsp \;\; p \;\; 0]^\trsp\) is a \(\xring\)-linear combination of
  the columns of \(\bar{\mL}\), and therefore of its first \(k'\) columns;
  restricting to the first \(\dispRk+1\) rows shows that \(\rl{\vr}{p}\) is a
  \(\xring\)-linear combination of the columns of \(\bar{\mP}\).
\end{proof}

%

\section{Deterministic approach for systems and nullspaces}
\label{sec:main}

In this section, we prove \zcref[S]{thm:main} by presenting three algorithms which
deterministically solve \zcref[S]{pbm:linsys,pbm:kernel} for Toeplitz-like,
Vandermonde-like, or Cauchy-like matrices. They are direct instantiations
of our approach sketched in \zcref[S]{sec:intro}.

\subsection{Some preliminaries on polynomial matrices}
\label{sec:main:polymat}

The next lemma elaborates upon the predictable degree property and provides a
key technical ingredient for the derivation of degree bounds on polynomial
unknowns that arise in our algorithm.

\begin{lemma}
  \label{lem:degree_bound_lambda}%
  Let \(\mP\in\xmatspace{\dispRk}{\dispRk}\) be in weak Popov form with column
  degree \(\delta = \cdeg(\mP) \in \N^\dispRk\), and let \(\vp \in
  \xmatspace{\dispRk}{1}\) be a vector such that \(\rdeg(\vp) < \delta\). For
  any \(\lambda\in\xmatspace{\dispRk}{1}\) and \(n \in \N\) such that
  \(\cdeg(\mP\lambda + \vp) < n\), one has both
  \(\cdeg(\mP\lambda) < n\) and \(\cdeg(\vp) < n\),
  and furthermore \(\rdeg(\lambda) < n-\delta\), that is,
  \(\deg(\lambda_i) < n - \delta_i\) for all \(1 \le i \le \dispRk\).
\end{lemma}

\begin{proof}
  First observe that since \(\cdeg(\mP\lambda + \vp) < n\), if any of the
  inequalities \(\cdeg(\mP\lambda) < n\) and \(\cdeg(\vp) < n\) holds, then the
  other one holds as well. Furthermore, if \(\cdeg(\mP\lambda) < n\), then the
  predictable degree property in \zcref[S]{lem:predictable_degree} gives
  \(\cdeg_{\delta}(\lambda) < n\) since \(\mP\) is reduced; this precisely means
  \(\deg(\lambda_i) < n - \delta_i\) for all \(1 \le i \le \dispRk\). By the
  above remark, the same conclusion also holds when \(\cdeg(\vp) < n\); thus
  it remains to prove that it cannot happen that both \(\cdeg(\vp) \ge n\) and
  \(\cdeg(\mP\lambda) \ge n\) hold. By contradiction, assume \(\cdeg(\vp) \ge
  n\) and \(\cdeg(\mP\lambda) \ge n\); in particular, \(\vp\)
  and \(\mP\lambda\) are nonzero. It follows from \(\cdeg(\mP\lambda + \vp) <
  n\) that \(\cdeg(\vp) = \cdeg(\mP\lambda)\). More precisely, the bottommost
  entries of largest degree in \(\vp\) and \(\mP\lambda\) must collide:
  formally, writing \(\vq = \mP\lambda\) and defining \(i\) (resp.\ \(j\)) as
  the largest index in \(\{1,\ldots,\dispRk\}\) such that \(\deg(\vp_i) =
  \cdeg(\vp)\) (resp.\ \(\deg(\vq_j) = \cdeg(\vq)\)), we have \(i = j\) and
  \(\deg(\vp_i) = \deg(\vq_i)\). On the other hand, this definition of \(j\)
  along with the fact that \(\mP\) is in weak Popov form imply that the
  \(j\)-th entry of \(\vq = \mP\lambda\) has degree at least \(\delta_j\),
  according to \cite[Lem.\,1.17]{Neiger2016phd}. Altogether, we have proved
  \(\deg(\vp_i) = \deg(\vq_i) \ge \delta_i\), which contradicts the assumption
  that \(\rdeg(\vp) < \delta\).
\end{proof}

Another useful property that intervenes in our main algorithm is that, assuming
suitable degree bounds, the reversal of a product of two polynomial matrices is
equal to the product of their reversals.

\begin{lemma}
  \label{lem:reverse_product}%
  Let \(\mP \in \xmatspace{\beta}{\dispRk}\) and let \(\delta \in \N^\dispRk\) be such that
  \(\cdeg(\mP) \le \delta\). Let \(\mQ \in \xmatspace{\dispRk}{\gamma}\) and let \(n
  \in \N\) be such that \(\rdeg(\mQ) < \shift\)
  where \(\shift = (n-\delta_1,\ldots, n-\delta_\dispRk)\). Then, one has
  \(\deg(\mP\mQ) < n\), and \(\polrev{n-1}{\mP\mQ} = \bar{\mP} \bar{\mQ}\),
  where \(\bar{\mP} = \crev{\delta}{\mP}\) and \(\bar{\mQ} =
  \rrev{\shift-1}{\mQ}\) and $\shift-1 = (s_1-1,\ldots,s_\dispRk-1)$.
\end{lemma}

\begin{proof}
  The bound on \(\deg(\mP\mQ)\) is obvious. Observe that any index \(i\) such
  that \(\delta_i \ge n\), i.e.\ \(s_i \le 0\), corresponds to a zero
  row in \(\mQ\). Then, by definition of reversals, we have
  \begin{align*}
    \polrev{n-1}{\mP(x) \mQ(x)} & = x^{n-1} \mP(x^{-1}) \mQ(x^{-1}) \\
                                & = \mP(x^{-1}) \diag{x^{n-1},\ldots,x^{n-1}} \mQ(x^{-1}) \\
                                & = \mP(x^{-1}) \diag{x^{\delta_1},\ldots,x^{\delta_\dispRk}}\diag{x^{s_1-1},\ldots,x^{s_\dispRk-1}} \mQ(x^{-1}) \\
                                & = \bar{\mP}(x) \bar{\mQ}(x).
  \end{align*}
\end{proof}

%

We conclude these preliminaries with a computational tool. For a given nonzero
polynomial \(\polmod \in \xring\), a square matrix \(\mP \in \xmatspace{\dispRk}{\dispRk}\)
is invertible modulo \(\polmod\) if and only if \(\det(\mP)\) is coprime with
\(M\). In this case, we write \(\mP^{-1} \rem \polmod\) for the unique matrix
\(\mQ \in \xmatspace{\dispRk}{\dispRk}\) such that both \(\deg(\mQ) < \deg(\polmod)\)
and \(\mQ \mP = \ident_\dispRk \bmod \polmod\) hold.

\begin{lemma}
  \label{lem:Pinv_vec_mod}%
  Let \(\dispRk,m,n\) be positive integers with \(m\ge \dispRk\) and \(n \ge
  \dispRk\). Let \(\mP \in \xmatspace{\dispRk}{\dispRk}\) have column degree
  \(\delta = \cdeg(\mP) \in \N^\dispRk\) with \(\delta_1 + \cdots +
  \delta_\dispRk \le m\), let \(\vv \in \xmatspace{\dispRk}{1}_{<n}\),
  and let \(\polmod \in \xring\) have degree \(n\) and be coprime with
  \(\det(\mP)\). Then \(\mP\) is invertible modulo \(\polmod\) and one can
  compute \(\mP^{-1} \vv \rem \polmod\) using \(\bigO{\dispRk^{\expmm-1}
  \timepm{m+n} \log((m+n)/\dispRk)}\) operations in \(\field\).
\end{lemma}

\begin{proof}
  Define \(\shift = (\delta_1, \ldots, \delta_\dispRk, n-1) \in \N^{\dispRk+1}\) and \(\mF
  = [\mP \;\; {-\vv}] \in \xmatspace{\dispRk}{(\dispRk+1)}\). We have \(\cdeg(\mF) \le
  \shift\) by construction. Since \(\gcd(\polmod,\det(\mP)) = 1\),
  \(\mP\) is invertible modulo \(\polmod\). In particular, \(\mP\)
  is nonsingular and \(\mF\) has rank \(\dispRk\). Thus,
  any basis of the right kernel \(\{\vp \in \xmatspace{(\dispRk+1)}{1} \mid
  \mF\vp = 0\}\) consists of a single vector.
  Let \(\mK \in \xmatspace{(\dispRk+1)}{1}\) be such a kernel basis, and write it as
  \(\mK = [\begin{smallmatrix} \vu \\ \mu \end{smallmatrix}]\) where \(\vu \in
  \xmatspace{\dispRk}{1}\) and \(\mu \in \xring \setminus \{0\}\). From \(\mF \mK =
  0\), we get \(\mP \vu = \mu \vv\). We claim that \(\mu\) is invertible modulo
  \(\polmod\): then, the sought vector \(\mP^{-1} \vv \rem \polmod\) can be
  obtained as \(\mu^{-1} \vu \rem \polmod\).

  To prove our claim, we show that \(\mu\) divides \(\det(\mP)\), which is
  sufficient since \(\det(\mP)\) is coprime with \(\polmod\). Let \(\mA =
  \det(\mP) \mP^{-1} \in \xmatspace{\dispRk}{\dispRk}\) be the adjugate of \(\mP\).
  Left-multiplying the above identity \(\mP \vu = \mu \vv\) by \(\mA\), we
  obtain \(\det(\mP) \vu = \mu \mA \vv\). It follows that \(\mu\) divides
  \(\det(\mP) u_i\) for \(1 \le i \le \dispRk\). Defining \(g = \gcd(\mu,
  \det(\mP))\), we deduce that \(\mu / g\) divides \(u_i \det(\mP) / g\), and
  since  \(\mu / g\) is coprime with \(\det(\mP) / g\), in fact \(\mu / g\)
  divides \(u_i\) for all \(i\). Thus \(\mu / g\) is a common divisor to all
  entries of \(\mK\), which implies that \(\mu/g\) is a nonzero constant from
  \(\K\). Indeed, the polynomial vector \(g\mK/\mu = [\begin{smallmatrix}
  g\vu/\mu \\ g \end{smallmatrix}]\) is in the kernel of \(\mF\) and thus must
  be a polynomial multiple of the basis \(\mK\), which is not possible if
  \(\deg(\mu/g) \ge 1\). We have proved \(\mu/g \in \K\setminus\{0\}\), and
  therefore \(\mu\) divides \(\det(\mP)\).

  As for the cost bound, the main task is to compute a right kernel basis
  \(\mK\) of \(\mF\), for which we use \cite[Algo.\,1]{ZLS2012}. According to
  the analysis in \cite[Prop.\,B.1]{JeannerodNeigerSchostVillard2017}, since
  the sum of the entries of \(\shift\) is less than \(m+n\), this costs
  \(\bigO{\dispRk^{\expmm-1}\timepm{m+n} + \dispRk^\expmm
  \timepm{(m+n)/\dispRk} \log((m+n)/\dispRk)}\) operations in \(\field\). The
  superlinearity of \(\timepm{\cdot}\) implies that \(\dispRk^\expmm
  \timepm{(m+n)/\dispRk}\) is in \(\bigO{\dispRk^{\expmm-1}\timepm{m+n}}\)
  \cite[Sec.\,8.3]{ModernComputerAlgebra}, so computing \(\mK\) costs
  \(\bigO{\dispRk^{\expmm-1}\timepm{m+n}\log((m+n)/\dispRk)}\).
  Since the sum of the entries of \(\shift\) is less than \(m+n\), we have
  \(\cdeg_\shift(\mK) < m+n\) by \cite[Thm.\,3.6]{ZLS2012}. Since all entries of
  \(\shift\) are nonnegative, this implies that \(\deg(\vu) < m+n\) and
  \(\deg(\mu) < m+n\). Then, computing \(\mu^{-1} \vu \rem \polmod\) can be done
  using \(\bigO{\M(m+n) \log(m+n) + \dispRk \timepm{m+n}}\) operations in \(\field\) using
  \(\dispRk\) modular multiplications and one modular inversion
  \cite[Cor.\,11.11]{ModernComputerAlgebra}.
\end{proof}

\subsection{The Toeplitz-like case}
\label{sec:main:toeplitz}

Here, we describe our main algorithm for the Toeplitz-like case, and
prove the following result.

\begin{theorem}
  \label{thm:algo_toeplitz}%
  \zcref[S]{algo:StructuredSolve-Toeplitz} solves \zcref[S]{pbm:linsys,pbm:kernel} with
  \(\dispOp = \dispSte{\Zmat_{m,0}, \Zmat_{n,0}^{\trsp}}\), and under the
  assumptions on \(\timepm{\cdot}\) and \(\expmm\) stated in
  \zcref[S]{sec:notation}, it uses
  \[
    \bigO{\dispRk^{\expmm-1} (\timepm{m} \log(m) + \timepm{n} \log(n))}
  \]
  operations in \(\field\). It
  computes \((\ell, \vp, \vd, \shiftt, \vu)\) with
  \(\ell \in \{0,\ldots,\dispRk+1\}\), \(\vd = (d_1,\ldots,d_\ell) \in \N^\ell\), \(\shiftt =
  (t_1,\ldots,t_\ell) \in \N^\ell\), \(\vp = (p_1, \ldots, p_\ell) \in \xring^{\ell}\),
  and \(\vu\) is either \(\emptyset\) or in \(\matspace{n}{1}\),
  such that
  \begin{itemize}
    \item \(\vu\) is a linear system solution \(\mA \vu = \vv\) if one exists,
      otherwise \(\vu = \emptyset\);
    \item for all \(1 \le i \le \ell\), one has \(\deg(p_i) \le d_i \le n - t_i < n\);
    \item the nullspace \(\{\vz \in \matspace{n}{1} \mid \mA \vz = 0\}\) of
      \(\mA\) is the set of vectors \(\vz \in \matspace{n}{1}\)
      such that \(\vz\) is the coefficient vector of a polynomial in
      \[
        \left\{ \sum_{i = 1}^{\ell} x^{n-d_i-t_i} p_i {q_i} \;\Big|\; q_i \in \xring_{<t_i} \text{ for } 1 \le i \le \ell\right\}
        \subseteq \xring_{<n}.
      \]
  \end{itemize}
\end{theorem}

Note that here, if the system is homogeneous with \(\vv = 0\), the returned
system solution \(\vu\) may be zero; yet, one can easily deduce a nonzero
solution (if one exists) from the computed nullspace description.
We prove this theorem through \zcref[S]{lem:toeplitz_vmpade,
lem:toeplitz_transform, lem:toeplitz_smpade}  which describe the
properties and cost bounds for the three successive phases of the algorithm:
\zcref[range]{step:StructuredSolve-Toeplitz:start,
step:StructuredSolve-Toeplitz:pivdeg},
\zcref[range]{step:StructuredSolve-Toeplitz:transform,
step:StructuredSolve-Toeplitz:invert_h}, and
\zcref[range]{step:StructuredSolve-Toeplitz:start_smpade,
step:StructuredSolve-Toeplitz:final_ret}. The properties in these lemmas
directly give the ones stated in \zcref[S]{thm:algo_toeplitz}, and combining the
cost bounds in these lemmas yields
\[
  \bigOPar{\dispRk^\expmm \timepmPar{\frac{m}{\dispRk}} \log(m)
    + \dispRk^{\expmm-1} \timepm{m+n} \log\!\left(\frac{m+n}{\dispRk}\right)
    + \dispRk^{\expmm-1} \timepm{n} \log(n)}.
\]
Simplifying this cost bound, using the assumptions on
\(\timepm{\cdot}\) from \zcref[S]{sec:notation}, yields the one in
\zcref[S]{thm:algo_toeplitz}.

\begin{algorithm}[ht]
  \algoCaptionLabel{StructuredSolve-Toeplitz}{}
  \begin{algorithmic}[1]
    \Require{%
      \algoitem matrices $\mG \in \K^{m \times \dispRk}$ and $\mH \in \K^{n
      \times \dispRk}$, with \(\dispRk \le \min(m,n)\)
      (they represent a Toeplitz-like matrix \(\mA \in
      \matspace{m}{n}\) through the displacement operator
      $\dispSte{\Zmat_{m,0}, \Zmat_{n,0}^{\trsp}}(\mA) = \mG \mH^{\trsp}$); \par
      \algoitem a vector $\vv \in \K^{m \times 1}$ (optional, default: \(\vv = 0\)). \par
    }
    \Ensure{\((\ell, \vp, \vd, \shiftt, \vu)\) that solves \zcref[S]{pbm:linsys,pbm:kernel} as in \zcref[S]{thm:algo_toeplitz}}

    \LComment{build polynomials + vector Hermite-Pad\'e approximation (\zcref[S]{lem:toeplitz_vmpade})}
    \label{step:StructuredSolve-Toeplitz:start}

    \State $v \in \xring_{<m} \gets \sum_{1 \le i \le m} \vv_i x^{i-1}$
          \label{step:StructuredSolve-Toeplitz:poly_formulation_v}
    \For{$1 \le j \le \dispRk$}
      \State $g_j \in \xring_{<m} \gets \sum_{1 \le i \le m} \mG_{ij} x^{i-1}$
          \label{step:StructuredSolve-Toeplitz:poly_formulation_g}
      \State $\bar{h}_j \in \xring_{<n} \gets \sum_{1 \le i \le n} \mH_{ij} x^{i-1}$%
          \label{step:StructuredSolve-Toeplitz:poly_formulation_h}
    \EndFor

    \State $(\mP, \mu, \vsol) \gets
                  \Call{algo:VectorMPade}{m, \dispRk, x^m,
                                          [g_1 \;\; \cdots \;\; g_\dispRk],
                                          (0,\ldots,0), v}$
          \label{step:StructuredSolve-Toeplitz:vmpade}

    \State $\delta = (\delta_1,\ldots,\delta_\dispRk) \in \N^\dispRk \gets (\deg(\mP_{11}), \ldots, \deg(\mP_{\dispRk\dispRk}))$
        \Comment{\(\delta = \rdeg(\mP) = \cdeg(\mP)\)}
        \label{step:StructuredSolve-Toeplitz:pivdeg}

    \LComment{transform via reversal + inversion (\zcref[S]{lem:toeplitz_transform})}
         \label{step:StructuredSolve-Toeplitz:transform}

    \State $\bar\mP \in \xmatspace{\dispRk}{\dispRk} \gets \mP(x^{-1}) \diag{x^{\delta_1},\ldots,x^{\delta_\dispRk}}$
        \Comment{column reversal \(\crev{\delta}{\mP}\)}
        \label{step:StructuredSolve-Toeplitz:crevP}

    \State \(\vw \in \xmatspace{\dispRk}{1}_{<n} \gets
            \bar{\mP}^{-1} \polrev{n-1}{\vsol} \rem x^{n}\)
            if \(\deg(\vsol) < n\), and \(\vw \gets 0\) otherwise
           \label{step:StructuredSolve-Toeplitz:invert_sol}

    \State \(\mF \in \xmatspace{\dispRk}{1}_{<n} \gets
                  \bar{\mP}^{-1} [\bar{h}_1 \;\; \cdots \;\; \bar{h}_\dispRk]^{\trsp} \rem x^{n}\)
           \label{step:StructuredSolve-Toeplitz:invert_h}

    \LComment{simultaneous Hermite-Pad\'e approximation (\zcref[S]{lem:toeplitz_smpade})}
           \label{step:StructuredSolve-Toeplitz:start_smpade}

    \State \(\shift \in \N^\dispRk \gets (\max(0,n-\delta_1), \ldots, \max(0,n-\delta_\dispRk))\)

    \State $(\ell,\bar{\vp},\shiftt,\ssol[\vw]) \gets \Call{algo:SimultaneousMPade}{%
                                                             n, \dispRk, x^n,
                                                             \mF, \shift, \vw}$
           \label{step:StructuredSolve-Toeplitz:smpade}

    \If{$\mu \neq 1$ \OR $\deg(\vsol) \ge n$ \OR $\ssol[\vw] = \emptyset$}
      \Comment{no solution to nonhomogeneous system}
      \State $\vu \gets \emptyset$
       \label{step:StructuredSolve-Toeplitz:nosol}
    \Else
      \State$\vu \in \matspace{n}{1} \gets [u_1 \;\; u_2 \;\; \cdots \;\; u_n]^{\trsp}$,
      where \(\ssol[\vw] = u_n + u_{n-1} x + \cdots + u_1x^{n-1}\)
       \label{step:StructuredSolve-Toeplitz:sol}
    \EndIf

    \State \(\vp \in \xmatspace{1}{\ell}_{<n} \gets \crev{\vd}{\bar{\vp}} = \bar{\vp}(1/x) \diag{x^{d_1},\ldots,x^{d_\ell}}\),
    where \(\vd = (d_1, \ldots, d_\ell) \gets \cdeg(\bar{\vp})\)
       \label{step:StructuredSolve-Toeplitz:rev_p}

    \State \Return $(\ell, \vp, \vd, \shiftt, \vu)$
    \label{step:StructuredSolve-Toeplitz:final_ret}
  \end{algorithmic}
\end{algorithm}

\begin{lemma}
  \label{lem:toeplitz_vmpade}%
  \zcref[range]{step:StructuredSolve-Toeplitz:start,
  step:StructuredSolve-Toeplitz:pivdeg} of \zcref[S]{algo:StructuredSolve-Toeplitz}
  use \(\bigO{\dispRk^\expmm \timepm{m/\dispRk} \log(m)}\) operations in
  \(\field\) and compute \(v, g_j, \bar{h}_j, \mP, \mu, \vsol, \delta\) such that,
  for any vector \(\vu \in \matspace{n}{1}\),
  \begin{itemize}
    \item \(\mA \vu = \vv\) if and only if \(\mu = 1\) and
      \(\vc = \mP \lambda + \vsol\) for some \(\lambda \in \xmatspace{\dispRk}{1}\);

    \item \(\mA \vu = 0\) if and only if \(\vc = \mP \lambda\) for some
      \(\lambda \in \xmatspace{\dispRk}{1}\);
  \end{itemize}
  where we define
  the vector
  \[
    \vc =
    \begin{bmatrix}
      c_1 \\ \vdots \\ c_\dispRk
    \end{bmatrix}
    =
    \begin{bmatrix}
      h_1 \\ \vdots \\ h_\dispRk
    \end{bmatrix}
    u
    \quo x^{n-1}
  \]
  in \(\xmatspace{\dispRk}{1}_{<n}\), where \(h_j = \polrev{n-1}{\bar{h}_j}\)
  and \(u \in \xring_{<n}\) is the polynomial with coefficient vector \(\vu\).
\end{lemma}

\begin{proof}
  The first steps at \zcref[range]{step:StructuredSolve-Toeplitz:poly_formulation_g,
    step:StructuredSolve-Toeplitz:poly_formulation_h,
  step:StructuredSolve-Toeplitz:poly_formulation_v}
  construct polynomials from \(\mG\), \(\mH\), and \(\vv\), without any arithmetic operation: \(g_j\) and \(v\) are as in
  \zcref[S]{thm:linsys_poly:toeplitz}, and \(\bar{h}_j = \polrev{n-1}{h_j}\) for
  \(h_j\) as in \zcref[S]{thm:linsys_poly:toeplitz}.
  According to \zcref[S]{thm:linsys_poly:toeplitz},
  \(\mA \vu = \vv\) is equivalent to
  \begin{equation}
    \label{eqn:eq_poly_form}%
    v =
    \begin{bmatrix} g_1 & \cdots & g_\dispRk \end{bmatrix}
    \vc
    \rem x^m,
  \end{equation}
  for \(\vc = [c_j]_j \in \xmatspace{\dispRk}{1}_{<n}\) as in the lemma. The degree bound
  comes from the fact that \(\deg(c_j) = \deg(h_j u \quo x^{n-1}) < n\) for all \(j\).

  According to \zcref[S]{prop:vmpade},
  \zcref[S]{step:StructuredSolve-Toeplitz:vmpade,
  step:StructuredSolve-Toeplitz:pivdeg} use \(\bigO{\dispRk^\expmm
  \timepm{m/\dispRk} \log(m)}\) and compute \(\mP\) in Popov form (for the
  shift \((0,\ldots,0)\)) and \(\vsol\) such that \(\rdeg(\vsol) < \delta =
  \cdeg(\mP) = \rdeg(\mP)\). The latter equalities and the fact that
  \(\rdeg(\mP)\) coincides with the diagonal degrees \(\delta\) (as computed in
  \zcref[S]{step:StructuredSolve-Toeplitz:pivdeg}) follow from the definition of
  Popov forms. Note that \zcref[S]{prop:vmpade} also ensures \(\delta_1 + \cdots +
  \delta_\dispRk \le m\).
  Using \zcref[S]{lem:vmpade_generates_all}, from \zcref[S]{eqn:eq_poly_form} we get
  the stated equivalence
  \[
    \mA \vu = \vv
    \quad \Leftrightarrow\quad
    \mu = 1 \text{ and }
    \vc
    =
    \mP \lambda + \vsol
    \text{ for some } \lambda \in \xmatspace{\dispRk}{1}
    .
  \]
  All our reasoning above also applies to finding the second equivalence,
  concerning \(\mA \vu = 0\), by considering \(\vv=0\). In this case,
  \(\mP\) is unchanged (it does not depend on \(\vv\)), and we
  have \(v = 0\), \(\mu = 1\), and \(\vsol = 0\), hence the claimed
  equivalence.
\end{proof}

\begin{lemma}
  \label{lem:toeplitz_transform}%
  We follow on with notation from \zcref[S]{lem:toeplitz_vmpade} and
  \zcref[S]{algo:StructuredSolve-Toeplitz}. Define the integer tuple
 \(\shift = (\max(0,n-\delta_1), \ldots,
  \max(0,n-\delta_\dispRk)) \in \N^\dispRk\);
  note that for \(\lambda \in \xmatspace{\dispRk}{1}\), the constraint
  \(\rdeg(\lambda) < \shift\) is equivalent to
  \(\rdeg(\lambda) < n - \delta\): when \(\delta_i \ge n\), both mean
  \(\lambda_i = 0\).
  \zcref[range]{step:StructuredSolve-Toeplitz:transform,
  step:StructuredSolve-Toeplitz:invert_h} of this algorithm use
  \(\bigO{\dispRk^{\expmm-1} \timepm{m+n} \log((m+n)/\dispRk)}\) operations in
  \(\field\) to
  compute two vectors \(\mF\) and \(\vw\) in \(\xmatspace{\dispRk}{1}_{<n}\)
  such that, for any \(\vu \in \matspace{n}{1}\),
  \begin{itemize}
    \item
    \(\mA \vu = \vv\)
    if and only if
    \(\mu = 1\) and
    \(\deg(\vsol) < n\) and
    \(\rdeg((\mF \bar{u} - \vw) \rem x^n) < \shift\);
    \item
    \(\mA \vu = 0\)
    if and only if
    \(\rdeg(\mF \bar{u} \rem x^n) < \shift\);
  \end{itemize}
  where \(\bar{u} \in \xring_{<n}\) is the polynomial with coefficient vector
  \(\vecrev{\vu}\).
\end{lemma}

\begin{proof}
  The equivalence of \zcref[S]{lem:toeplitz_vmpade} can be augmented with degree
  bounds and then reversed with respect to the maximum possible degree of the
  right- and left-hand sides, as follows:
  \begin{align*}
    \mA \vu = \vv
    \;\Leftrightarrow&\;
    \mu = 1 \;\text{and}\;
    \vc
    =
    \mP \lambda + \vsol
    \;\text{for some}\; \lambda \in \xmatspace{\dispRk}{1}
    \\
    \;\Leftrightarrow&\;
    \mu = 1 \text{ and }
    \vc
    =
    \mP \lambda + \vsol
    \;\text{for some}\; \lambda \in \xmatspace{\dispRk}{1}
    \;\text{with}\; \rdeg(\lambda) < \shift
    \\
    \;\Leftrightarrow&\;
    \mu = 1 \;\text{and}\; \deg(\vsol) < n \;\text{and}\;
    \polrev{n-1}{\vc}
    =
    \bar{\mP} \bar{\lambda} + \polrev{n-1}{\vsol} \\
    & \;\text{for some}\; \bar{\lambda} \in \xmatspace{\dispRk}{1}
    \;\text{with}\; \rdeg(\bar{\lambda}) < \shift.
  \end{align*}
  The first equivalence is from \zcref[S]{lem:toeplitz_vmpade}.
  The implication \(\Leftarrow\) in the second equivalence is obvious, while
  the implication \(\Rightarrow\) follows from \zcref[S]{lem:degree_bound_lambda},
  which ensures that \(\cdeg(\mP\lambda) < n\), \(\cdeg(\vsol) < n\), and
  \(\rdeg(\lambda) < n - \delta\) hence \(\rdeg(\lambda) < \shift\). Indeed,
  its assumptions are satisfied: by construction \(\mP\) is in Popov form with
  \(\delta = \cdeg(\mP)\) and we have \(\rdeg(\vsol) < \delta\) and
  \(\cdeg(\mP\lambda + \vsol) = \cdeg(\vc) < n\).
  For the implication \(\Rightarrow\) of the third equivalence, we apply
  \zcref[S]{lem:reverse_product} with \(\mQ = \lambda\), which gives
  \(\polrev{n-1}{\mP\lambda} = \bar{\mP} \bar{\lambda}\), for \(\bar{\mP} =
  \crev{\delta}{\mP}\) as in \zcref[S]{step:StructuredSolve-Toeplitz:crevP}, and
  \(\bar{\lambda} = \rrev{\shift-1}{\lambda}\) which is such that
  \(\rdeg(\bar{\lambda}) < \shift\). Its other direction \(\Leftarrow\) is
  proved similarly: using \zcref[S]{lem:reverse_product}, we also get
  \(\polrev{n-1}{\bar{\mP} \bar{\lambda}} = \mP\lambda\).

  By construction, the constant matrix \(\bar{\mP}(0)\) is equal to the
  leading matrix \(\clm(\mP)\), which is invertible since \(\mP\) is in reduced form.
  Thus \(\det(\bar{\mP})\) is coprime with \(x\), meaning that \(\bar{\mP}\) is
  invertible modulo \(x^n\). Since we restrict to \(\bar{\lambda}\)
  with \(\rdeg(\bar{\lambda}) < \shift\), and since \(\cdeg(\bar{\mP}) \le
  \delta\), both \(\bar{\lambda}\) and \(\bar{\mP}\bar{\lambda}\) have degree
  less than \(n\). Thus, we have the equivalence
  \[
    \polrev{n-1}{\vc} = \bar{\mP} \bar{\lambda} + \polrev{n-1}{\vsol}
    \quad\Leftrightarrow\quad
    \bar{\mP}^{-1} (\polrev{n-1}{\vc} - \polrev{n-1}{\vsol}) \rem x^n = \bar{\lambda}.
  \]
  We have proved that
  \begin{align}
    \label{eqn:equiv_reversed}%
    \mA \vu = \vv
    \Leftrightarrow\;
    & \mu = 1
    \;\text{and}\;
    \deg(\vsol) < n
    \text{ and} \\
    & \rdeg(\bar{\mP}^{-1} (\polrev{n-1}{\vc} - \polrev{n-1}{\vsol}) \rem x^n)
    < \shift.
    \nonumber
  \end{align}

  If \(\deg(\vsol) \ge n\), both sides of the equivalence in the first item
  of the lemma are false whatever \(\vw\) is (the left-hand side by
  \zcref[S]{eqn:equiv_reversed}), so the vector \(\vw = 0\) computed at
  \zcref[S]{step:StructuredSolve-Toeplitz:invert_sol} is suitable. Otherwise,
  \(\polrev{n-1}{\vsol}\) is in \(\xmatspace{\dispRk}{1}_{<n}\), and computing
  the vector \(\vw = \bar{\mP}^{-1}\polrev{n-1}{\vsol} \rem x^n\) at
  \zcref[S]{step:StructuredSolve-Toeplitz:invert_sol} uses
  \(\bigO{\dispRk^{\expmm-1} \timepm{m+n} \log((m+n)/\dispRk)}\) operations
  in \(\field\), by \zcref[S]{lem:Pinv_vec_mod} since \(\delta_1+\cdots+\delta_\dispRk
  \le m\). The rest of the proof of the first item concerns this case
  \(\deg(\vsol) < n\).
  The computation of \(\mF\in \xmatspace{\dispRk}{1}_{<n}\) at
  \zcref[S]{step:StructuredSolve-Toeplitz:invert_h} has the same cost bound,
  and this vector satisfies
  \[
    \mF \bar{u} \rem x^{n}
    =
      \bar{\mP}^{-1} [\bar{h}_1 \;\; \cdots \;\; \bar{h}_\dispRk]^{\trsp}
      \bar{u}
      \rem x^n
    = \bar{\mP}^{-1}\polrev{n-1}{\vc} \rem x^n,
  \]
  where \(\bar{u} = \polrev{n-1}{u} \in \xring_{<n}\),
  and where \(\bar{h}_j = \polrev{n-1}{h_j}\) is computed at
  \zcref[S]{step:StructuredSolve-Toeplitz:poly_formulation_h}. The last
  equality uses \(\polrev{n-1}{c_j} = \bar{h}_j \bar{u} \rem x^n\), obtained by
  reversing the identity \(h_j u = c_j x^{n-1} + (h_j u \rem x^{n-1})\) with
  respect to the degree bound \(2n-2\).

  Substituting in \zcref[S]{eqn:equiv_reversed}, we obtain
  \[
    \mA \vu = \vv
    \quad\Leftrightarrow\quad
    \mu = 1
    \;\text{and}\;
    \deg(\vsol) < n
    \;\text{and}\;
    \rdeg((\mF \bar{u} - \vw) \rem x^{n}) < \shift.
  \]
  This concludes the proof of the equivalence in the first item of the lemma.

  For the second item, as in the proof of \zcref[S]{lem:toeplitz_vmpade}, all the above arguments
  apply to the homogeneous case \(\mA \vu = 0\), by considering
  \(\vv=0\). In this case, \(\mP\) is unchanged (it does not depend on
  \(\vv\)), and we have \(\mu = 1\), \(\vsol = 0\), and \(\vw = 0\), hence the
  above equivalence simplifies into
  \(\mA \vu = 0 \Leftrightarrow \rdeg(\mF \bar{u} \rem x^n) < \shift\).
\end{proof}

\begin{lemma}
  \label{lem:toeplitz_smpade}%
  We follow on with notation from
  \zcref[S]{lem:toeplitz_vmpade,lem:toeplitz_transform} and
  \zcref[S]{algo:StructuredSolve-Toeplitz}.
  \zcref[range]{step:StructuredSolve-Toeplitz:start_smpade,
  step:StructuredSolve-Toeplitz:final_ret} of this algorithm use
  \(\bigO{\dispRk^{\expmm-1} \timepm{n} \log(n)}\) operations in \(\field\),
  and compute $(\ell, \vp, \vd, \shiftt, \vu)$ such that
  \begin{itemize}
    \item \(\vu \in \matspace{n}{1}\) is a linear system solution \(\mA \vu = \vv\)
      if one exists, otherwise \(\vu = \emptyset\);
    \item \(\ell \in \{0,\ldots,\dispRk+1\}\),
      \(\vd \in \{0,\ldots,n-1\}^\ell\),
      and \(\shiftt \in \{1,\ldots,n\}^\ell\),
      with \(\max(\vd + \shiftt) \le n\);
    \item \(\vp \in \xmatspace{1}{\ell}\) satisfies \(\cdeg(\vp) \le \vd\);
    \item the nullspace \(\{\vz \in \matspace{n}{1} \mid \mA \vz = 0\}\) of
      \(\mA\) is the set of vectors \(\vz \in \matspace{n}{1}\)
      such that \(\vz\) is the coefficient vector of a polynomial in
      \(
        \{ \vp \mX \vq \mid \vq \in \xmatspace{\ell}{1} \text{ and } \rdeg(\vq) < \shiftt\}
      \),
      where we define the diagonal matrix \(\mX = \diag{x^{n-d_1-t_1},\ldots,x^{n-d_\ell-t_\ell}} \in \xmatspace{\ell}{\ell}\).
  \end{itemize}
\end{lemma}

\begin{proof}
  The only step that involves arithmetic operations is
  \zcref[S]{step:StructuredSolve-Toeplitz:smpade}, which costs
  \(\bigO{\dispRk^{\expmm-1} \timepm{n} \log(n)}\) operations in \(\field\)
  by \zcref[S]{prop:smpade}. Furthermore,
  \zcref[S]{step:StructuredSolve-Toeplitz:smpade} computes \((\ell,\bar{\vp}, \shiftt,
  \ssol[\vw])\) which solves \zcref[S]{pbm:smpade}.

  In particular, \(\ssol[\vw]\) is either a polynomial in \(\xring_{<n}\) such
  that \(\rdeg((\mF \ssol[\vw] - \vw) \rem x^n) < \shift\) if one exists,
  otherwise it is \(\ssol[\vw] = \emptyset\). Thus,
  \zcref[S]{step:StructuredSolve-Toeplitz:nosol} correctly sets \(\vu = \emptyset\)
  precisely when the linear system defined by \(\mA\) and \(\vv\) admits no
  solution \(\vu\), according to \zcref[S]{lem:toeplitz_transform}. On the
  other hand, if \(\ssol[\vw] \neq \emptyset\), and if the other requirements
  for the existence of a solution \(\mu = 1\) and  \(\deg(\vsol) < n\) also
  hold, then again according to \zcref[S]{lem:toeplitz_transform},
  \zcref[S]{step:StructuredSolve-Toeplitz:sol} correctly sets \(\vu\) to the
  solution \(\mA\vu = \vv\)  formed by the reversed coefficient vector of
  \(\ssol[\vw]\).

  This shows the first item in the lemma. Except for the bound on \(\vd +
  \shiftt\), the second item follows from the definition of a solution
  basis \((\ell,\bar{\vp}, \shiftt)\) for \((x^n, \mF, \shift)\) and from
  \zcref[S]{lem:smpade_solution_space}. Since \zcref[S]{step:StructuredSolve-Toeplitz:rev_p}
  computes \(\vp = \crev{\vd}{\bar{\vp}}\) where \(\vd = \cdeg(\bar{\vp})\), we
  deduce \(\cdeg(\vp) \le \vd\) as in the third item. The third item of
  \zcref[S]{def:smpade_basis} implies in particular that \(\cdeg_{-n}(\bar{\vp})
  \le -\shiftt\). Since \(\cdeg_{-n}(\bar{\vp}) = \cdeg(\bar{\vp}) - (n,\ldots,n)
  = \vd - (n,\ldots,n)\), this gives \(\max(\vd + \shiftt) \le n\).

  It remains to show the last item, about the nullspace of \(\mA\). By the
  second item in \zcref[S]{lem:toeplitz_transform}, this nullspace is the set
  of vectors \(\vz \in \matspace{n}{1}\) such that \(\vz\) is the coefficient
  vector of \(\polrev{n-1}{\bar{u}}\), for any polynomial \(\bar{u} \in
  \xring_{<n}\)  such that  \(\rdeg(\mF \bar{u} \rem x^n) < \shift\). On the
  other hand, \zcref[S]{lem:smpade_solution_space} states that
  \[
    \{\bar{u} \in \xring_{<n} \mid \rdeg(\mF \bar{u} \rem x^n) < \shift\}
    =
    \{\bar{\vp} \bar{\vq} \mid \bar{\vq} \in \xmatspace{\ell}{1}, \rdeg(\bar{\vq}) < \shiftt\}.
  \]
  Therefore, to conclude the proof, it remains to show that
  \[
    \{\polrev{n-1}{\bar{\vp} \bar{\vq}} \mid \bar{\vq} \in \xmatspace{\ell}{1}, \rdeg(\bar{\vq}) < \shiftt\}
    = \{ \vp \mX \vq \mid \vq \in \xmatspace{\ell}{1} \text{ and } \rdeg(\vq) < \shiftt\},
  \]
  where \(\mX \in \xmatspace{\ell}{\ell}\) is the diagonal matrix as in the
  lemma. This set identity is easily deduced from definitions and
  \zcref[S]{lem:reverse_product}. Indeed, since \(\vd = \cdeg(\bar{\vp})\), for any
  \(\bar{\vq} \in \xmatspace{\ell}{1}\) such that \(\rdeg(\bar{\vq}) <
  \shiftt\) (and thus \(\rdeg(\bar{\vq}) < n - \vd\)), we can apply
  \zcref[S]{lem:reverse_product} which yields \(\polrev{n-1}{\bar{\vp} \bar{\vq}} =
  \crev{\vd}{\bar{\vp}} \rrev{n - 1 - \vd}{\bar{\vq}}\). We have
  \(\crev{\vd}{\bar{\vp}} = \vp\) by definition of \(\vp\). Furthermore,
  defining \(\vq = \rrev{\shiftt-1}{\bar{\vq}}\), which is in
  \(\xmatspace{\ell}{1}\) and such that \(\rdeg(\vq) < \shiftt\), we have
  \(\rrev{n - 1 - \vd}{\bar{\vq}} = \mX \vq\).
\end{proof}

\subsection{The Vandermonde-like case}
\label{sec:main:vandermonde}

Now, we derive similar results about the Vandermonde-like case. To keep
the presentation concise, we focus on points that differ from those
for the Toeplitz-like case in \zcref[S]{sec:main:toeplitz}.

\begin{theorem}
  \label{thm:algo_vandermonde}%
  The exact same statement as in \zcref[S]{thm:algo_toeplitz} holds for
  \zcref[S]{algo:StructuredSolve-Vandermonde}, for solving
  \zcref[S]{pbm:linsys,pbm:kernel} with \(\dispOp = \dispSte{\D(\vx),
  \Zmat_{n,0}^{\trsp}}\).
\end{theorem}

The proximity of the equations in
\zcref[S]{thm:linsys_poly:toeplitz,thm:linsys_poly:vandermonde} leads to
computations that are very close for the Toeplitz-like and Vandermonde-like
cases. More specifically, the initial translation of the input data into a
vector M-Pad\'e approximation problem marginally differs: in the
Vandermonde-like situation, we interpolate data into an equation modulo
\(\mpol_{\vx}\), whereas in the Toeplitz-like case, we used coefficient vectors
to get an equation modulo \(x^m\). We detail this in the first lines of
\zcref[S]{algo:StructuredSolve-Vandermonde}. Then, all the subsequent computations
are strictly the same as in \zcref[S]{algo:StructuredSolve-Toeplitz}; to emphasize
this, the corresponding part of \zcref[S]{algo:StructuredSolve-Vandermonde} does not repeat the
instructions but simply refers directly to \zcref[S]{algo:StructuredSolve-Toeplitz}.

Furthermore, our choice of consistent notations makes
\zcref[S]{lem:toeplitz_vmpade} hold as such for
\zcref[range]{step:StructuredSolve-Vandermonde:start,
step:StructuredSolve-Vandermonde:pivdeg} of
\zcref[S]{algo:StructuredSolve-Vandermonde}; concerning its proof, the only change
is about the cost of computing polynomials from \(\mG\), \(\mH\), and \(\vv\) at
\zcref[range]{step:StructuredSolve-Vandermonde:poly_formulation_v,
  step:StructuredSolve-Vandermonde:poly_formulation_g,
step:StructuredSolve-Vandermonde:poly_formulation_h}, which is now in
\(\bigO{\dispRk \timepm{m} \log(m)}\) operations in \(\field\) via fast
interpolation \cite[Cor.\,10.12]{ModernComputerAlgebra} (instead of
\(\bigO{1}\) for the Toeplitz-like case). Then, we use the same steps as in
\zcref[S]{algo:StructuredSolve-Toeplitz} to compute and solve an instance of
simultaneous Hermite-Pad\'e approximation and all properties in
\zcref[S]{lem:toeplitz_transform, lem:toeplitz_smpade} remain valid.

\begin{algorithm}[ht]
  \algoCaptionLabel{StructuredSolve-Vandermonde}{}
  \begin{algorithmic}[1]
    \Require{%
      \algoitem repetition-free list $\vx = (x_1,\ldots,x_m) \in \K^m$; \par
      \algoitem matrices $\mG \in \K^{m \times \dispRk}$ and $\mH \in \K^{n
      \times \dispRk}$, with \(\dispRk \le \min(m,n)\)
      (they represent a Vandermonde-like matrix \(\mA \in
      \matspace{m}{n}\) through the operator
      \(\dispSte{\D(\vx), \Zmat_{n,0}^{\trsp}}(\mA) = \mG \mH^{\trsp}\)); \par
      \algoitem a vector $\vv \in \K^{m \times 1}$ (optional, default: \(\vv = 0\)). \par
    }
    \Ensure{\((\ell, \vp, \vd, \shiftt, \vu)\) that solves \zcref[S]{pbm:linsys,pbm:kernel} as in \zcref[S]{thm:algo_vandermonde}}

    \LComment{build polynomials and perform vector M-Pad\'e approximation}
    \label{step:StructuredSolve-Vandermonde:start}

    \State $v \in \xring_{<m} \gets$ the interpolant of \((\vx,\vv)\);
           ~ \(\mpol_{\vx} \in \xring_{m} \gets \prod_{1\le i \le m} (x-x_i)\)
          \label{step:StructuredSolve-Vandermonde:poly_formulation_v}
    \For{$1 \le j \le \dispRk$}
      \State $g_j \in \xring_{<m} \gets$ the interpolant of \((\vx,\mG_{*j})\)
          \label{step:StructuredSolve-Vandermonde:poly_formulation_g}
      \State $\bar{h}_j \in \xring_{<n} \gets \sum_{1 \le i \le n} \mH_{ij} x^{i-1}$%
          \label{step:StructuredSolve-Vandermonde:poly_formulation_h}
    \EndFor

    \State $(\mP, \mu, \vsol) \gets
                  \Call{algo:VectorMPade}{m, \dispRk, \mpol_{\vx},
                                          [g_1 \;\; \cdots \;\; g_\dispRk],
                                          (0,\ldots,0), v}$
          \label{step:StructuredSolve-Vandermonde:vmpade}

    \State $\delta = (\delta_1,\ldots,\delta_\dispRk) \in \N^\dispRk \gets (\deg(\mP_{11}), \ldots, \deg(\mP_{\dispRk\dispRk}))$
        \Comment{\(\delta = \rdeg(\mP) = \cdeg(\mP)\)}
        \label{step:StructuredSolve-Vandermonde:pivdeg}

    \LComment{transform into simultaneous Hermite-Pad\'e via reversal + inversion: \\
    follow \zcref[range]{step:StructuredSolve-Toeplitz:transform,step:StructuredSolve-Toeplitz:invert_h} of \zcref[S]{algo:StructuredSolve-Toeplitz} without any modification}

    \LComment{perform simultaneous Hermite-Pad\'e approximation: \\
      follow \zcref[range]{step:StructuredSolve-Toeplitz:start_smpade,step:StructuredSolve-Toeplitz:rev_p} of \zcref[S]{algo:StructuredSolve-Toeplitz} without any modification
    }

    \State \Return $(\ell, \vp, \vd, \shiftt, \vu)$
    \label{step:StructuredSolve-Vandermonde:final_ret}
  \end{algorithmic}
\end{algorithm}

\subsection{The Cauchy-like case}
\label{sec:main:cauchy}

Finally, we describe our algorithm for the Cauchy-like case. Because we follow
the same global approach, most aspects are close to the Toeplitz-like case in
\zcref[S]{sec:main:toeplitz}. Still, for completeness and because of some minor
differences, we provide detailed statements and proofs:
\zcref[S]{thm:algo_cauchy,lem:cauchy_vmpade,lem:cauchy_transform,lem:cauchy_smpade}
are the Cauchy-like twins of
\zcref[S]{thm:algo_toeplitz,lem:toeplitz_vmpade,lem:toeplitz_transform,lem:toeplitz_smpade}.

\begin{theorem}
  \label{thm:algo_cauchy}%
  \zcref[S]{algo:StructuredSolve-Cauchy} solves \zcref[S]{pbm:linsys,pbm:kernel} when
  \(\dispOp = \dispSyl{\D(\vx), \D(\vy)}\), and under the
  assumptions on \(\timepm{\cdot}\) and \(\expmm\) stated in \zcref[S]{sec:notation},
  it uses
  \[
    \bigO{\dispRk^{\expmm-1} (\timepm{m} \log(m) + \timepm{n} \log(n))}
  \]
  operations in \(\field\). It computes \((\ell, \vp, \vd, \shiftt, \vu)\) with
  \(\ell \in \{0,\ldots,\dispRk+1\}\), \(\vd = (d_1,\ldots,d_\ell) \in \N^\ell\), \(\shiftt =
  (t_1,\ldots,t_\ell) \in \N^\ell\), \(\vp = (p_1, \ldots, p_\ell) \in \xring^{\ell}\),
  and \(\vu\) is either \(\emptyset\) or in \(\matspace{n}{1}\),
  such that
  \begin{itemize}
    \item \(\vu\) is a linear system solution \(\mA \vu = \vv\) if one exists,
      otherwise \(\vu = \emptyset\);
    \item for all \(1 \le i \le \ell\), one has \(\deg(p_i) \le d_i \le n - t_i < n\);
    \item the nullspace \(\{\vz \in \matspace{n}{1} \mid \mA \vz = 0\}\) of
      \(\mA\) is the set of evaluation vectors
      \[
        \left\{
          u(\vy) \in \matspace{n}{1}
          \;\Big|\;
          u = \sum_{i = 1}^{\ell} p_i q_i \in \xring_{<n}, \;
          q_i \in \xring_{<t_i} \text{ for } 1 \le i \le \ell\right\}
          .
      \]
  \end{itemize}
\end{theorem}

Similarly to the Toeplitz-like case, the returned solution \(\vu\) may be zero
for a homogeneous system; yet, a nonzero one can easily be deduced (if one
exists) from the nullspace description. We prove this
theorem through \zcref[S]{lem:cauchy_vmpade, lem:cauchy_transform,
lem:cauchy_smpade} which
describe the properties and cost bounds for the three
successive phases of the algorithm:
\zcref[range]{step:StructuredSolve-Cauchy:start,
step:StructuredSolve-Cauchy:pivdeg},
\zcref[range]{step:StructuredSolve-Cauchy:transform,
step:StructuredSolve-Cauchy:invert_h}, and
\zcref[range]{step:StructuredSolve-Cauchy:start_smpade,
step:StructuredSolve-Cauchy:final_ret}. The properties in these lemmas
directly give the ones stated in \zcref[S]{thm:algo_cauchy}, and combining the cost
bounds in these lemmas leads to the cost bound in \zcref[S]{thm:algo_cauchy}.

\begin{algorithm}[ht]
  \algoCaptionLabel{StructuredSolve-Cauchy}{}
  \begin{algorithmic}[1]
    \Require{%
      \algoitem repetition-free lists $\vx = (x_1,\ldots,x_m) \in \K^m$ and $\vy = (y_1,\ldots,y_n) \in \K^n$ which are disjoint; \par
      \algoitem matrices $\mG \in \K^{m \times \dispRk}$ and $\mH \in \K^{n
      \times \dispRk}$, with \(\dispRk \le \min(m,n)\)
      (they represent a Cauchy-like matrix \(\mA \in
      \matspace{m}{n}\) through the displacement operator
      \(\dispSyl{\D(\vx), \D(\vy)}(\mA) = \mG \mH^{\trsp}\)); \par

      \algoitem a vector $\vv \in \K^{m \times 1}$ (optional, default: \(\vv = 0\)). \par
    }
    \Ensure{\((\ell, \vp, \vd, \shiftt, \vu)\) that solves \zcref[S]{pbm:linsys,pbm:kernel} as in \zcref[S]{thm:algo_cauchy}}

    \LComment{build polynomials and perform vector M-Pad\'e approximation (\zcref[S]{lem:cauchy_vmpade})}
    \label{step:StructuredSolve-Cauchy:start}

    \State \(\mpol_{\vx} \in \xring_{m} \gets \prod_{1\le i \le m} (x-x_i)\);
         ~ \(\mpol_{\vy} \in \xring_{n} \gets \prod_{1\le i \le n} (x-y_i)\)
          \label{step:StructuredSolve-Cauchy:mpol}

   \State \(\invmpol_{\vy} \in \xring_{<m} \gets \mpol_{\vy}^{-1} \rem \mpol_{\vx}\);
      ~ \(\mpol'_{\vy} \in \xring_{<n} \gets\) derivative of \(\mpol_{\vy}\)
          \label{step:StructuredSolve-Cauchy:invmpol}

    \State $v \in \xring_{<m} \gets$ the interpolant of \((\vx,\vv)\)
          \label{step:StructuredSolve-Cauchy:poly_formulation_v}
    \For{$1 \le j \le \dispRk$}
      \State $\bar{g}_j \in \xring_{<m} \gets$ the interpolant of \((\vx, \mG_{*j})\);
      ~ \(g_j \in \xring_{<m} \gets \invmpol_{\vy} \bar{g}_j \rem \mpol_\vx\)
          \label{step:StructuredSolve-Cauchy:poly_formulation_g}
      \State $\bar{h}_j \in \xring_{<n} \gets$ the interpolant of \((\vy, \mH_{*j})\);
      ~ \(h_j \in \xring_{<n} \gets \mpol'_{\vy} \bar{h}_j \rem \mpol_\vy\)
          \label{step:StructuredSolve-Cauchy:poly_formulation_h}
    \EndFor

    \State $(\mP, \mu, \vsol) \gets
                  \Call{algo:VectorMPade}{m, \dispRk, \mpol_{\vx},
                                          [g_1 \;\; \cdots \;\; g_\dispRk],
                                          (0,\ldots,0), v}$
          \label{step:StructuredSolve-Cauchy:vmpade}

    \State $\delta = (\delta_1,\ldots,\delta_\dispRk) \in \N^\dispRk \gets (\deg(\mP_{11}), \ldots, \deg(\mP_{\dispRk\dispRk}))$
        \Comment{\(\delta = \rdeg(\mP) = \cdeg(\mP)\)}
        \label{step:StructuredSolve-Cauchy:pivdeg}

    \LComment{transform into simultaneous M-Pad\'e via inversion (\zcref[S]{lem:cauchy_transform})}
         \label{step:StructuredSolve-Cauchy:transform}

    \State \(\vw \in \xmatspace{\dispRk}{1}_{<n} \gets \mP^{-1} \vsol \rem \mpol_{\vy}\)
           if \(\deg(\vsol) < n\), and \(\vw \gets 0\) otherwise
           \label{step:StructuredSolve-Cauchy:invert_sol}

    \State \(\mF \in \xmatspace{\dispRk}{1}_{<n} \gets
              \mP^{-1} [h_1 \;\; \cdots \;\; h_\dispRk]^{\trsp} \rem \mpol_{\vy}\)
           \label{step:StructuredSolve-Cauchy:invert_h}

    \LComment{perform simultaneous M-Pad\'e approximation (\zcref[S]{lem:cauchy_smpade})}
           \label{step:StructuredSolve-Cauchy:start_smpade}

    \State \(\shift \in \N^\dispRk \gets (\max(0,n-\delta_1), \ldots, \max(0,n-\delta_\dispRk))\)

    \State $(\ell,\vp,\shiftt,\ssol[\vw]) \gets \Call{algo:SimultaneousMPade}{%
                                                             n, \dispRk, \mpol_\vy,
                                                             \mF, \shift, \vw}$
           \label{step:StructuredSolve-Cauchy:smpade}

    \If{$\mu \neq 1$ \OR $\deg(\vsol) \ge n$ \OR $\ssol[\vw] = \emptyset$}
      \Comment{no solution to nonhomogeneous system}
      \State $\vu \gets \emptyset$
       \label{step:StructuredSolve-Cauchy:nosol}
    \Else
      \State$\vu \in \matspace{n}{1} \gets \ssol[\vw](\vy) = [\ssol[\vw](y_1) \;\; \ssol[\vw](y_2) \;\; \cdots \;\; \ssol[\vw](y_n)]^{\trsp}$
       \label{step:StructuredSolve-Cauchy:sol}
    \EndIf

    \State \(\vd = (d_1, \ldots, d_\ell) \gets \cdeg(\vp)\)
       \label{step:StructuredSolve-Cauchy:rev_p}

    \State \Return $(\ell, \vp, \vd, \shiftt, \vu)$
    \label{step:StructuredSolve-Cauchy:final_ret}
  \end{algorithmic}
\end{algorithm}

\begin{lemma}
  \label{lem:cauchy_vmpade}%
  Using \(\bigO{\dispRk^\expmm \timepm{m/\dispRk} \log(m) + \dispRk \timepm{n}
  \log(n)}\) operations, \zcref[range]{step:StructuredSolve-Cauchy:start,
  step:StructuredSolve-Cauchy:pivdeg} of \zcref[S]{algo:StructuredSolve-Cauchy}
  compute \(v, g_j, h_j, \mP, \mu, \vsol, \delta\) such that, for any vector \(\vu \in \matspace{n}{1}\),
  \begin{itemize}
    \item \(\mA \vu = \vv\) if and only if \(\mu = 1\) and
      \(\vc = \mP \lambda + \vsol\) for some \(\lambda \in \xmatspace{\dispRk}{1}\);

    \item \(\mA \vu = 0\) if and only if \(\vc = \mP \lambda\) for some
      \(\lambda \in \xmatspace{\dispRk}{1}\);
  \end{itemize}
  where we define
  the vector
  \[
    \vc =
    \begin{bmatrix}
      c_1 \\ \vdots \\ c_\dispRk
    \end{bmatrix}
    =
    \begin{bmatrix}
      h_1 \\ \vdots \\ h_\dispRk
    \end{bmatrix}
    u
    \rem \mpol_\vy
  \]
  in \(\xmatspace{\dispRk}{1}_{<n}\), and where \(u \in \xring_{<n}\)
  is the interpolant of \((\vy,\vu)\).
\end{lemma}

\begin{proof}
  The first steps at \zcref[range]{step:StructuredSolve-Cauchy:mpol,
    step:StructuredSolve-Cauchy:poly_formulation_g,
    step:StructuredSolve-Cauchy:poly_formulation_h,
  step:StructuredSolve-Cauchy:poly_formulation_v} construct polynomials from
  \(\mG\), \(\mH\), and \(\vv\), as defined in \zcref[S]{thm:linsys_poly:cauchy}. Altogether,
  this costs \(\bigO{\dispRk \timepm{m} \log(m) + \dispRk \timepm{n} \log(n)}\) operations in \(\field\),
  using fast polynomial operations \cite[Chap.\,10
  and\,11]{ModernComputerAlgebra}. According to
  \zcref[S]{thm:linsys_poly:cauchy}, \(\mA \vu = \vv\) is
  equivalent to
  \begin{equation}
    \label{eqn:eq_poly_form_cauchy}%
    v =
    \begin{bmatrix} g_1 & \cdots & g_\dispRk \end{bmatrix}
    \vc
    \rem \mpol_\vx,
  \end{equation}
  for \(\vc = [c_j]_j \in \xmatspace{\dispRk}{1}_{<n}\) as in the lemma. The degree bound
  comes from the fact that \(\deg(\mpol_\vy) = n\).

  According to \zcref[S]{prop:vmpade},
  \zcref[S]{step:StructuredSolve-Cauchy:vmpade,
  step:StructuredSolve-Cauchy:pivdeg} use \(\bigO{\dispRk^\expmm
  \timepm{m/\dispRk} \log(m)}\) and compute \(\mP\) in Popov form (for the
  shift \((0,\ldots,0)\)) and \(\vsol\) such that \(\rdeg(\vsol) < \delta =
  \cdeg(\mP) = \rdeg(\mP)\). The latter equalities and the fact that
  \(\rdeg(\mP)\) coincides with the diagonal degrees \(\delta\) (as computed in
  \zcref[S]{step:StructuredSolve-Cauchy:pivdeg}) follow from the definition of
  Popov forms. Note that \zcref[S]{prop:vmpade} also ensures \(\delta_1 + \cdots +
  \delta_\dispRk \le m\).
  Using \zcref[S]{lem:vmpade_generates_all}, from \zcref[S]{eqn:eq_poly_form_cauchy} we get
  the stated equivalence
  \[
    \mA \vu = \vv
    \quad \Leftrightarrow\quad
    \mu = 1 \text{ and }
    \vc
    =
    \mP \lambda + \vsol
    \text{ for some } \lambda \in \xmatspace{\dispRk}{1}
    .
  \]
  All our reasoning above also applies to finding the second equivalence,
  concerning \(\mA \vu = 0\), by considering \(\vv=0\). In this case,
  \(\mP\) is unchanged (it does not depend on \(\vv\)), and we
  have \(v = 0\), \(\mu = 1\), and \(\vsol = 0\), hence the claimed
  equivalence.
\end{proof}

\begin{lemma}
  \label{lem:cauchy_transform}%
  We follow on with notation from \zcref[S]{lem:cauchy_vmpade} and
  \zcref[S]{algo:StructuredSolve-Cauchy}.
  Define the integer tuple \(\shift = (\max(0,n-\delta_1), \ldots,
  \max(0,n-\delta_\dispRk)) \in \N^\dispRk\).
  \zcref[range]{step:StructuredSolve-Cauchy:transform,
  step:StructuredSolve-Cauchy:invert_h} of this algorithm use
  \(\bigO{\dispRk^{\expmm-1} \timepm{m+n} \log((m+n)/\dispRk)}\) operations in
  \(\field\) to compute two vectors \(\mF\) and \(\vw\) in \(\xmatspace{\dispRk}{1}_{<n}\)
  such that, for any \(\vu \in \matspace{n}{1}\),
  \begin{itemize}
    \item
    \(\mA \vu = \vv\)
    if and only if
    \(\mu = 1\) and
    \(\deg(\vsol) < n\) and
    \(\rdeg((\mF u - \vw) \rem \mpol_\vy) < \shift\);
    \item
    \(\mA \vu = 0\)
    if and only if
    \(\rdeg(\mF u \rem \mpol_\vy) < \shift\);
  \end{itemize}
  where \(u \in \xring_{<n}\) is the interpolant of \((\vy,\vu)\).
\end{lemma}

\begin{proof}
  The equivalence of \zcref[S]{lem:cauchy_vmpade} can be augmented with
  degree bounds, as follows:
  \begin{align*}
    \mA \vu = \vv
    \;
    &\Leftrightarrow\;
    \mu = 1 \text{ and }
    \vc
    =
    \mP \lambda + \vsol
    \text{ for some } \lambda \in \xmatspace{\dispRk}{1}
    \\
    &\Leftrightarrow\;
    \mu = 1 \text{ and }
    \vc
    =
    \mP \lambda + \vsol
    \text{ for some } \lambda \in \xmatspace{\dispRk}{1}
    \text{ with } \rdeg(\lambda) < \shift.
  \end{align*}
  The first equivalence is from \zcref[S]{lem:cauchy_vmpade}.
  The proof of the second equivalence is the same as for the Toeplitz case (see
  the proof of \zcref[S]{lem:toeplitz_transform}). Recall, in particular,
  that \(\vc = \mP \lambda + \vsol\) with \(\deg(\vc) < n\) implies
  \(\deg(\vsol) < n\).

  By \zcref[S]{prop:vmpade}, \(\det(\mP)\) divides \(\mpol_\vx\). On the other
  hand, since \(\vx\) and \(\vy\) are disjoint, \(\mpol_\vx\) and \(\mpol_\vy\)
  are coprime. Thus \(\det(\mP)\) is coprime with \(\mpol_\vy\), meaning that
  \(\mP\) is invertible modulo \(\mpol_\vy\).
  For the considered vectors \(\lambda\), with \(\rdeg(\lambda) < \shift\),
  both \(\lambda\) and \(\mP\lambda\) have degree less than \(n =
  \deg(\mpol_\vy)\) since \(\cdeg(\mP) = \delta\). We therefore have the
  equivalence
  \[
    \vc = \mP \lambda + \vsol
    \;\;\Leftrightarrow\;\;
    \deg(\vsol) < n \;\;\text{and}\;\; \mP^{-1} (\vc - \vsol) \rem \mpol_\vy = \lambda.
  \]
  At \zcref[S]{step:StructuredSolve-Cauchy:invert_sol, step:StructuredSolve-Cauchy:invert_h},
  one computes \(\vw = \mP^{-1} \vsol \rem \mpol_\vy\) if \(\deg(\vsol) < n\)
  (otherwise, \(\vw = 0\) is suitable since both sides of the equivalence in the
  first item of the lemma are false) and
  \(\mF = \mP^{-1} [h_1 \;\; \cdots \;\; h_\dispRk]^{\trsp} \rem \mpol_\vy\)
  using \(\bigO{\dispRk^{\expmm-1} \timepm{m+n} \log((m+n)/\dispRk)}\)
  operations in \(\field\), according to \zcref[S]{lem:Pinv_vec_mod}. The latter satisfies
  that \(\mF u \rem \mpol_\vy = \mP^{-1} \vc \rem \mpol_\vy\), by definition of
  \(\vc\).

  We have proved
  the first item in the lemma.
  For the second item, as in the proof of \zcref[S]{lem:cauchy_vmpade}, all the above arguments
  apply to the homogeneous case \(\mA \vu = 0\), by considering
  \(\vv=0\). In this case, \(\mP\) is unchanged (it does not depend on
  \(\vv\)), and we have \(\mu = 1\), \(\vsol = 0\), and \(\vw = 0\), hence the
  above equivalence simplifies into
  \(\mA \vu = 0 \Leftrightarrow \rdeg(\mF u \rem \mpol_\vy) < \shift\).
\end{proof}

\begin{lemma}
  \label{lem:cauchy_smpade}%
  We follow on with notation from
  \zcref[S]{lem:cauchy_vmpade,lem:cauchy_transform} and
  \zcref[S]{algo:StructuredSolve-Cauchy}.
  \zcref[range]{step:StructuredSolve-Cauchy:start_smpade,
  step:StructuredSolve-Cauchy:final_ret} of this algorithm use
  \(\bigO{\dispRk^{\expmm-1} \timepm{n} \log(n)}\) operations in \(\field\),
  and compute $(\ell, \vp,
  \vd, \shiftt, \vu)$ such that
  \begin{itemize}
    \item \(\vu \in \matspace{n}{1}\) is a linear system solution \(\mA \vu = \vv\)
      if one exists, otherwise \(\vu = \emptyset\);
    \item \(\ell \in \{0,\ldots,\dispRk+1\}\),
      \(\vd \in \{0,\ldots,n-1\}^\ell\),
      and \(\shiftt \in \{1,\ldots,n\}^\ell\),
      with \(\max(\vd + \shiftt) \le n\);
    \item \(\vp \in \xmatspace{1}{\ell}\) satisfies \(\cdeg(\vp) \le \vd\);
    \item the nullspace \(\{\vz \in \matspace{n}{1} \mid \mA \vz = 0\}\) of
      \(\mA\) is the set of evaluation vectors \(\vz = u(\vy)\)
      of a polynomial \(u \in \xring_{<n}\) which is in
      \(
        \{ \vp \vq \mid \vq \in \xmatspace{\ell}{1} \text{ and } \rdeg(\vq) < \shiftt\}
      \).
  \end{itemize}
\end{lemma}

\begin{proof}
  By \zcref[S]{prop:smpade}, \zcref[S]{step:StructuredSolve-Cauchy:smpade} costs
  \(\bigO{\dispRk^{\expmm-1} \timepm{n} \log(n)}\) operations in \(\field\);
  the only other arithmetic operations are in the multipoint evaluation at
  \zcref[S]{step:StructuredSolve-Cauchy:sol}, which costs \(\bigO{\timepm{n}
  \log(n)}\) \cite[Chap.\,10]{ModernComputerAlgebra}. Furthermore,
  \zcref[S]{step:StructuredSolve-Cauchy:smpade} computes \((\ell,\vp, \shiftt,
  \ssol[\vw])\) which solves \zcref[S]{pbm:smpade}.

  In particular, \(\ssol[\vw]\) is either a polynomial in \(\xring_{<n}\) such
  that \(\rdeg((\mF \ssol[\vw] - \vw) \rem \mpol_\vy) < \shift\) if one exists,
  otherwise it is \(\ssol[\vw] = \emptyset\). Thus,
  \zcref[S]{step:StructuredSolve-Cauchy:nosol} correctly sets \(\vu = \emptyset\)
  precisely when the linear system defined by \(\mA\) and \(\vv\) admits no
  solution \(\vu\), according to \zcref[S]{lem:cauchy_transform}. On the
  other hand, if \(\ssol[\vw] \neq \emptyset\), and if the other requirements
  for the existence of a solution \(\mu = 1\) and \(\deg(\vsol) < n\) also
  hold, then again according to \zcref[S]{lem:cauchy_transform},
  \zcref[S]{step:StructuredSolve-Cauchy:sol} correctly sets \(\vu\) to the
  solution \(\vu = \ssol[\vw](\vy)\) of \(\mA\vu = \vv\).

  This shows the first item in the lemma. Except for the bound on \(\vd +
  \shiftt\), the second item follows from the definition of a solution
  basis \((\ell,\vp, \shiftt)\) for \((\mpol_\vy, \mF, \shift)\) and from
  \zcref[S]{lem:smpade_solution_space}. The third item is by definition of
  \(\vd = \cdeg(\vp)\) at \zcref[S]{step:StructuredSolve-Cauchy:rev_p}.
  The third item of \zcref[S]{def:smpade_basis} implies in particular that \(\cdeg_{-n}(\vp)
  \le -\shiftt\). Since \(\cdeg_{-n}(\vp) = \cdeg(\vp) - (n,\ldots,n)
  = \vd - (n,\ldots,n)\), this gives \(\max(\vd + \shiftt) \le n\).
  It remains to show the last item, about the nullspace of \(\mA\). By the
  second item in \zcref[S]{lem:cauchy_transform}, this nullspace is the set
  of vectors \(\vz \in \matspace{n}{1}\) such that \(\vz = u(\vy)\) for some
  polynomial \(u \in
  \xring_{<n}\) such that \(\rdeg(\mF u \rem \mpol_\vy) < \shift\). On the
  other hand, \zcref[S]{lem:smpade_solution_space} states that
  \[
    \{u \in \xring_{<n} \mid \rdeg(\mF u \rem \mpol_\vy) < \shift\}
    =
    \{\vp \vq \mid \vq \in \xmatspace{\ell}{1}, \rdeg(\vq) < \shiftt\},
  \]
  which concludes the proof.
\end{proof}

\subsection{Rank and generators of the nullspace}
\label{sec:main:nullspace}

\zcref[S]{thm:algo_toeplitz,thm:algo_vandermonde,thm:algo_cauchy} describe the
nullspace of \(\mA\) through the output \((\ell,\vp,\vd,\shiftt)\) of the
algorithms. This output mainly consists of univariate polynomials, along with
integers that indicate some degree properties: this does not follow the exact
requirements of a solution to \zcref[S]{pbm:kernel}. Thus, we now show that this
description directly yields displacement generators of length \(\ell \le
\dispRk+1\) for a matrix whose columns form a basis of the nullspace; from
this, one can also deduce the rank of \(\mA\). This basis is Toeplitz-like in
the Toeplitz-like and Vandermonde-like cases, and Vandermonde-like in the
Cauchy-like case: in \zcref[S]{rmk:nullspace_structure}, we give a rough intuition
behind this and behind the fact that the length is at most \(\dispRk+1\).

In what follows, we let \(\sumti = t_1 + \cdots + t_\ell\) and \(\kappa_i = 1
+ t_1 + \cdots + t_{i-1}\) for \(1 \le i \le \ell\). We also define
\(
  \mZ = [\ve_{\kappa_1} \;\; \cdots \;\; \ve_{\kappa_\ell}] \in \matspace{\sumti}{\ell},
\)
where \(\ve_\kappa \in \matspace{\sumti}{1}\) is the \(\kappa\)th unit vector.
Note that when \(\ell = 0\), all matrices below are empty, and \(\mA\) has full
column rank \(n\).

\begin{proposition}
  \label{prop:nullspace_toeplitz}%
  Let \((\ell, \vp, \vd, \shiftt, \vu)\) be the output of
  \zcref[S]{algo:StructuredSolve-Toeplitz} (resp.\
  \zcref[S]{algo:StructuredSolve-Vandermonde}) for a Toeplitz-like (resp.\
  Vandermonde-like) matrix \(\mA\), as in \zcref[S]{thm:algo_toeplitz} (resp.\
  \zcref[S]{thm:algo_vandermonde}). For \(1 \le i
  \le \ell\), let \(a_i = x^{n-d_i-t_i} p_i\), which has degree at most \(n -
  t_i\), and let \(\va_i \in \matspace{n}{1}\) be its coefficient vector. Define
  \[
    \mK = [\L_{t_1}(\va_1) \;\; \cdots \;\; \L_{t_\ell}(\va_\ell)] \in \matspace{n}{\sumti}.
  \]
  Then, the columns of \(\mK\) form a basis of the nullspace of \(\mA\), and
  the rank of \(\mA\) is \(n - \sumti\). Furthermore,
  \(\dispSte{\Zmat_{n,0}, \Zmat_{\sumti,0}^\trsp}(\mK) = \mY \mZ^\trsp\), where
  the \(i\)th column of \(\mY \in \matspace{n}{\ell}\) is the coefficient
  vector of \(b_i = (a_i - x^{t_{i-1}} a_{i-1}) \rem x^n\), taking the
  convention \(a_0 = 0\). Given \((\ell,\vp,\vd,\shiftt)\), this matrix
  \(\mY\) can be computed using \(\bigO{\dispRk n}\) operations in \(\field\).
\end{proposition}
\begin{proof}
  Take any vector \(\lambda \in \matspace{\sumti}{1}\), and interpret it as
  being the one obtained by stacking the coefficient vectors of a list of
  polynomials \(q_i \in \xring_{<t_i}\) for \(1 \le i \le \ell\). Then, by
  construction of \(\mK\), the matrix-vector product \(\mK\lambda\) is
  precisely the coefficient vector of \(a_1 q_1 + \cdots + a_\ell q_\ell\). Thus, by
  the last item of \zcref[S]{thm:algo_toeplitz} (resp.\
  \zcref[S]{thm:algo_vandermonde}), the column space of \(\mK\) is the nullspace of
  \(\mA\).

  Furthermore, by the second item of \zcref[S]{lem:toeplitz_transform}, the map
  sending \(\vz\) to the polynomial with coefficient vector \(\vecrev{\vz}\)
  is a \(\field\)-linear bijection between the nullspace of \(\mA\) and the set
  of solutions for \((x^n,\mF,\shift)\), for which
  \zcref[S]{step:StructuredSolve-Toeplitz:smpade} computes a solution basis
  \((\ell,\bar{\vp},\shiftt)\). By \zcref[S]{lem:smpade_solution_space}, the nullity
  of \(\mA\) is thus \(\sumti\), which is the number of columns of \(\mK\).
  Hence, these columns form a basis of the nullspace, and the rank of \(\mA\) is
  \(n - \sumti\).

  The displacement formula \(\dispSte{\Zmat_{n,0}, \Zmat_{\sumti,0}^\trsp}(\mK)
  = \mK - \Zmat_{n,0}\mK\Zmat_{\sumti,0}^\trsp = \mY \mZ^\trsp\) is easily
  verified, based on the mosaic Toeplitz structure of \(\mK\).
  Furthermore, each \(b_i\) is obtained in \(\bigO{n}\) from \(p_i\) and
  \(p_{i-1}\) by shifting coefficients and one polynomial subtraction, so the
  claimed cost bound for computing \(\mY\) follows from \(\ell \le \dispRk+1\).
\end{proof}

\begin{proposition}
  \label{prop:nullspace_cauchy}%
  Let \(\mA\) be Cauchy-like, and let \((\ell, \vp, \vd, \shiftt, \vu)\) be
  the output of \zcref[S]{algo:StructuredSolve-Cauchy}, as in
  \zcref[S]{thm:algo_cauchy}. Define
  \[
    \mK = [\D(p_1(\vy))\V_{t_1}(\vy) \;\; \cdots \;\; \D(p_\ell(\vy))\V_{t_\ell}(\vy)] \in \matspace{n}{\sumti}.
  \]
  Then, the columns of \(\mK\) form a basis of the nullspace of \(\mA\), and the rank
  of \(\mA\) is \(n - \sumti\). Furthermore,
  \(\dispSte{\D(\vy), \Zmat_{\sumti,0}^\trsp}(\mK) = \mY \mZ^\trsp\), where
  \(\mY = [b_1(\vy) \;\; \cdots \;\; b_\ell(\vy)] \in \matspace{n}{\ell}\)
  with \(b_i = p_i - x^{t_{i-1}} p_{i-1}\), taking the convention \(p_0 = 0\).
  Given \((\ell,\vp,\vd,\shiftt)\), this matrix \(\mY\) can be computed using
  \(\bigO{\dispRk \timepm{n} \log(n)}\) operations in \(\field\).
\end{proposition}
\begin{proof}
  Take any vector \(\lambda \in \matspace{\sumti}{1}\), and interpret it as
  being the one obtained by stacking the coefficient vectors of a list of
  polynomials \(q_i \in \xring_{<t_i}\) for \(1 \le i \le \ell\). Then, by
  construction of \(\mK\), the matrix-vector product \(\mK\lambda\) is
  precisely the vector of evaluations \(u(\vy)\) of the polynomial \(u =
  p_1 q_1 + \cdots + p_\ell q_\ell\), which is in \(\xring_{<n}\) since
  \(\deg(p_i) \le n - t_i\). Thus, by the last item of
  \zcref[S]{thm:algo_cauchy}, the column space of \(\mK\) is the nullspace of
  \(\mA\).

  Furthermore, by the second item of \zcref[S]{lem:cauchy_transform}, the map
  sending \(\vz\) to its interpolant at \(\vy\) is a \(\field\)-linear bijection
  between the nullspace of \(\mA\) and the set of solutions for
  \((\mpol_\vy,\mF,\shift)\), for which
  \zcref[S]{step:StructuredSolve-Cauchy:smpade} computes a solution basis
  \((\ell,\vp,\shiftt)\). By \zcref[S]{lem:smpade_solution_space}, the nullity of
  \(\mA\) is thus \(\sumti\), which is the number of columns of \(\mK\). Hence
  these columns form a basis of the nullspace, and the rank of \(\mA\) is \(n -
  \sumti\).

  The displacement formula \(\dispSte{\D(\vy), \Zmat_{\sumti,0}^\trsp}(\mK)
  = \mK - \D(\vy)\mK\Zmat_{\sumti,0}^\trsp = \mY \mZ^\trsp\) is easily
  verified, based on the structure of \(\mK\), whose blocks are Vandermonde
  matrices with rows scaled by the evaluations of the \(p_i\)'s.
  Furthermore, each \(b_i\) has degree at most \(n\) and is obtained in
  \(\bigO{n}\) from \(p_i\) and \(p_{i-1}\) by shifting coefficients and one
  polynomial subtraction; then it can be evaluated at the \(n\) points of
  \(\vy\) using \(\bigO{\timepm{n}\log(n)}\) operations
  \cite[Chap.\,10]{ModernComputerAlgebra}. Hence the claimed cost bound for
  computing \(\mY\), using \(\ell \le \dispRk+1\).
\end{proof}

As a consequence, \zcref[S]{algo:StructuredSolve-Toeplitz,
algo:StructuredSolve-Vandermonde, algo:StructuredSolve-Cauchy}, followed by the
computation of \((\mY,\mZ)\) as above, solve \zcref[S]{pbm:kernel}.  The nullspace
basis \(\mK\) is described by the \(\dispOp'\)-generator
\((\mY,\mZ)\) of length \(\ell \le \dispRk+1\), where \(\dispOp' =
\dispSte{\Zmat_{n,0},\Zmat_{\sumti,0}^\trsp}\) is a Toeplitz-like operator in
the Toeplitz-like and Vandermonde-like cases, and \(\dispOp' =
\dispSte{\D(\vy),\Zmat_{\sumti,0}^\trsp}\) is a Vandermonde-like operator in
the Cauchy-like case (we use \(\dispOp'\) to distinguish this operator from
the one \(\dispOp\) used for the input \(\mA\) of the problem). Observe that
the additional cost \(\bigO{\dispRk \timepm{n}\log(n)}\) is within the bounds
of \zcref[S]{thm:algo_toeplitz, thm:algo_vandermonde, thm:algo_cauchy}. Besides,
when \(\ell > 0\), the first column of \(\mK\) is a nonzero vector in the
nullspace of \(\mA\).

\begin{remark}
  \label{rmk:nullspace_structure}%
  The fact that a generator length \(\ell \le \dispRk+1\) suffices for a
  nullspace basis of \(\mA\) reflects the property that the nullspace
  of \(\mA\) is, roughly, ``almost stable under multiplication by \(x\)'', as we can see
  from the polynomial descriptions in
  \zcref[S]{thm:algo_toeplitz,thm:algo_vandermonde,thm:algo_cauchy}. This can also
  be seen from the generator point of view: for example in the Cauchy-like case,
  \(\dispSyl{\D(\vx),\D(\vy)}(\mA) = \mG \mH^\trsp\) implies that for any
  \(\vz\) in the nullspace of \(\mA\), we have \(\mA\D(\vy)\vz = \D(\vx)\mA\vz
  - \mG\mH^\trsp\vz = -\mG\mH^\trsp\vz\). Thus \(\D(\vy)\),
  which represents multiplication by \(x\) on evaluation vectors, maps this
  nullspace into
  \[
    \{\mathsf{w} \in \matspace{n}{1} \mid \mA\mathsf{w} \text{ is in the column space of } \mG\},
  \]
  which contains the nullspace with codimension at most \(\dispRk\).

  Let us also give an intuition behind the displacement operator \(\dispOp'\)
  that is relevant for \(\mK\). Each column of \(\mK\) represents a polynomial
  of the form \(x^k a_i\), interpreted as a vector like what was done with the
  unknown \(u\) in \zcref[S]{thm:linsys_poly:toeplitz, thm:linsys_poly:vandermonde,
  thm:linsys_poly:cauchy}. This is using polynomial coefficients in the
  Toeplitz-like and Vandermonde-like cases, and evaluations at \(\vy\) in the
  Cauchy-like case. Going from one column to the next inside a given block of
  \(\mK\) is similar to multiplying by \(x\), which is (depending on the
  case) a shift of coefficients or a scaling of evaluations. Hence \(\mK\)
  has a Toeplitz-like structure in the first case and a Vandermonde-like
  structure in the second case. So, the structure of the nullspace basis only
  depends on how the vectors \(\vu\) are encoded, which means it only depends
  on the side of the displacement operator of \(\mA\) that acts on its columns.
  This explains in particular why Toeplitz-like and Vandermonde-like matrices
  lead to the same structure.
\end{remark}

\appendix

\section{Revisiting fast simultaneous M-Pad\'e approximation}
\label{sec:appendix_smpade}

In this appendix, we revisit the divide-and-conquer algorithm
\textproc{RecursiveSHPade} of \cite[Algo.\,3]{RosenkildeStorjohann2021} for
homogeneous simultaneous M-Pad\'e approximation, and we prove
\zcref[S]{thm:smpade_homogeneous}. Unlike in the latter reference, we restrict to the
situation that arises in \zcref[S]{algo:SimultaneousMPade}, where all moduli are
equal to a single polynomial \(\polmod\). Through minimal changes in the
algorithm and its analysis, this appendix achieves two results: improving the
cost bound by one logarithmic factor (which induces the factor \(\log(n)\) in
\zcref[S]{thm:main}, instead of \(\log(n)^2\)), and ensuring a weak Popov form
property for the output (simplifying the use of this output in
\zcref[S]{algo:SimultaneousMPade}). The logarithmic factor improvement in the cost
comes from the fact that, with a common modulus, the approximant basis problems
arising in the recursion involve input shifts that are somewhat balanced. This
implies that one can use the algorithms of
\cite[Thm.\,1.2]{JeannerodNeigerVillard2020} instead of the general, and more
expensive, shifted Popov approximant basis algorithm of
\cite[Thm.\,1.1]{JeannerodNeigerVillard2020}.

In \zcref[S]{sec:appendix_smpade:basecase}, we recall a folklore algorithm which performs
simultaneous M-Pad\'e approximation via a single minimal approximant basis
computation, from a matrix which involves \(\polmod\) and \(\mF\). We closely
follow the description in \cite[Sec.\,2.5.2 and
Algo.\,1]{RosenkildeStorjohann2021}, and give a refined cost analysis
exploiting the fact that all moduli are the same. This algorithm is
particularly efficient when \(\beta \in \Theta(\dispRk)\); however, our target
situation is rather \(\beta \in \Theta(1)\), as we have seen in
\zcref[S]{sec:mpade,sec:main}. Thus, we will only use it as a base case, when
\(\dispRk \le \beta\). When \(\dispRk\) is larger, we revisit the divide and
conquer approach of \cite[Sec.\,4 and Algo.\,3]{RosenkildeStorjohann2021},
which consists in performing two recursive calls on the top and bottom halves
of \(\mF\) (so, with \(\dispRk\) replaced by \(\dispRk/2\)), and combining the
result through the computation of a module intersection. The minor
modifications we bring (notably the reordering swap mentioned in
\zcref[S]{rmk:solution_specification}) allow us to show that the obtained solution
basis is in shifted weak Popov form. More importantly, by studying degree and
shift properties, we observe that we can choose \cite[Algo.\,8 and
Thm.\,1.2]{JeannerodNeigerVillard2020} as the main tool for the intersection
step and we obtain the cost bound in \zcref[S]{thm:smpade_homogeneous}, which
improves upon \cite[Thm.\,4.5]{RosenkildeStorjohann2021} by a logarithmic
factor.

Throughout this appendix, we take \((\polmod,\mF,\shift,\vd)\) as in
\zcref[S]{def:smpade_basis}, with \(\mF \in \xmatspace{\dispRk}{\beta}_{<d}\). We
write \(\tau = \max(\vd)\) and \(\bar{\shift} = (\shift,\vd)\).
We consider the \(\xring\)-module and the \(\field\)-vector space
\[
  \relmod = \left\{ \rl{\vr}{\vp} \in \xmatspace{(\dispRk+\beta)}{1} \;\Big|\; \vr = \mF \vp \bmod \polmod \right\}
  \quad\text{and}\quad
  \relsol = \left\{ \rl{\vr}{\vp} \in \relmod \mid \cdeg_{-\bar{\shift}}(\rl{\vr}{\vp}) < 0 \right\}.
\]
Note that \(\relmod\) is free of rank \(\dispRk+\beta\), with basis
\(
  [\begin{smallmatrix}
    \polmod \ident_\dispRk & \mF \\
    0 & \ident_\beta
  \end{smallmatrix}],
\)
and a vector \(\rl{\vr}{\vp}\) is in \(\relsol\) if and only if \(\vp\) is a
solution for \((\polmod,\mF,\shift,\vd)\) and \(\vr = \mF\vp \rem \polmod\).
Relating these sets to \zcref[S]{def:smpade_basis}, one can see that if a matrix
\(\bar{\mP} \in \xmatspace{(\dispRk+\beta)}{\ell}\) is
\(-\bar{\shift}\)-reduced and its columns are in \(\relsol\) and generate
\(\relsol\), then \((\ell,\mP,-\cdeg_{-\bar{\shift}}(\bar{\mP}))\) is a
solution basis for \((\polmod,\mF,\shift,\vd)\), where \(\mP\) is formed by
the last \(\beta\) rows of \(\bar{\mP}\); if moreover \(\bar{\mP}\) is in
\(-\bar{\shift}\)-weak Popov form with \(-\bar{\shift}\)-pivot index \(\pi\),
then \((\ell,\mP,-\cdeg_{-\bar{\shift}}(\bar{\mP}),\pi)\) is a weak Popov
solution basis. In particular, by \zcref[S]{lem:negative_part}, this applies to the
\(-\bar{\shift}\)-negative part of any \(-\bar{\shift}\)-weak Popov basis of
\(\relmod\).

To compute approximant bases, we use algorithms from
\cite{GiorgiJeannerodVillard2003,JeannerodNeigerVillard2020}. Given \(\mG \in
\xmatspace{\nu}{m}\) and \(\gamma \in \Z_{>0}\) and some shift, they compute a
shifted weak Popov basis of the module of approximants of \(\mG\) at order
\(\gamma\), which is \(\appmod{\gamma}{\mG} = \{\vu \in \xmatspace{m}{1}
\mid \mG \vu = 0 \bmod x^\gamma\}\).

\subsection{Base case: direct approach via a minimal approximant basis}
\label{sec:appendix_smpade:basecase}

Following \cite[Lem.\,2.1]{Neiger2016fast} and
\cite[Sec.\,2.5.2]{RosenkildeStorjohann2021}, we consider the matrix and shift
\[
  \mG = [\ident_\dispRk \;\; {-\mF} \;\; \polmod\ident_\dispRk] \in
  \xmatspace{\dispRk}{(2\dispRk+\beta)}
  \quad\text{and}\quad
  \hat{\shift} = (\bar{\shift}, \tau, \ldots, \tau) \in \Z^{2\dispRk+\beta}.
\]
As discussed above, we are seeking the \(-\bar{\shift}\)-negative part of some
\(-\bar{\shift}\)-weak Popov basis of \(\relmod\). These two references exploit
the direct relation between vectors in the kernel of \(\mG\) and vectors in
\(\relmod\) in order to compute such a basis of \(\relmod\) or specifically the
sought negative part.

\begin{lemma}
  \label{lem:hsmpade_kernel}%
  For any nonzero vector \(\vu = [\vr^\trsp \;\; \vp^\trsp \;\; *]^\trsp \in
  \xmatspace{(2\dispRk+\beta)}{1}\) in the kernel of \(\mG\), where
  \(\rl{\vr}{\vp}\) forms its first \(\dispRk+\beta\) entries, one has
  \(\cdeg_{-\hat{\shift}}(\vu) = \cdeg_{-\bar{\shift}}(\rl{\vr}{\vp})\) and the
  \(-\hat{\shift}\)-pivot index of \(\vu\) is the \(-\bar{\shift}\)-pivot index
  of \(\rl{\vr}{\vp}\). Furthermore, the vectors of negative
  \(-\hat{\shift}\)-degree in \(\appmod{\tau+d}{\mG}\) are exactly the vectors
  \([\vr^\trsp \;\; \vp^\trsp \;\; *]^\trsp\) in the kernel of \(\mG\) such
  that \(\rl{\vr}{\vp} \in \relsol\).
\end{lemma}
\begin{proof}
  Write \(\vu = [\vr^\trsp \;\; \vp^\trsp \;\; \vq^\trsp]^\trsp\). Since \(\vu
  \neq 0\) and \(\vq = (\mF \vp - \vr) / \polmod\), we have \(\rl{\vr}{\vp} \neq
  0\). The degree \(\delta = \cdeg_{-\bar{\shift}}(\rl{\vr}{\vp})\) is such
  that \(\rdeg(\vp) \le \vd + \delta \le \tau + \delta\) and \(\rdeg(\vr) \le
  \shift + \delta \le d + \delta\). Since \(\deg(\mF) < d\) and \(\tau \ge 1\),
  from the identity \(\vq = (\mF \vp - \vr) / \polmod\) we get \(\deg(\vq) <
  \tau + \delta\). The claims on \(\cdeg_{-\hat{\shift}}(\vu)\) and on the
  \(-\hat{\shift}\)-pivot index of \(\vu\) then follow from the definitions.

  For the second claim, let \(\vu = [\vr^\trsp \;\; \vp^\trsp \;\;
  \vq^\trsp]^\trsp \in \appmod{\tau+d}{\mG}\) with
  \(\cdeg_{-\hat{\shift}}(\vu) < 0\). Then, \(\rdeg(\vr) < \shift \le d\),
  \(\rdeg(\vp) < \vd \le \tau\), and \(\rdeg(\vq) < \tau\), so that
  \(\deg(\mG\vu) = \deg(\vr - \mF\vp + \polmod\vq) < \tau + d\). Thus, \(\mG\vu
  = 0 \bmod x^{\tau+d}\) ensures that \(\mG\vu = 0\). For the converse, a
  vector in the kernel of \(\mG\) is also in \(\appmod{\tau+d}{\mG}\), and from
  the previous paragraph such a vector \([\vr^\trsp \;\; \vp^\trsp \;\;
  \vq^\trsp]^\trsp\) has negative \(-\hat{\shift}\)-degree if and only if
  \(\rl{\vr}{\vp}\) (which is in \(\relmod\)) has negative
  \(-\bar{\shift}\)-degree, that is, \(\rl{\vr}{\vp}\) is in \(\relsol\).
\end{proof}

In this lemma and proof, the negative part arguments are close to those in
\cite[Sec.\,2.5.2]{RosenkildeStorjohann2021}. The first part of the lemma also
refines the degree considerations in \cite[Lem.\,2.1]{Neiger2016fast}: the
latter suggests the shift \((\bar{\shift}, d, \ldots, d)\),  whereas here we
show that \(\hat{\shift}\), involving \(\tau\) rather than \(d\), is suitable.
Using \cite[Lem.\,2.1]{Neiger2016fast} directly leads to order \(2d\) for
approximants, slightly larger than our \(\tau+d\).
This leads to \algoName{algo:BaseCaseSMPade}, whose main step is to compute an
approximant basis using algorithm \PMB{} from
\cite{GiorgiJeannerodVillard2003}. More precisely, we use its description in
\cite[Algo.\,2]{JeannerodNeigerVillard2020} which guarantees that the output
basis is in shifted weak Popov form. These references consider row
approximants; using them with our column convention is just a matter of
transposing the input and output matrices.

\begin{algorithm}[htb]
  \algoCaptionLabel{BaseCaseSMPade}{\polmod,\mF,\shift,\vd}
  \begin{algorithmic}[1]
    \Require{\((\polmod,\mF,\shift,\vd)\) as in \zcref[S]{def:smpade_basis}.}
    \Ensure{a weak Popov solution basis \((\ell,\mP,\shiftt,\pi)\) for \((\polmod,\mF,\shift,\vd)\).}
    \State \(\tau \gets \max(\vd) \in \Z_{>0}\); \;
          \(\mG \gets [\ident_\dispRk \;\; {-\mF} \;\; \polmod\ident_\dispRk] \in \xmatspace{\dispRk}{(2\dispRk+\beta)}\); \;
          \(\hat{\shift} \gets (\shift, \vd, \tau, \ldots, \tau) \in \Z^{2\dispRk+\beta}\)
    \State \(\mQ \in \xmatspace{(2\dispRk+\beta)}{(2\dispRk+\beta)} \gets\)
    a \(-\hat{\shift}\)-weak Popov basis of \(\appmod{\tau+d}{\mG}\), via \PMB{}
      \label{step:BaseCaseSMPade:pmb}
    \State \(\mQ_{\mathrm{neg}} \in \xmatspace{(2\dispRk+\beta)}{\ell} \gets\) the \(-\hat{\shift}\)-negative part of \(\mQ\)
      \label{step:BaseCaseSMPade:neg}
    \State \(\mP \in \xmatspace{\beta}{\ell} \gets\) submatrix formed by the rows \(\dispRk+1, \ldots, \dispRk+\beta\) of \(\mQ_{\mathrm{neg}}\)
    \State \(\shiftt \gets -\cdeg_{-\hat{\shift}}(\mQ_{\mathrm{neg}})\); \;
           \(\pi \gets\) the \(-\hat{\shift}\)-pivot index of \(\mQ_{\mathrm{neg}}\)
    \State \Return \((\ell, \mP, \shiftt, \pi)\)
  \end{algorithmic}
\end{algorithm}

\begin{proposition}
  \label{prop:hsmpade_basecase}%
  \algoName{algo:BaseCaseSMPade} (described in \zcref[S]{algo:BaseCaseSMPade})
  is correct and uses \(\bigO{(\dispRk+\beta)^\expmm \timepm{d} \log(d)}\)
  operations in \(\field\).
\end{proposition}
\begin{proof}
  Let \(\bar{\mP}\) be the submatrix formed by the first \(\dispRk+\beta\) rows of
  \(\mQ_{\mathrm{neg}}\) at \zcref[S]{step:BaseCaseSMPade:neg}, so that the
  output \(\mP\) consists of the last \(\beta\) rows of \(\bar{\mP}\). By
  \zcref[S]{lem:negative_part}, the columns of \(\mQ_{\mathrm{neg}}\) generate all the
  vectors of negative \(-\hat{\shift}\)-degree in \(\appmod{\tau+d}{\mG}\). By
  \zcref[S]{lem:hsmpade_kernel}, these are the vectors \([\vr^\trsp \;\;
  \vp^\trsp \;\; *]^\trsp\) in the kernel of \(\mG\) such that \(\rl{\vr}{\vp}
  \in \relsol\), thus the columns of \(\bar{\mP}\) are in \(\relsol\).
  Furthermore, every vector \(\rl{\vr}{\vp}\) in \(\relsol\) extends to a
  kernel vector which is a \(\xring\)-linear combination of the columns of
  \(\mQ_{\mathrm{neg}}\), implying that \(\rl{\vr}{\vp}\) is the same
  combination of the columns of \(\bar{\mP}\). Since \(\mQ_{\mathrm{neg}}\) is
  a subset of the columns of \(\mQ\), it is in \(-\hat{\shift}\)-weak Popov
  form, and by \zcref[S]{lem:hsmpade_kernel} the \(-\hat{\shift}\)-degree and
  \(-\hat{\shift}\)-pivot index of each column of \(\mQ_{\mathrm{neg}}\) are the
  \(-\bar{\shift}\)-degree and \(-\bar{\shift}\)-pivot index of the
  corresponding column of \(\bar{\mP}\). Therefore \(\bar{\mP}\) is in
  \(-\bar{\shift}\)-weak Popov form, with \(\cdeg_{-\bar{\shift}}(\bar{\mP}) =
  -\shiftt\) and \(-\bar{\shift}\)-pivot index \(\pi\), which concludes the
  proof of correctness (thanks to the discussion just before
  \zcref[S]{sec:appendix_smpade:basecase}).
  The cost is bounded by that of \PMB{} at \zcref[S]{step:BaseCaseSMPade:pmb},
  which uses \(\bigO{(\dispRk+\beta)^\expmm \timepm{d} \log(d)}\) operations
  according to \cite[Prop.\,3.2]{JeannerodNeigerVillard2020}.
\end{proof}

\begin{remark}
  Where we insert \(\tau\) in \(\hat{\shift}\) and where we use order \(\tau+d\),
  the algorithm from \cite[Sec.\,2.5.2]{RosenkildeStorjohann2021} rather
  inserts \(\tau-1\) and uses order \(\tau+d-1\). Increasing \(\tau-1\) towards
  \(\tau\) in the shift is what ensures the equality of pivot indices in
  \zcref[S]{lem:hsmpade_kernel}, leading to weak Popov properties rather
  than only reducedness (as in \cite[Sec.\,2.5.2]{RosenkildeStorjohann2021}).
  To preserve correctness, this increase has to be reflected as an increase
  of the order; going from \(\tau+d-1\) to \(\tau+d\) suffices.
\end{remark}

\subsection{Recursion: divide and conquer via intersections}
\label{sec:appendix_smpade:recursion}

The recursive approach of \cite[Sec.\,4]{RosenkildeStorjohann2021} works  by
splitting the \(\dispRk\) rows of \(\mF\) into two sets of  \(\dispRk_1\) and
\(\dispRk_2\) rows, with \(\dispRk_1 + \dispRk_2 = \dispRk\). Write \(\mF =
\rl{\mF_1}{\mF_2}\) with \(\mF_i \in \xmatspace{\dispRk_i}{\beta}\) and split
\(\shift = (\shift^{(1)}, \shift^{(2)})\) accordingly, thus with \(\shift^{(i)}
\in \Z^{\dispRk_i}\). We also let \(\bar{\shift}^{(i)} = (\shift^{(i)},\vd)\).

The two sub-instances \(i \in \{1,2\}\) solved recursively are \((\polmod,
\mF_i, \shift^{(i)}, \vd)\). From solutions to these, one recovers solutions
for the input instance \((\polmod,\mF,\shift,\vd)\) via intersection. Indeed, a
vector \(\vp\) is a solution for \((\polmod,\mF,\shift,\vd)\) if and only if it
is a solution for both sub-instances:
\[
  \begin{bmatrix}
    \vr_1 \\ \vr_2 \\ \vp
  \end{bmatrix}
  \in \relsol
  \;\;\Leftrightarrow\;\;
  \begin{bmatrix}
    \vr_1 \\ \vp
  \end{bmatrix}
  \in \relsol_1
  \text{ and }
  \begin{bmatrix}
    \vr_2 \\ \vp
  \end{bmatrix}
  \in \relsol_2,
\]
where \(\relsol_i\) stands for the solution space of the sub-instance \(i\).
The next lemma is essentially \cite[Thm.\,4.2]{RosenkildeStorjohann2021},
with the column-wise convention and the additional weak Popov property.

\begin{lemma}
  \label{lem:intersection}%
  For each \(i \in \{1,2\}\), let \((\ell_i, \mP_i, \shiftt^{(i)}, \pi^{(i)})\) be a
  weak Popov solution basis for the sub-instance \((\polmod, \mF_i, \shift^{(i)}, \vd)\), and let \(\bar{\mP}_i =
  \rl{\mN_i}{\mP_i} \in \xmatspace{(\dispRk_i+\beta)}{\ell_i}\) be the
  corresponding matrix, where \(\mN_i = \mF_i \mP_i \rem \polmod\). Define
  \[
    \mR =
    \begin{bmatrix}
      -\mP_1 & & \ident_\beta \\
      & -\mP_2 & \ident_\beta
    \end{bmatrix}
    \in \xmatspace{2\beta}{(\ell_1+\ell_2+\beta)}
    \qquad\text{and}\qquad
    \bar{\shiftt} = (\shiftt^{(1)}, \shiftt^{(2)}, \vd) \in \Z^{\ell_1+\ell_2+\beta},
  \]
  and let
  \[
    \mW =
    \begin{bmatrix}
      \mW_1 \\
      \mW_2 \\
      \mP
    \end{bmatrix}
    \in \xmatspace{(\ell_1 + \ell_2 + \beta)}{\ell}
    \quad \text{with }
    \mW_i \in \xmatspace{\ell_i}{\ell}
    \text{ and }
    \mP \in \xmatspace{\beta}{\ell}
  \]
  be the \(-\bar{\shiftt}\)-negative part of a \(-\bar{\shiftt}\)-weak Popov
  basis of \(\appmod{\tau}{\mR}\), with \(-\bar{\shiftt}\)-pivot index
  \((\theta_1, \ldots, \theta_\ell)\). Then \((\ell, \mP, \shiftt, \pi)\) is a
  weak Popov solution basis for \((\polmod,\mF,\shift,\vd)\), where \(\shiftt
  = -\cdeg_{-\bar{\shiftt}}(\mW)\) and, for \(1 \le k \le \ell\),
  \[
    \pi_k =
    \begin{cases}
      \pi^{(1)}_{\theta_k} & \text{if } \theta_k \le \ell_1 \text{ (with } \pi^{(1)}_{\theta_k} \le \dispRk_1\text{)} , \\
      \dispRk_1 + \pi^{(2)}_{\theta_k - \ell_1} & \text{if } \ell_1 < \theta_k \le \ell_1 + \ell_2 \text{ (with } \pi^{(2)}_{\theta_k - \ell_1} \le \dispRk_2\text{)}, \\
      \dispRk + \theta_k - \ell_1 - \ell_2 & \text{if } \theta_k > \ell_1 + \ell_2 \text{ (with } \theta_k - \ell_1 - \ell_2 \le \beta\text{)}.
    \end{cases}
  \]
  Furthermore, the corresponding matrix is
  \[
    \bar{\mP} = \begin{bmatrix} \mF\mP \rem \polmod \\ \mP \end{bmatrix}
    =
    \begin{bmatrix}
      \mN_1 \mW_1 \\
      \mN_2 \mW_2 \\ 
      \mP
    \end{bmatrix}
    \in \xmatspace{(\dispRk+\beta)}{\ell}.
  \]
\end{lemma}
\begin{proof}
  Let \(\bar{\mP}\) be the right-hand side matrix in the latter identity. Up to
  transposition and the superficial differences listed in
  \zcref[S]{rmk:solution_specification}, \cite[Thm.\,4.2 and
  Lem.\,4.3]{RosenkildeStorjohann2021} and the end of the proof of
  \cite[Thm.\,4.5]{RosenkildeStorjohann2021} show that \(\mR\mW = 0\), and that
  the columns of \(\bar{\mP}\) are in \(\relsol\) and generate \(\relsol\). In
  particular, \(\mF\mP \rem \polmod = \rl{\mN_1 \mW_1}{\mN_2 \mW_2}\).
  According to the discussion at the beginning of \zcref[S]{sec:appendix_smpade},
  it remains to show that \(\bar{\mP}\) is in \(-\bar{\shift}\)-weak Popov form
  with \(\cdeg_{-\bar{\shift}}(\bar{\mP}) = -\shiftt\) and
  \(-\bar{\shift}\)-pivot index \(\pi\).

  Let \(\mE_i = \rl{0}{\ident_\beta} \in \matspace{(\dispRk_i+\beta)}{\beta}\)
  for \(i \in \{1,2\}\) and consider
  \[
    \mK =
    \begin{bmatrix}
      -\bar{\mP}_1 & 0 & \mE_1 \\
      0 & -\bar{\mP}_2 & \mE_2 \\
      0 & 0 & \ident_\beta
    \end{bmatrix}
    \in \xmatspace{(\dispRk+3\beta)}{(\ell_1+\ell_2+\beta)}
    \;\;\;\text{and}\;\;\;
    \vh = -(\bar{\shift}^{(1)}, \bar{\shift}^{(2)}, \vd) \in \Z^{\dispRk+3\beta}.
  \]
  By construction, the matrix \(\mK\) is in \(\vh\)-weak Popov form, and
  \(\cdeg_{\vh}(\mK) = -\bar{\shiftt}\). The properties of the first two blocks
  of columns are clear from those of \(\bar{\mP}_1\) and \(\bar{\mP}_2\).
  Concerning the third block, its \(j\)th column has \(\vh\)-degree \(-d_j\)
  which is reached at the three entries equal to \(1\), and the \(\vh\)-pivot
  in this column corresponds to the bottommost such entry which is located in
  the diagonal block. From \(\mR\mW = 0\) we deduce
  \[
    \mK\mW =
    \begin{bmatrix}
      -\mN_1\mW_1 \\ \mP - \mP_1\mW_1 \\ -\mN_2\mW_2 \\ \mP - \mP_2\mW_2 \\ \mP
    \end{bmatrix}
    =
    \begin{bmatrix}
      -\mN_1\mW_1 \\ 0 \\ -\mN_2\mW_2 \\ 0 \\ \mP
    \end{bmatrix}.
  \]
  By \zcref[S]{lem:weak_popov_product}, \(\mK\mW\) is in \(\vh\)-weak
  Popov form with \(\cdeg_{\vh}(\mK\mW) = \cdeg_{-\bar{\shiftt}}(\mW) =
  -\shiftt\), and the \(\vh\)-pivot index of its \(k\)th column is the
  \(\vh\)-pivot index of the \(\theta_k\)th column of \(\mK\). Deleting the
  zero rows in \(\mK\mW\) and negating its suitable rows yields \(\bar{\mP}\),
  with the corresponding shift \(-(\shift^{(1)}, \shift^{(2)}, \vd) =
  -\bar{\shift}\). As discussed in \zcref[S]{sec:mpade:polymat}, these operations
  preserve the weak Popov form, the shifted degrees, and the pivot entries (up
  to renumbering rows). The description of \(\pi\) also follows, by tracking
  the pivot rows through the above deletion, noting that pivot entries are
  never among the deleted rows since they are nonzero.
\end{proof}

\begin{remark}
  \label{rmk:num_den_order_weak_popov}%
  Obtaining the weak Popov property in \zcref[S]{lem:intersection} motivates our
  choice, made in \zcref[S]{sec:mpade:hsimultaneous}, of placing the denominator
  rows last: \(\bar{\mP}\) has \(\mP\) at its bottom rows in
  \zcref[S]{def:smpade_basis}, and the numerator rows \(\mF\mP \rem \polmod\) are
  above it. With denominator rows first as in \cite{RosenkildeStorjohann2021},
  one only obtains a reduced solution basis, and recovering the weak Popov form
  afterwards would require the numerators, which the algorithm does not compute
  (their size may exceed the target cost).
\end{remark}

We now describe \algoName{algo:HomogeneousSMPade}. Its correctness follows by
induction from \zcref[S]{prop:hsmpade_basecase,lem:intersection}. Observe that
\zcref[S]{lem:intersection} only involves the matrices \(\mP_i\), the shifts
\(\shiftt^{(i)}\) and the \(-\bar{\shift}^{(i)}\)-pivot indices \(\pi^{(i)}\)
of the solution bases of the sub-instances, and never uses their numerator
parts \(\mN_i\), which are thus not computed. At the intersection step, for
efficiency, we use the approximant basis algorithm from
\cite[Sec.\,6]{ZhouLabahn2012}, and more specifically its refinement \SAM{}
presented and analyzed in \cite[Sec.\,7.2 and
Algo.\,8]{JeannerodNeigerVillard2020} (except for some parameters ranges where
we rather use \PMB{} to simplify the cost analysis).

\begin{algorithm}[htb]
  \algoCaptionLabel{HomogeneousSMPade}{\polmod,\mF,\shift,\vd}
  \begin{algorithmic}[1]
    \Require{\((\polmod,\mF,\shift,\vd)\) as in \zcref[S]{def:smpade_basis}.}
    \Ensure{a weak Popov solution basis \((\ell,\mP,\shiftt,\pi)\) for \((\polmod,\mF,\shift,\vd)\).}

    \State\InlineIf{\(\dispRk \le \beta\)}{\Return \Call{algo:BaseCaseSMPade}{\polmod,\mF,\shift,\vd} \Comment{see \zcref[S]{prop:hsmpade_basecase}}}
    \State \(\dispRk_1 \gets \lfloor \dispRk/2 \rfloor\);
          \(\dispRk_2 \gets \dispRk - \dispRk_1\);
          write \(\mF = \rl{\mF_1}{\mF_2}\) and \(\shift = (\shift^{(1)},\shift^{(2)})\) with \(\mF_i \in \xmatspace{\dispRk_i}{\beta}\)
          \label{step:HomogeneousSMPade:start}%
    \State \((\ell_1, \mP_1, \shiftt^{(1)}, \pi^{(1)}) \gets \Call{algo:HomogeneousSMPade}{\polmod, \mF_1, \shift^{(1)}, \vd}\)
      \label{step:HomogeneousSMPade:rec1}
    \State\InlineIf{\(\ell_1 = 0\)}{\Return the empty solution basis with \(\ell = 0\)}
    \State \((\ell_2, \mP_2, \shiftt^{(2)}, \pi^{(2)}) \gets \Call{algo:HomogeneousSMPade}{\polmod, \mF_2, \shift^{(2)}, \vd}\)
      \label{step:HomogeneousSMPade:rec2}
    \State\InlineIf{\(\ell_2 = 0\)}{\Return the empty solution basis with \(\ell = 0\)}
    \State \(\tau \gets \max(\vd)\); \;
           \(\bar{\shiftt} \gets (\shiftt^{(1)}, \shiftt^{(2)}, \vd) \in \Z^{\ell_1+\ell_2+\beta}\)
    \State \(\mR \gets
           \begin{bmatrix}
             -\mP_1 & & \ident_\beta \\
                    & -\mP_2 & \ident_\beta
           \end{bmatrix}
           \in \xmatspace{2\beta}{(\ell_1+\ell_2+\beta)}\)
    \State \(\mB \gets\) a \(-\bar{\shiftt}\)-weak Popov basis of \(\appmod{\tau}{\mR}\),
          via \PMB{} if \(\ell_1+\ell_2 \le 4\beta\) and via \SAM{} otherwise
      \label{step:HomogeneousSMPade:inter}
    \State \(\mW \in \xmatspace{(\ell_1+\ell_2+\beta)}{\ell} \gets\) the \(-\bar{\shiftt}\)-negative part of \(\mB\)
    \State \(\mP \in \xmatspace{\beta}{\ell} \gets\) submatrix formed by the last \(\beta\) rows of \(\mW\)
    \State \(\shiftt \gets -\cdeg_{-\bar{\shiftt}}(\mW)\); \;
           \((\theta_1,\ldots,\theta_\ell) \gets\) the \(-\bar{\shiftt}\)-pivot index of \(\mW\)
    \State \(\pi \in \Z^\ell \gets\) deduced from \(\pi^{(1)}\), \(\pi^{(2)}\), and \((\theta_1,\ldots,\theta_\ell)\) as in \zcref[S]{lem:intersection}
    \State \Return \((\ell, \mP, \shiftt, \pi)\)
      \label{step:HomogeneousSMPade:end}
  \end{algorithmic}
\end{algorithm}

\begin{proof}[Proof of {\zcref[S]{thm:smpade_homogeneous}}]
  Correctness has been discussed above, so we focus on the cost bound. If \(\dispRk \le
  \beta\), the algorithm directly calls \algoName{algo:BaseCaseSMPade},
  which uses \(\bigO{(\dispRk+\beta)^\expmm \timepm{d} \log(d)}\) operations in \(\field\)
  by \zcref[S]{prop:hsmpade_basecase}. This is in the target bound
  \(\bigO{(\dispRk+\beta)^{\expmm-1} \timepm{\beta d} \log(d)}\), since
  \(\dispRk+\beta \le 2\beta\) and \(\beta \timepm{d} \le \timepm{\beta d}\) by
  superlinearity.

  Now, assume \(\beta < \dispRk\). Apart from recursive calls, the only
  computation that uses arithmetic operations in
  \zcref[range]{step:HomogeneousSMPade:start, step:HomogeneousSMPade:end} is
  that of the approximant basis at \zcref[S]{step:HomogeneousSMPade:inter}.
  We are going to prove that it costs \(\bigO{\dispRk^{\expmm-1} \timepm{\beta d}
  \log(d)}\) operations in \(\field\).
  Define \(m = \ell_1+\ell_2+\beta\): the input matrix \(\mR\) has \(2\beta\)
  rows and \(m\) columns, and the order is \(\tau = \max(d_1,\ldots,d_\beta)
  \le d\). Applying \zcref[S]{lem:smpade_solution_space}, we obtain \(\ell_i
  \le \min(\dispRk_i + \beta, d_1 + \cdots + d_\beta)\) for \(i \in \{1,2\}\),
  hence \(m \le \min(\dispRk + 3\beta, 3\beta\tau)\) since \(d_1+\cdots+d_\beta
  \le \beta\tau\).

  If \(m \le 5\beta\), the approximant basis is computed by \PMB{}; by
  \cite[Prop.\,3.2]{JeannerodNeigerVillard2020}, this uses
  \(\bigO{m^\expmm \timepm{\tau} \log(\tau)}\) operations (note that \(\mR\) has
  \(2\beta \le 2m\) rows). This is in \(\bigO{\beta^{\expmm-1}
  \timepm{\beta\tau} \log(\tau)}\) by superlinearity, hence in
  \(\bigO{\dispRk^{\expmm-1} \timepm{\beta d} \log(d)}\).

  If \(m > 5\beta\), the approximant basis is computed by \SAM{}, with shift
  \(-\bar{\shiftt} = (-\shiftt^{(1)}, -\shiftt^{(2)}, -\vd)\). According to
  \zcref[S]{lem:smpade_solution_space}, the entries of \(\bar{\shiftt}\) are
  all in \(\{1,\ldots,\tau\}\), and their sum is at most
  \(3(d_1+\cdots+d_\beta) \le 3\beta\tau\). Thus, with the notation of
  \cite[Prop.\,7.4]{JeannerodNeigerVillard2020}, we have \(\sigma = 2\beta\tau\)
  and
  \[
    \zeta = \sigma + \textstyle\sum_{k=1}^m (\bar{t}_k - \min(\bar{\shiftt}))
    \le \sigma + 3\beta\tau - m < 5\beta\tau < m\tau.
  \]
  Since the order is uniform and \(\mR\) has fewer rows than columns (i.e.,
  \(2\beta < m\)), we can rely on the simplified cost bound on
  \cite[p.\,222]{JeannerodNeigerVillard2020}: \SAM{} uses \(\bigO{m^\expmm
  \timepm{\lceil\zeta/m\rceil} \log(\tau)}\) operations.
  Since \(m \le 3\beta\tau\) and \(\zeta \le 5\beta\tau\), we have
  \(\lceil\zeta/m\rceil \le 8\beta\tau/m\). By superlinearity, we deduce \(m
  \timepm{\lceil\zeta/m\rceil} \in \bigO{\timepm{\beta\tau}}\). Finally, it
  follows from \(\tau \le d\) and \(m \le \dispRk+3\beta < 4\dispRk\)
  that the call to \SAM{} costs \(\bigO{\dispRk^{\expmm-1} \timepm{\beta d}
  \log(d)}\).

  Overall, there are two recursive calls, on \(\lfloor\dispRk/2\rfloor\) and
  \(\lceil\dispRk/2\rceil\) rows, with \(\beta\) and \(\vd\) unchanged. A leaf
  of the recursion has \(\bar{\dispRk}\) rows with \(\beta/2 \le \bar{\dispRk}
  \le \beta\), and therefore also costs \(\bigO{\bar{\dispRk}^{\expmm-1}
  \timepm{\beta d} \log(d)}\). Thus, at depth \(c\) of the recursion tree,
  there are \(\le 2^c\) calls, each on \(\le 2\dispRk/2^c\) rows, whose
  total cost is therefore in \(\bigO{2^{c(2-\expmm)} \dispRk^{\expmm-1}
  \timepm{\beta d} \log(d)}\). Summing over the depths \(c\) gives the claimed
  bound, with no additional logarithmic factor since \(\expmm > 2\).
\end{proof}


\end{document}